\documentclass[12pt]{article}

\usepackage{fullpage}
\usepackage{amsfonts}
\usepackage{amsmath}
\usepackage{amsthm}
\usepackage[cmyk]{xcolor}
\usepackage{endnotes}
\usepackage{color}

\newcommand{\Astruct}{\mathbb{A}}
\newcommand{\Bstruct}{\mathbb{B}}

\newcommand{\Sstruct}{\mathbb{S}}

\newcommand{\Qbb}{\mathbb{Q}}

\newcommand{\emptyformula}{0}
\newcommand{\fullformula}{1}

\newcommand{\CSP}{\mathrm{CSP}}
\newcommand{\OPHP}{\mathrm{OPHP}}
\newcommand{\CNF}{\mathrm{CNF}}
\newcommand{\DNF}{\mathrm{DNF}}

\newcommand{\Res}{\mathrm{R}}

\newcommand{\INEQ}{\mathrm{INEQ}}
\newcommand{\EQ}{\mathrm{EQ}}
\newcommand{\E}{E}
\newcommand{\F}{F}
\newcommand{\Gg}{G}

\newtheorem{fact}{Fact}
\newtheorem{lemma}{Lemma}
\newtheorem{theorem}{Theorem}
\newtheorem{corollary}{Corollary}
\newtheorem{claim}{Claim}

\begin{document}

\title{{\bf Proof Complexity Meets Algebra}}

\author{Albert Atserias \\
Universitat Polit\`ecnica de Catalunya \\
Barcelona, Catalonia
\and
Joanna Ochremiak \\
Universit\'e Paris Diderot \\
Paris, France
}

\maketitle

\begin{abstract}
We analyse how the standard reductions between constraint satisfaction
problems affect their proof complexity. We show that, for the most
studied propositional, algebraic, and semi-algebraic proof systems,
the classical constructions of pp-interpretability, homomorphic
equivalence and addition of constants to a core preserve the proof
complexity of the CSP. As a result, for those proof systems, the
classes of constraint languages for which small unsatisfiability
certificates exist can be characterised algebraically.  We illustrate
our results by a gap theorem saying that a constraint language either
has resolution refutations of constant width, or does not have
bounded-depth Frege refutations of subexponential size. The former
holds exactly for the widely studied class of constraint languages of
bounded width. This class is also known to coincide with the class of
languages with refutations of sublinear degree in Sums-of-Squares and
Polynomial Calculus over the real-field, for which we provide
alternative proofs. We then ask for the existence of a natural proof
system with good behaviour with respect to reductions and
simultaneously small size refutations beyond bounded width. We give an
example of such a proof system by showing that bounded-degree
Lov\'asz-Schrijver satisfies both requirements. Finally, building on
the known lower bounds, we demonstrate the applicability of the method
of reducibilities and construct new explicit hard instances of the
graph 3-coloring problem for all studied proof systems.


\end{abstract}

\section{Introduction} \label{sec:introduction}

The notion of an efficient reduction lies at the heart of
computational complexity. However, in some of its subareas such as proof
complexity, even though the concept exists, it is much less
developed.
The study of the lengths of proofs has developed mostly by
studying combinatorial statements, each somewhat in isolation. There
is little theory, for instance, explaining why the best studied
families of propositional tautologies are encodings of the pigeonhole
principle or those derived from systems of linear equations over the
2-element field. Whether there is any connection between the two is an
even less explored mystery.

Luckily this fact is subject to revision, especially if proof
complexity exports its methods to the study of problems beyond
universal combinatorial statements. Consider the NP-hard optimization
problem called MAX-CUT. The objective is to find a partition of the
vertices of a given graph which maximizes the number of edges that
cross the partition. The best efficient approximation algorithm known
for this problem relies on certifying a bound on the optimum of its
semidefinite programming relaxation. Once the certificate for the
relaxation is in place, a rounding procedure gives an approximate
integral solution: at worst 87\% of the optimum in this case
\cite{GoemansWilliamson1995}.

In the example of the previous paragraph, the problem that is subject
to proof complexity analysis is that of certifying a bound on the
optimum of an arbitrary MAX-CUT instance. The celebrated Unique Games
Conjecture (UGC) can be understood as a successful approach to
explaining why current algorithms and proof complexity analyses stop
being successful where they do, and reductions play an important role
there \cite{Tulsiani2009}. One of the interesting open problems in
this area is whether the analysis of the Sums-of-Squares semidefinite
programming hierarchy of proof systems (SOS) could be used to improve
over the 87\% approximation ratio for MAX-CUT. Any improvement on this
would improve the approximation status of all problems that reduce to
it, and refute the UGC \cite{KhotKindlerMosselODonnell2007}. For the
constraint satisfaction problem, in which all constraints must be
satisfied, as well as for its optimisation version, the analogue question was resolved recently also by
exploiting the theory of reducibility: in that arena, low-degree SOS
unsatisfiability proofs exist only for problems of \emph{bounded
  width} \cite{DBLP:conf/lics/ThapperZ17,DBLP:conf/lics/DawarW17}.

The goal of this paper is to develop the standard theory of reductions
between constraint satisfaction problems in a way that it applies to
many of the proof systems from the literature, including but not
limited to Sums-of-Squares. Doing this requires a good amount of tedious work, but
at the same time has some surprises to offer that we discuss next.

Consider a constraint language $B$ given by a finite domain of values,
and relations over that domain. The instances of the constraint
satisfaction problem (CSP) over $B$ are given by a set of variables
and a set of constraints, each of which binds some tuple of the
variables to take values in one of the relations of $B$. The
literature on CSPs has focussed on three different types of conditions
that, if met by two constraint languages, give a reduction from the
CSP of one language to the CSP of the other. These conditions are a)
\emph{pp-interpretability}, b) \emph{homomorphic equivalence}, and c)
\emph{addition of constants to the core} (see
\cite{BulatovJeavonsKrokhin2005,DBLP:journals/corr/BartoOP15}).
What makes these three types of reductions important is that they
correspond to classical algebraic constructions at the level of the
\emph{algebras of polymorphisms} of the constraint languages. Indeed,
pp-interpretations correspond to taking homomorphic images,
subalgebras and powers. The other two types of reductions put together
ensure that the algebra of the constraint language is
idempotent. Thus, for any fixed algorithm, heuristic, or method
$\mathcal{M}$ for deciding the satisfiability of CSPs, if the class of
constraint languages that are solvable by $\mathcal{M}$ is closed
under these notions of reducibility, then this class admits a purely
algebraic characterization in terms of identities.

Our first result is that, for most proof systems $P$ in the
literature, each of these methods of reduction preserves the proof
complexity of the problem with respect to proofs in~$P$. Technically,
what this means is that if $B'$ is obtained from $B$ by a finite
number of constructions a), b) and c), then, for any appropriate
encoding scheme of the statement that an instance is unsatisfiable,
efficient proofs of unsatisfiability in $P$ for instances of $B$
translate into efficient proofs of unsatisfiability in $P$ for
instances of $B'$. Our results hold for a very general definition of an
appropriate encoding scheme that we call \emph{local}. The
propositional proof systems for which we prove these results include
DNF-resolution with terms of bounded size, Bounded-Depth Frege, and
(unrestricted) Frege. The algebraic and semi-algebraic proof systems
for which we prove it include Polynomial Calculus (PC) over any field,
Sherali-Adams~(SA), Lasserre/SOS, and Lov\'asz-Schrijver (LS) of bounded and
unbounded degree.
This is the object of Section~\ref{sec:ppdefinitions}.

Our second main result is an application: we obtain unconditional gap
theorems for the proof complexity of CSPs.
Building on the bounded-width theorem for CSPs \cite{BartoKozik2014,Bulatov09}, the
known correspondence between local consistency algorithms, existential
pebble games and bounded width resolution
\cite{KolaitisVardi2000,AtseriasDalmau2008}, the lower bounds for
propositional, algebraic and semi-algebraic proof systems
\cite{A88,KPW95,BIKPPW92,Ben02,DBLP:journals/tcs/Grigoriev01,BussGIP01,Chan2016},
and a modest amount of additional work to fill in the gaps, we prove
the following strong gap theorem:

\begin{theorem} \label{thm:stronggap} Let $B$ be a finite constraint
  language. Then, exactly one of the
  following holds:
\begin{enumerate} \itemsep=0pt
\item $B$ has resolution refutations of constant width,
\item $B$ has neither bounded-depth Frege refutations of
  subexponential size, nor PC over the reals,
  nor SOS refutations of sublinear degree.
\end{enumerate}
\end{theorem}

\noindent In Theorem~\ref{thm:stronggap} and below, the statement that
the constraint language $B$ has efficient proofs in proof system $P$
means that, for some and hence every local encoding scheme, all
unsatisfiable instances of $B$ have efficient refutations in $P$. 
Also, here and below, sublinear means $o(n)$, 
sublinear-exponential means $2^{o(n)}$, and 
subexponential means $2^{n^{o(1)}}$, where $n$ is the number of variables of the instance.

The proof of Theorem~\ref{thm:stronggap} actually shows that
case~\emph{1}\ happens precisely if $B$ has bounded width. As noted
earlier, the collapse of Lasserre/SOS to bounded width was already
known; here we give a different proof.  By a very recent result on the
simulation of Polynomial Calculus over the real-field by
Lasserre/SOS~\cite{DBLP:journals/eccc/Berkholz17}, the collapse of
Lasserre/SOS implies the collapse of Polynomial Calculus. The proof we
present does not depend on that. Instead we exploit directly the
theory of reducibility.

As an immediate corollary we get that
resolution is also captured by algebra, despite the fact that our
methods fall short to prove that it is closed under reductions. 

\begin{corollary} \label{cor:corollaryintro}
Let $B$ be a finite constraint language. The following are
equivalent:
\begin{enumerate} \itemsep=0pt
\item $B$ has bounded width,
\item $B$ has resolution refutations of constant width,
\item $B$ has resolution refutation of sublinear width,
\item $B$ has resolution refutations of polynomial size,
\item $B$ has resolution refutations of sublinear-exponential size,
\item $B$ has Frege refutations of bounded depth and polynomial size,
\item $B$ has Frege refutations of bounded depth and subexponential size,
\item $B$ has SA, SOS, and PC refutations over the reals of constant degree,
\item $B$ has SA, SOS, and PC refutations over the reals of sublinear degree.
\end{enumerate}
\end{corollary}
\noindent The proof of this is the object of Sections~\ref{sec:upperbound} and~\ref{sec:lowerbound}.

Section~\ref{sec:upperboundinLS} is about proof systems that operate
with polynomial inequalities and that are stronger than
Lasserre/SOS. Theorem~\ref{thm:stronggap} raises the question of
identifying a proof system that is closed under reducibilities and that can
surpass bounded width. In other words: is there a natural proof system
for which the class of languages that have efficient unsatisfiability
proofs is closed under the standard reducibility methods for CSPs, and
that at the same time has efficient unsatisfiability proofs beyond
bounded width? By the bounded-width theorem for CSPs, one way, and
indeed the only way, of surpassing bounded width is by having
efficient proofs of unsatisfiability for systems of linear equations
over some finite Abelian group.  A straightforward answer to our
question is thus the following:
Polynomial Calculus over a field of non-zero characteristic $p$ has
efficient unsatisfiability proofs for systems of linear equations over
$\mathbb{Z}_p$. On the other hand, in view of the limitations of
Polynomial Calculus over the real-field, and of certain semi-algebraic
proof systems that are imposed by Theorem~\ref{thm:stronggap}, it is
perhaps a surprise that, as we show, bounded degree Lov\'asz-Schrijver
also satisfies both requirements.

\begin{theorem}
  Unsatisfiable systems of linear equations over the 2-element group
  have LS refutations of bounded degree and polynomial size.
\end{theorem}

\noindent Proving this amounts to showing that Gaussian elimination
over $\mathbb{Z}_2$ can be simulated by reasoning with low-degree
polynomial inequalities over $\mathbb{R}$. The proof of this
counter-intuitive fact relies on earlier work in proof complexity for
reasoning about gaps of the type $(-\infty,c] \cup [c+1,+\infty)$, for
$c \in \mathbb{Z}$, through quadratic polynomial inequalities
\cite{GrigorievHirschPasechnik2002}.

It should be pointed out that another proof system that can
efficiently solve CSPs of bounded width, and that at the same time
goes beyond bounded width, is the proof system that operates with
ordered binary decision diagrams from \cite{AKV04}. Although it looks
unlikely that our methods could be used for this proof system, whether
it is closed under the standard CSP reductions is something that was
not checked, neither in \cite{AKV04}, nor here.

In Section~\ref{sec:threecoloring} we demonstrate the applicability of
our results. Consider the graph 3-coloring problem seen as the CSP of
a finite constraint language on a 3-element domain in the standard
way. Since it is known that 3-coloring has unbounded width,
Corollary~\ref{cor:corollaryintro} applies to it, and we get
3-coloring instances that are hard for all indicated proof systems. We
open the box of the method, and elaborate on that, in order to get
explicit 3-coloring instances that are hard for Polynomial Calculus
over all fields simultaneously.  This gives a new proof of the main
result in \cite{DBLP:conf/coco/LauriaN17}. Indeed, the same analysis
applies to all CSPs that are NP-complete and all proof systems that
are closed under reducibilities. This way we solve Open Problem~5.3 in
\cite{DBLP:conf/coco/LauriaN17} that asks for explicit 3-coloring
instances that are hard for Lasserre/SOS.

This article is an extended version of~\cite{DBLP:conf/icalp/AtseriasO17}. Except for providing 
full proof details, we generalise the main gap theorem to cover Polynomial Calculus over the reals
and apply our results to the 3-coloring problem, as explained in the paragraph above.

\section{Preliminaries} \label{sec:preliminaries}

\subsection{Propositional logic and propositional proofs}

\paragraph{Formulas.}  Fix a set of propositional
variables taking values \emph{true} or \emph{false}. A literal is a
variable~$X$ or the negation of a variable~$\overline{X}$. We write
propositional formulas out of literals using conjunctions $\wedge$,
disjunctions $\vee$, and parentheses, with the usual conventions on
parentheses. Also we implicitly think of $\wedge$ and $\vee$ as
commutative, associative and idempotent. Thus the formula $A \wedge A$
is viewed literally the same as $A$, the formula $A \wedge B$ is
viewed literally the same as $B \wedge A$, and the formula $(A \wedge
B) \wedge C$ is viewed literally the same as $A \wedge (B \wedge
C)$. The same applies to disjunctions. Negation is allowed only at the
level of literals, so our formulas are written in negation normal
form. If $A$ is a formula, we define its complement
$\overline{A}$ inductively: if $A$ is a variable $X$, then
$\overline{A} = \overline{X}$; if $A$ is a negated variable
$\overline{X}$, then $\overline{A} = X$; if $A$ is a conjunction $C
\wedge D$, then $\overline{A} = \overline{C} \vee \overline{D}$; if
$A$ is a disjunction $C \vee D$, then $\overline{A} = \overline{C}
\wedge \overline{D}$. The empty formula is denoted $\emptyformula$ and
is always false by convention. Its complement
$\overline{\emptyformula}$ is denoted $\fullformula$, and is always
true by convention. We think of $\emptyformula$ and $\fullformula$ as
the neutral elements of $\vee$ and $\wedge$, respectively, and the
absorbing elements of $\wedge$ and $\vee$, respectively.  Thus we view
the formulas $\emptyformula \vee A$ and $\fullformula \wedge A$ as
literally the same as $A$, and $\emptyformula \wedge A$ and
$\fullformula \vee A$ as literally the same as $\emptyformula$ and
$\fullformula$, respectively.  The size $s(A)$ of a formula $A$ is
defined inductively: if $A$ is $\emptyformula$ or $\fullformula$, then
$s(A) = 0$; if $A$ is a literal, then $s(A) = 1$; if $A$ is a
conjunction $C \wedge D$ or a disjunction $C \vee D$ with
non-absorbing and non-neutral $C$ and $D$, then $s(A) = s(C) + s(D) +
1$.

\paragraph{Propositional proof systems.}  We work with a Tait-style proof
system for propositional logic that we call Frege. The system
manipulates formulas in negation normal form and has the following
four rules of inference called \emph{axiom}, \emph{cut},
\emph{introduction of conjunction}, and \emph{weakening}:
\begin{equation}
\frac{}{A \vee \overline{A}} \;\;\;\;\;\;\;\;
\frac{C \vee A \;\;\;\; D \vee \overline{A}}{C \vee D} \;\;\;\;\;\;\;\;
\frac{C \vee A \;\;\;\; D \vee B}{C \vee D \vee (A \wedge B)} \;\;\;\;\;\;\;\;
\frac{C}{C \vee D}.
\end{equation}
In these rules, $C$ and $D$ could be the empty formula $\emptyformula$
or its complement $\fullformula$. In particular~$\fullformula$ is
an instance of an axiom rule. A Frege proof is called cut-free if it
does not use the cut rule. A Frege proof \emph{from} a set of formulas
$F$ is a proof in which the formulas in $F$ are allowed as additional
axioms. In case such a proof ends with the empty formula we call it a
Frege \emph{refutation} of $F$. As a proof system, Frege is sound and
implicationally complete, which means that if $A$ is a logical
consequence of $A_1,\ldots,A_m$, then there is a Frege proof of $A$
from $A_1,\ldots,A_m$. We will give a proof of this in
Section~\ref{sec:completenessFrege} that will apply also to certain
subsystems of Frege. If $\mathcal{C}$ is a class of formulas, a
$\mathcal{C}$-Frege proof is one that has all its formulas in the
class $\mathcal{C}$.  The size of a proof is the sum of the sizes of
the formulas in it. The length of a proof is the number of formulas in
it. 

\paragraph{Resolution, $k$-DNF Frege and Bounded Depth Frege.} A term is a conjunction of literals and a clause is a disjunction of
literals. A $k$-term or a $k$-clause is one with at most $k$ literals. A $k$-DNF is a disjunction of $k$-terms and a $k$-CNF is a conjunction
of $k$-clauses.

We define
the classes of $\Sigma_{t,k}$- and $\Pi_{t,k}$-formulas
inductively. For $t = 1$, these are just the classes of $k$-DNF and
$k$-CNF formulas, respectively. For $t \geq 2$, a formula is
$\Sigma_{t,k}$ if it is a disjunction of $\Pi_{t-1,k}$-formulas, and
it is $\Pi_{t,k}$ if it is a conjunction of $\Sigma_{t-1,k}$-formulas.

In this paper, we use the expression \emph{Frege proof of depth $t$
 and bottom fan-in $k$} to mean a $\Sigma_{t,k}$-Frege
proof. \emph{Bounded-depth Frege} means $\Sigma_{t,k}$-Frege for some
fixed $t$ and $k$. This coincides with other definitions in the literature. Frege of depth $t$ and bottom fan-in $k$, as a proof system, is sound
and implicationally complete for proving $\Sigma_{t,k}$-formulas from
$\Sigma_{t,k}$-formulas. A proof of this will follow from the general completeness
theorem below.

$\Sigma_{1,1}$ is the class of clauses. It is well-known that $\Sigma_{1,1}$-Frege and resolution proofs are basically
the same thing (the difference is that in $\Sigma_{1,1}$-Frege proofs
we allow clause axioms and weakening, but these can always be removed
at no cost). A resolution proof which uses only $l$-clauses is called a proof of \emph{width} $l$.

$\Sigma_{1,k}$-Frege, for $k \geq 2$, is the system $\Res(k)$
introduced by Krajicek \cite{Kr00}, also known as $\mathrm{Res}(k)$,
$k$-DNF resolution, and $k$-DNF Frege. This family of proof systems is
important for us because, by letting $k$ range over all constants
(i.e.,\ by considering $\Res(\mathrm{const})$), it is the weakest for
which we can prove closure under reductions.

\subsection{Completeness of Frege and its subsystems} \label{sec:completenessFrege}

The proof that Frege is implicationally complete is rather
standard. We give a detailed proof nonetheless because we want to have
concrete bounds.

\begin{theorem}[Quantitative Completeness] \label{thm:completeness}
Let $\mathcal{C}$ be a class of formulas that is closed under
subformulas and complementation, and let $\mathcal{C}'$ be the closure
of $\mathcal{C}$ under disjunctions. Let $A_1,\ldots,A_m$ and $A$ be
formulas in $\mathcal{C}'$. If $A_1,\ldots,A_m$ logically imply $A$,
then there is a $\mathcal{C}'$-Frege proof of $A$ from
$A_1,\ldots,A_m$. Moreover, if the formulas $A_1,\ldots,A_m$ and $A$
have $n$ variables and size at most $s$, then the size of the proof is
at most polynomial in $n$, $s$, $m$, $2^n$ and $s^m$.
\end{theorem}

\begin{proof}
Let $X_1,\ldots,X_n$ be the variables in $A_1,\ldots,A_m$ and $A$.
For $c \in \{0,1\}$, let $X_i^{(c)}$ denote the negative literal
$\overline{X_i}$ if $c = 0$, and the positive literal $X_i$ if $c =
1$. For a truth assignment $b = (b_1,\ldots,b_n) \in \{0,1\}^n$, let
$S_b$ be the formula $\bigvee_{i=1}^n X_i^{(1-b_i)}$. First we show
that, for each formula $B$ on the variables $X_1,\ldots,X_n$ and each
truth assignment $b = (b_1,\ldots,b_n) \in \{0,1\}^n$, if $b$
satisfies $B$, then there is a cut-free proof of $S_b \vee B$ from no
assumptions.  This is proved by induction on the size of $B$.  If $B$
is a literal, say $B = X_i$ or $B = \overline{X_i}$, then $S_b \vee B$
is obtained as the weakening of the axiom $X_i^{(1-b_i)} \vee
X_i^{(b_i)}$. If $B$ is a conjunction, say $B = C \wedge D$, then $b$
satisfies both $C$ and $D$, and by induction hypothesis there are
cut-free proofs of $S_b \vee C$ and $S_b \vee D$. A cut-free proof of
$S_b \vee B$ then follows from applying introduction of
conjunction. If $B$ is a disjunction, say $B = C \vee D$, then $b$
satisfies either $C$ or $D$, and by induction hypothesis there is a
cut-free proof of either $S_b \vee C$ or $S_b \vee D$. A cut-free
proof of $S_b \vee B$ then follows from applying weakening.  Note that
the length of the proof constructed this way is bounded by $s(B)$, and
since all the formulas in the proof have sizes bounded by $n + s(B)$,
the size of the proof is bounded by $(n + s(B))s(B)$.  Note also for
later use that, as a consequence of the assumption that $\mathcal{C}$
is closed under subformulas, the following holds: if $B$ is a
disjunction of formulas in $\mathcal{C}$, say $B = \bigvee_i B_i$,
then each formula in this proof is a disjunction of formulas in
$\mathcal{C}$, and if $B$ is a conjunction of formulas in
$\mathcal{C}$, say $B = \bigwedge_i B_i$, then the construction gives
a cut-free proof of $S_b \vee B_i$ for each $B_i$, and each formula in
the proof of $S_b \vee B_i$ is again a disjunction of formulas in
$\mathcal{C}$.

Now we assume that $A$ is a logical consequence of $A_1,\ldots,A_m$
and we build a proof of $A$ from $A_1,\ldots,A_m$. This proof will not
yet be guaranteed to have all its formulas in $\mathcal{C}'$. We will
deal with this issue later.  For each truth assignment $b \in
\{0,1\}^n$, the following hold: 1) if $b$ satisfies $A$, then the
previous paragraph gives a proof of $S_b \vee A$, and 2) if $b$
falsifies $A$, then it also falsifies some $A_j$ for some $j \in [m]$,
and the previous paragraph gives a proof of $S_b \vee
\overline{A_j}$. From these $2^n$ proofs, a sequence of $2^n - 1$ cuts
followed by one weakening gives a proof of $A \vee \overline{A_1} \vee
\cdots \vee \overline{A_m}$. From there a sequence of $m$ cuts with
the $m$ hypotheses $A_1,\ldots,A_m$ gives a proof of $A$. Finally we
argue how to turn this proof into one that uses only formulas in
$\mathcal{C}'$. For the proofs of the type $S_b \vee A$ there is no
issue because $A$ is a disjunction of formulas in $\mathcal{C}$ and
the previous paragraph argues that such proofs have all its formulas
in $\mathcal{C}'$. The problem comes from the proofs of the type $S_b
\vee \overline{A_j}$. However, since each $A_j$ is a disjunction of
formulas in $\mathcal{C}$, say $A_j = \bigvee_{k \in I_j} A_{jk}$, its
negation $\overline{A_j}$ is a conjunction of formulas in
$\mathcal{C}$, because $\mathcal{C}$ is closed under
complementation. This means that instead of using the proof of $S_b
\vee \overline{A_j}$, we could have used the proof of $S_b \vee
\overline{A_{jk}}$ for each $k \in I_j$. We do this for each choice of
$(k_1,\ldots,k_m) \in I_1 \times \cdots \times I_m$, and what we get
are proofs of $A \vee \overline{A_{1k_1}} \vee \cdots \vee
\overline{A_{mk_m}}$. These proofs now have all their formulas in
$\mathcal{C}'$. Combining these at most $s^m$ many proofs with the
hypotheses $A_1,\ldots,A_m$ in a sequence of at most $s^m$ cuts, we
get a proof of $A$ from $A_1,\ldots,A_m$, and all the formulas in this
proof are in $\mathcal{C}'$. The size is polynomial in $n$, $s$, $m$,
$2^n$ and $s^m$, and the proof is complete.
\end{proof}

The quantitative completeness theorem applies to $\Sigma_{1,k}$-Frege
($k$-DNF Frege and resolution) because if $\mathcal{C}$ is the class
of $k$-terms and $k$-clauses, then $\mathcal{C}$ is closed under
subformulas and complementation, and the closure of $\mathcal{C}$
under disjunctions is the class of $k$-DNFs. It also applies to
$\Sigma_{t,k}$-Frege, for $t \geq 2$, because the class
$\Sigma_{t-1,k} \cup \Pi_{t-1,k}$ is closed under subformulas and
complementation, and its closure under disjunctions is precisely
$\Sigma_{t,k}$.

\subsection{Polynomials and algebraic proofs} \label{sec:algproofs}

\paragraph{Polynomials.} We define everything for
the real field $\mathbb{R}$ for simplicity. For algebraic proofs the
field would not matter, but for semi-algebraic proofs we need an
ordered field such as~$\mathbb{R}$. Let $X_1,\ldots,X_n$ be $n$
algebraic commuting variables ranging over $\mathbb{R}$. We want to
define proof systems that manipulate equations of the form $P = 0$ and
inequalities of the form $P \geq 0$, where $P$ is a polynomial in
$\mathbb{R}[X_1,\ldots,X_n]$, the ring of polynomials with commuting
variables $X_1,\ldots,X_n$ and coefficients in $\mathbb{R}$. For our
purposes it will suffice to assume that the variables range over
$\{0,1\}$. Accordingly, it will also be convenient to introduce
\emph{twin} variables $\bar{X_1},\ldots,\bar{X_n}$ with the intended
meaning that $\bar{X_i} = 1 - X_i$ for $i = 1,\ldots,n$.
In all proof systems of this section, the following axioms will be
imposed on the variables:
\begin{equation}
\begin{array}{lllll}
X_i^2 - X_i = 0 & \;\;\; & 
\bar{X}_i^2 - \bar{X}_i = 0 & \;\;\; & X_i + \bar{X_i} - 1 = 0. 
\end{array}
\label{eqn:axioms}
\end{equation}
Observe that $X_i\bar{X_i} = 0$ follows from these axioms: multiply
$X_i + \bar{X_i} - 1 = 0$ by $X_i$ and subtract $X_i^2 - X_i = 0$. This
sort of reasoning is captured by the proof systems we are
about to define.

\paragraph{Algebraic and semi-algebraic proof systems.}  Let $P$ and 
$Q$ denote polynomials. In addition to the axioms
in~\eqref{eqn:axioms}, consider the following inference rules called \emph{addition} and \emph{multiplication}:
\begin{equation}
\frac{P = 0 \;\;\;\;\;\;\;\;\;\; Q = 0}{ P + Q = 0 } \;\;\;\;\;\;\;\;\;\;\;
\frac{P = 0}{PQ = 0}.
\label{eqn:algebraicrules}
\end{equation}
Clearly, these rules are sound: any assignment $f : \{ X_1,\ldots,X_n,
\bar{X}_1,\ldots,\bar{X}_n \} \rightarrow \mathbb{R}$
that satisfies the equations in the
premises, also satisfies the equation in the conclusions. For
semi-algebraic proofs we add the following axioms:
\begin{equation}
\begin{array}{lllllllll}
X_i \geq 0 & \;\;\; & \bar{X}_i \geq 0 & \;\;\; 
& 1-X_i \geq 0 & \;\;\; & 1-\bar{X}_i \geq 0 & \;\;\; &
1 \geq 0.
\end{array} \label{eqn:axiomssemi}
\end{equation}
and the following inference rules for
polynomial inequalities:
\begin{equation}
\frac{P \geq 0 \;\;\;\;\;\;\;\;\;\; Q \geq 0}{ P + Q \geq 0 } \;\;\;\;\;\;\;\;\;\;\;
\frac{P \geq 0 \;\;\;\;\;\;\;\;\;\; Q \geq 0}{ PQ \geq 0 } \;\;\;\;\;\;\;\;\;\;\;
\frac{}{P^2 \geq 0}.
\label{eqn:semialgebraicrules}
\end{equation}
These rules are called \emph{addition}, \emph{multiplication} and
\emph{positivity of squares} and are also sound for assignments $f :
\{X_1,\ldots,X_n,\bar{X}_1,\ldots,\bar{X}_n\} \rightarrow
\mathbb{R}$. One could also consider additional rules that link
equalities with inequalities, such as deriving $P \geq 0$ from $P =
0$, or deriving $P = 0$ from $P \geq 0$ and $-P \geq 0$, but if we
think of an equality as two inequalities, then they are not strictly
necessary.  On the other hand, some of the axioms are redundant, such
as $1 \geq 0$ which can be obtained from adding $X_i \geq 0$ and
$1-X_i\geq 0$, but for the sake of clarity in writing proofs we prefer
to keep them.

If $H$ denotes a system of polynomial equations $P_1 = 0,\ldots,P_r =
0$ and $P = 0$ is a further equation, an \emph{algebraic proof} of $P
= 0$ from $H$ is a sequence of polynomial equations ending with $P =
0$ where each equation in the proof is either a hypothesis equation
from $H$, or an axiom equation as in~\eqref{eqn:axioms}, or follows
from previous equations in the sequence by one of the inference rules
in \eqref{eqn:algebraicrules}. If $H$ in addition includes a system of
polynomial inequalities $Q_1 \geq 0,\ldots,Q_s \geq 0$, then a
\emph{semi-algebraic proof} of $Q \geq 0$ from $H$ is defined
analogously except that we think of each equation as two inequalities,
we use additionally the axioms in~\eqref{eqn:axiomssemi}, and we use
additionally the rules in~\eqref{eqn:semialgebraicrules}. Note that by
writing $Q = Q^+ - Q^-$, where $Q^+$ and $Q^-$ have only positive
coefficients, the rules in~\eqref{eqn:algebraicrules} are actually
easily simulated by the rules in~\eqref{eqn:semialgebraicrules} (for
the multiplication rule, this uses also the axioms
in~\eqref{eqn:axiomssemi}). If an algebraic proof ends with the
equation $1 = 0$, or similarly if a semi-algebraic proof ends with the
inequality $-1 \geq 0$, we call it a \emph{refutation} of $H$.

As proof systems for deriving new polynomial equations or
inequalities that follow from old ones on all evaluations of their
variables in $\{0,1\}$, both systems are sound and implicationally
complete (we note, however, that without some restrictions on the
domain of evaluation, such as $\{0,1\}$ in our case, the completeness
claim is not true). In Section~\ref{sec:completenessSA} below we will
prove implicational completeness for two subsystems of algebraic and
semi-algebraic proofs, and hence for algebraic and semi-algebraic
proofs themselves.

The main complexity measures for algebraic and semi-algebraic proofs
are size and degree. Size is measured by the number of symbols it
takes to write the representations of the polynomials in the proofs,
and degree is the maximum of the total degrees of the polynomials in
the proofs. Polynomials are typically represented as explicit sums of
monomials, or as algebraic formulas or circuits. Using formulas or
circuits as representations requires some additional technicalities in
the definitions of the rules, that we want to avoid (see
\cite{Pit97,GrigorievHirsch2003}). For all our examples below, we use the
representation of an explicit sum of monomials.

\paragraph{Some proof systems from the literature.}  The proofs in the
\emph{Polynomial Calculus} (PC) are algebraic proofs restricted in
such a way that the polynomial $Q$ in the multiplication rule
in~\eqref{eqn:algebraicrules} is either a scalar or a variable
\cite{CEI95}. In the literature, this has been called PCR for \emph{PC
  with resolution} (see \cite{ABRW00b}), due to the presence of twin
variables, but in recent works the shorter original name PC is used.
As pointed out earlier, algebraic proofs can be defined over arbitrary
scalar-fields $F$ beyond the real-field $\mathbb{R}$. A claim about algebraic proofs in which the field is
omitted is meant to hold for all fields simultaneously.
Whenever we need
to specify the field~$F$, we speak of algebraic and PC proofs
\emph{over} $F$. 

The proofs in the \emph{Lov\'asz-Schrijver} (LS) proof system are
semi-algebraic proofs for which the following restrictions apply: 1)
the polynomial $Q$ in the multiplication rule
in~\eqref{eqn:semialgebraicrules} is either a positive scalar or a
variable, and 2) the positivity-of-squares rule
in~\eqref{eqn:semialgebraicrules} is not allowed. When the
positivity-of-squares is also allowed, the system is called
\emph{Positive Semidefinite Lov\'asz-Schrijver} and is denoted
LS$^+$. Originally the Lov\'asz-Schrijver proof system was defined to
manipulate quadratic polynomials only (see
\cite{LovaszSchrijver1991,Pudlak1998}). We follow
\cite{GrigorievHirschPasechnik2002} and consider the extension to
arbitrary degree. For the original
Lov\'asz-Schrijver proof systems we use LS$_2$ and LS$_2^+$.
Degree-$d$ Lov\'asz-Schrijver and degree-$d$ Positive Semidefinite Lov\'asz-Schrijver are denoted LS$_d$ and LS$_d^+$,
respectively. For LS and LS$^+$ proofs, an important complexity
measure originally studied by Lov\'asz and Schrijver is their
\emph{rank}, which is the maximum nesting depth of multiplication by a
variable in the proof. Note that, due to possible cancellations, the
degree of an LS proof could in principle be much smaller than its
rank.

We define four additional proof systems called \emph{Nullstellensatz}
(NS), \emph{Sherali-Adams} (SA), \emph{Positive Semidefinite
  Sherali-Adams} (SA$^+$), and \emph{Lasserre/Sums-of-Squares}
(SOS). For NS, SA and SA$^+$, we define them as the subsystems of PC,
LS and LS$^+$, respectively, in which all applications of the
multiplication rule must \emph{precede} all applications of the
addition rule. Due to the structural restriction in which
multiplications precede additions, we can think of a proof from a set
$H$ of hypotheses as a \emph{static} polynomial identity of the form
\begin{equation}
\sum_{i=1}^r P_i \cdot c_i \prod_{j \in J_i} X_j \prod_{k \in K_i} \bar{X}_k = P,
\label{eqn:identity}
\end{equation}
where $P_1,\ldots,P_r$ are polynomials that either come from the set
$H$ of hypotheses, or they are axiom polynomials from the lists
in~\eqref{eqn:axioms} and~\eqref{eqn:axiomssemi} as appropriate (i.e.,
from~\eqref{eqn:axioms} for NS, and from both~\eqref{eqn:axioms}
and~\eqref{eqn:axiomssemi} for SA and SA$^+$), or are squares of
polynomials when they are allowed (i.e., for SA$^+$), and
$c_1,\ldots,c_r$ are scalars of the appropriate type (i.e., arbitrary
when the $P_i$ they multiply comes from an equation, or positive when
the $P_i$ they multiply comes from an inequality). Finally we define
Lasserre/Sums-of-Squares proof system as the subsystem of
semi-algebraic proofs to which the following restrictions apply: 1) the
polynomial $Q$ is arbitrary in the multiplication rule
in~\eqref{eqn:algebraicrules} and it is a square polynomial in the
multiplication rule in~\eqref{eqn:semialgebraicrules}, and 2) all
multiplications precede all additions. Thus, in terms of static
identities, these are proofs of the form
\begin{equation}
\sum_{i=1}^r P_i \cdot S_i = P, \label{eqn:sosproof}
\end{equation}
where $P_1,\ldots,P_r$ are polynomials that either come from the set
$H$ of hypotheses, or they are axiom polynomials from the
lists~\eqref{eqn:axioms} and~\eqref{eqn:axiomssemi}, or they are
squares, and $S_1,\ldots,S_r$ are arbitrary polynomials or square
polynomials as appropriate (i.e., arbitrary if the $P_i$ they multiply
comes from an equation, and squares if the $P_i$ they multiply comes
from an inequality).  Note that the size of an NS, SA, SA$^+$ or SOS
proof is polynomially related to the sum of the sizes of the non-zero
$c_i$'s and $S_i$'s in the corresponding static
identities~\eqref{eqn:identity} and~\eqref{eqn:sosproof}. Non-static
proofs are sometimes called \emph{dynamic}
\cite{GrigorievHirschPasechnik2002}. We will avoid using this term
here.

We close this section by noting the relationships between these proof
systems.  Clearly, every NS proof of degree $d$ is also a PC proof of
degree $d$. The converse is certainly not true, but what is true is
that every PC proof of degree $d$ and rank $k$ can be converted into
an NS proof of degree $d+k$, where the rank of a PC proof is the
analogue of the rank measure for LS proofs that we defined
earlier. The same relationships hold between SA and LS, and SA$^+$ and
LS$^+$. In all three cases, the conversions go by swapping the order
in which the addition and the multiplication rules are applied, when
they appear in the \emph{wrong} order.  Also, every NS proof over the
reals is an SA proof, which is an SA$^+$ proof. Finally, thanks to the
axioms~\eqref{eqn:axioms}, each SA$^+$ proof can be easily converted to an
SOS proof of twice the degree: replace each multiplication by a
variable $X$ by a multiplication by $X^2$, and subtract the
appropriate multiple of the axiom $X^2 - X = 0$ to effectively
simulate the multiplication by~$X$.  See \cite{Laurent2001} for a
related discussion.

\paragraph{Discussion on variants of NS, SA, SA$^+$ and SOS.}
The polynomial identity interpretations of NS, SA, SA$^+$ and SOS,
c.f.,~\eqref{eqn:identity} and~\eqref{eqn:sosproof}, are closely
related to the original definitions by Beame et
al. \cite{BIKPP96} for NS, and the
settings of Sherali and Adams \cite{SheraliAdams1990} and Lasserre
\cite{Lasserre2011} for SA and SOS, respectively. In most incarnations
of these proof systems the twin variables are not present; in some
others they are (e.g.,~\cite{AtseriasLN16}). If we care only about
degree, the presence of twin variables makes no difference at all for
Nullstellensatz since we can always simulate a multiplication by $\bar{X}_i$ by
subtracting a multiplication by $X_i$. Note, however, that this blows
up the size exponentially in the degree. In order to make sense of Sherali-Adams
without twin variables, we need to extend the definition to allow $Q$
in the multiplication rule to be, besides a positive scalar or a
variable~$X_i$, a linear polynomial of the form $1-X_i$.  The static
form of such a proof is an identity such as
\begin{equation}
\sum_{i=1}^r P'_i \cdot c_i \prod_{j \in J_i} X_j \prod_{k \in K_i}
(1-X_k) = P',
\label{eqn:identitywithoutbar}
\end{equation}
where $P'_1,\ldots,P'_r$ and $P'$ are polynomials as
in~\eqref{eqn:identity}, but without twin variables.  If
$P'_1,\ldots,P'_r$ and $P'$ denote the polynomials over
$X_1,\ldots,X_n$ that result from the polynomials $P_1,\ldots,P_r$ and
$P$ over $X_1,\ldots,X_n,\bar{X}_1,\ldots,\bar{X}_n$ when each twin
variable $\bar{X}_i$ is replaced by $1-X_i$, then any valid proof with
twin variables as in~\eqref{eqn:identity} transforms into a valid
proof without twin variables as
in~\eqref{eqn:identitywithoutbar}. Thus, if we care only about degree,
the versions of Sherali-Adams and Positive Semidefinite
  Sherali-Adams without twin variables simulate the
versions with twin variables, for polynomials without twin
variables. As for Nullstellensatz the size could blow up exponentially in the
degree. The same facts are true for Sums-of-Squares.

Two further comments are in order. For Nullstellensatz, one could consider an
alternative definition in which proofs are polynomial identities of
the form $\sum_i P_i \cdot R_i = P$, where the $P_i$ are hypotheses or
axiom polynomials, and the $R_i$ are arbitrary polynomials. However
this difference is minor since we can always write each $R_i$ as a
combination of monomials $\sum_j c_{ij} M_{ij}$ and split $P_i \cdot
R_i$ into $\sum_j P_i \cdot c_{ij} M_{ij}$. Second, one could consider
the version of Sums-of-Squares in which in addition to squares $S_i$ as
in~\eqref{eqn:sosproof}, one is also allowed multiplication by
variables. As noted earlier, such multiplications by a variable $X$
can be simulated by multiplications by their squares $X^2$, thanks to
the axioms $X^2 - X = 0$ from~\eqref{eqn:axioms}, at the cost of at
most doubling the degree, and blowing up the size at most
polynomially.

\subsection{Completeness of Nullstellensatz and Sherali-Adams} \label{sec:completenessSA}

In this section we prove the implicational completeness of Nullstellensatz and Sherali-Adams
with quantitative bounds. We start with two technical lemmas that
will be used to justify the elimination of twin variables.

\begin{lemma} \label{lem:technicalone} For every polynomial $P$ of
  degree $d$ and every variable $Y$, there are NS and SA proofs of $P
  \cdot (1 - Y - \bar{Y}) = 0$ and $P \cdot (Y^2 - Y) = 0$ of degree
  $d+1$ and $d+2$, respectively, and size polynomial in the size
  of~$P$.
\end{lemma}
 
\begin{proof}
  Split $P$ into a sum of monomials $\sum_j c_j M_j$, lift the axiom
  $1-Y-\bar{Y} = 0$ by $c_j M_j$, and add up together to get $P \cdot
  (1-Y-\bar{Y}) = 0$. 
\end{proof}

\noindent The second technical lemma that we need formalizes the elimination
of twin variables.

\begin{lemma} \label{lem:simulationmultiplication}
  For every polynomial $P$ of degree $d$, every scalar $c$ and
  every two subsets $J$ and $K$ of $[n]$, with $|J|+|K| = \ell$, there
  are NS and SA proofs of the equation
  \begin{equation}
  P \cdot c \prod_{j \in J} X_j \prod_{k \in K} (1-X_k) - 
  P \cdot c \prod_{j \in J} X_j \prod_{k \in K} \bar{X}_k = 0
  \end{equation} 
  of degree $d+\ell$ and size polynomial in $2^\ell$ and the size of
  $c$ and $P$.
\end{lemma}

\begin{proof}
  Assume without loss of generality that $K = [t]$ where $t \leq
  \ell$. Let $Q = c P \prod_{j \in J} X_j$. Define $R_j =
  Q\prod_{k=1}^j (1-X_k)\prod_{k=j+1}^t \bar{X}_k$ for all $j \in [t]
  \cup \{0\}$. Observe that the goal equation is $R_t - R_0 = 0$. For
  each $j \in [t]$, let $T_j = Q \prod_{k=1}^{j-1} (1-X_k)
  \prod_{k=j+1}^t \bar{X}_k$. Now:
  \begin{equation}
  (1-X_j-\bar{X}_j) T_j = R_j - R_{j-1} \label{eqn:telescope}
  \end{equation}
  for each $j \in [t]$. Lemma~\ref{lem:technicalone} gives proofs of
  $(1-X_j-\bar{X}_j) T_j = 0$ for every $j \in [t]$. Adding them all
  together gives $R_t - R_0 = 0$ by~\eqref{eqn:telescope} and we are
  done.
\end{proof}

We will need the following
definitions.  For every assignment $a : \{X_1,\ldots,X_n\} \rightarrow
\{0,1\}$, define
\begin{align*}
J_a & = \{ i \in [n] : a(X_i) = 1 \}, \\
K_a & = \{ i \in [n] : a(X_i) = 0 \}.
\end{align*}
Define its \emph{indicator polynomial}:
\begin{equation}
I_a(X_1,\ldots,X_n) := \prod_{j \in J_a} X_j \prod_{k \in K_a} (1-X_k).
\end{equation}
For every polynomial $P$, let $P(a)$ denote the evaluation of $P$ when
$X_i$ is assigned $a(X_i)$.

For a polynomial $P$ on the variables $X_1,\ldots,X_n$, its
\emph{multilinearization} is the unique multilinear polynomial that
agrees with $P$ on all assignments of values in $\{0,1\}$ to its
variables. The uniqueness of the multilinearization follows from the
fact that the collection of multilinear polynomials in
$\mathbb{R}[X_1,\ldots,X_n]$ forms a vector space of dimension $2^n$
for which the monomials make a basis. Note that this holds for any
field; not just $\mathbb{R}$.

\begin{lemma} \label{lem:multilinearization}
  For every polynomial $P$ on the variables $X_1,\ldots,X_n$, there
  are polynomials $Q_1,\ldots,Q_n$ such that the following identity
  holds:
  \begin{equation}
  P + \sum_{i=1}^n Q_i \cdot (X_i^2 - X_i)
  = P^*,
  \label{eqn:multi}
  \end{equation}
  where $P^*$ denotes the multilinearization of $P$.  Moreover,
  each $Q_i$ has size polynomial in the size~of~$P$.
\end{lemma}

\begin{proof}
  Observe that it is enough to prove the lemma for the special case of monomials.
  Indeed, if $P$ is an arbitrary polynomial, we get the identity~\eqref{eqn:multi}
  by splitting $P$ into a sum of monomials, applying the lemma to each monomial, and 
  adding up the obtained identities.
  
  Let $P$ be a monomial.
  We proceed by induction on the sum of the individual degrees of the variables.
  If all variables have individual degree one, there is nothing to
  prove. Otherwise, some variable must have individual degree at least
  two. Say this variable is $X_j$ and let $P'$ and $P''$ be such that
  $P = X_j P'$ and $P' = X_j P''$. Note that the multilinearizations
  of $P$ and $P'$ are the same, and in both $P'$ and $P''$ the
  sum of the individual degrees is strictly smaller.  The induction
  hypothesis applied to $P'$ gives polynomials $Q'_1,\ldots,Q'_n$ such
  that
  \begin{equation}
  P' + \sum_{i=1}^n Q'_i \cdot (X_i^2 - X_i) = {P'}^* = P^*.
  \label{eqn:already}
  \end{equation}
  Now the identity we want is obtained by defining $Q_i = Q'_i$
  for $i \not= j$, and $Q_j = Q'_j - P''$.
  Indeed:
  \begin{align}
  P + \sum_{i} Q_i \cdot (X_i^2 - X_i) & =
  P + (Q'_j - P'') \cdot (X_j^2 - X_j) + \sum_{i\not=j} Q'_i \cdot (X_i^2 - X_i)
  \\
  & = P + \sum_{i} Q'_i \cdot (X_i^2 - X_i)
  - P'' \cdot X_j^2 + P'' \cdot X_j = \\
  & = P' + \sum_{i} Q'_i \cdot (X_i^2 - X_i),
  \end{align}
  and we already proved in~\eqref{eqn:already} that this last thing is
  $P^*$.
\end{proof}

\begin{theorem}\label{thm:completeness_alg}
  Let $H$ be a system of polynomial equations, let $I$ be a system of
  polynomial inequalities, and let $P$ be a polynomial, all over the
  same $n$ variables.  If $P = 0$ follows from $H$ on all evaluations
  of its variables in $\{0,1\}$, then there is an NS proof of $P = 0$
  from $H$. Similarly, if $P \geq 0$ follows from $H \cup I$ on all
  evaluations of its variables in $\{0,1\}$, then there is an SA proof
  of $P \geq 0$ from $H \cup I$. Moreover, in both cases the degree of
  the proof is at most $n+1$, and the size is polynomial in $2^n$ and
  in the size of $H$ and $H \cup I$, respectively.
\end{theorem}

\begin{proof}
  Both proofs are essentially the same; first we give the proof for Sherali-Adams
  and then indicate how to adapt it to Nullstellensatz.  We prove the theorem when
  $P$ is multilinear and then we adapt it to the general case.  Assume
  $P$ is multilinear and let $H \cup I = \{ P_1 \geq 0,\ldots,P_m \geq
  0\}$, where we have written each equation in $H$ as two
  inequalities.  For every assignment $a : \{X_1,\ldots,X_n\}
  \rightarrow \{0,1\}$, let $c_{a,0},c_{a,1},\ldots,c_{a,m}$ be the
  real numbers defined by cases as follows. If $P(a) \geq 0$, let
  $c_{a,i} = P(a)$ for $i = 0$ and $c_{a,i} = 0$ for $i \in [m]$. If
  $P(a) < 0$, let $i^*$ be the smallest element in $[m]$ such that
  $P_{i^*}(a) < 0$, which must exist by the hypothesis, and define
  $c_{a,i} = P(a)/P_{i}(a)$ for $i = i^*$ and $c_{a,i} = 0$ for each
  $i \in ([m]\cup\{0\})\setminus\{i^*\}$. Observe that in all cases
  $c_{a,i}$ is non-negative. In the first case because $P(a)$ was
  non-negative, and in the second case because both $P_{i^*}(a)$ and
  $P(a)$ were negative, so their ratio is positive.  The choice of
  these reals guarantees that
  \begin{equation}
  c_{a,0} + \sum_{i=1}^m c_{a,i} P_i(a) = P(a). \label{eqn:weights}
  \end{equation}
  We need the following claim.

  \begin{claim} \label{cl:cl}
    For every assignment $a$ and every $i \in [m]$, the polynomial
    $P_i(a) \cdot I_a$ is the multi\-linearization of $P_i \cdot I_a$.
    In addition, $\sum_a \big( c_{a,0} \cdot I_a + \sum_{i=1}^m c_{a,i}
    \cdot P_i(a) \cdot I_a \big) = P$.
  \end{claim}

  \begin{proof}
    Since the multilinearization is unique and the polynomial $P_i(a)
    \cdot I_a$ is multilinear, it suffices to show that $P_i(a) \cdot
    I_a$ and $P_i \cdot I_a$ agree on all assignments of values in
    $\{0,1\}$ to their variables. But this is easy: they both evaluate
    to $P_i(a)$, or both evaluate to $0$, depending on whether the
    assignment is $a$, or different from $a$, respectively.  For the
    second claim we use the same argument, and add the additional fact
    that $P$ is itself multilinear: the big sum over $a$ is a
    multilinear polynomial and, by~\eqref{eqn:weights}, it agrees with
    $P$ on all assignments of values in $\{0,1\}$ to its
    variables. Hence, by the uniqueness of the multilinearization, and
    since $P$ is multilinear, it is $P$ itself.
  \end{proof}

  \noindent Back to the proof, by the first part of Claim~\ref{cl:cl},
  for every assignment $a$ and every $i \in [m]$, there exist
  polynomials $Q_{a,i}^1,\ldots,Q_{a,i}^n$ according to
  Lemma~\ref{lem:multilinearization} that make the following identities
  hold:
  \begin{equation}
    P_i \cdot I_a + \sum_{j=1}^n Q_{a,i}^j \cdot (X_j^2 - X_j) = 
    P_i(a) \cdot I_a.
  \label{eqn:summand}
  \end{equation}
  We are ready to build up the proof of $P \geq 0$ from $P_1 \geq
  0,\ldots,P_m \geq 0$. We claim that the following identity holds:
  \begin{equation}
   \sum_{a} \Big({c_{a,0} \cdot I_a + 
    \sum_{i=1}^m c_{a,i} \cdot \Big({ P_i \cdot I_a + 
    \sum_{j=1}^n Q^j_{a,i} \cdot (X_j^2 - X_j)}\Big)}\Big) = P.
  \label{eqn:identitytwo}
  \end{equation}
  First we claim that the left-hand side can be converted into a valid
  SA proof (with multiplications by $X_j$'s and $1-X_j$'s, which can
  be simulated in our definition of Sherali-Adams as discussed in
  Lemma~\ref{lem:simulationmultiplication}). To see this, just reorder
  the terms and apply Lemma~\ref{lem:technicalone} to replace
  $Q^j_{a,i} \cdot (X_j^2 - X_j)$ by proper SA proofs. It remains to
  see that the identity~\eqref{eqn:identitytwo} holds; this will show that
  it is an SA proof of $P \geq 0$ from $P_1 \geq 0,\ldots,P_m \geq 0$.

  In order to see that~\eqref{eqn:identitytwo} holds, first use
  equation~\eqref{eqn:summand} to rewrite its left-hand side:
  \begin{equation}
    \sum_a \Big( {c_{a,0} \cdot I_a + \sum_{i=1}^m c_{a,i} \cdot P_i(a) \cdot I_a}\Big). 
  \end{equation}
  And now use the second part of Claim~\ref{cl:cl} to complete the
  proof when $P$ is multilinear.

  When $P$ is not multilinear, it suffices to apply the above argument
  to get its multilinearization $P^*$, and then apply the
  \emph{reverse} identity in
  Lemma~\ref{lem:multilinearization}. Indeed,
  \begin{equation}
  P^* - \sum_{i=1}^n Q_i \cdot (X_i^2 - X_i) = P.
  \end{equation}
  To turn this into a proper SA proof we need to use
  Lemma~\ref{lem:technicalone} again.

  For Nullstellensatz, the argument is the same except that, in order to handle
  arbitrary fields besides the real field $\mathbb{R}$, the
  coefficients $c_{a,i}$ need to be redefined. Let $H = \{P_1 =
  0,\ldots,P_m = 0\}$. If $P_i(a) = 0$, define $c_{i,a} = 0$ for all
  $i \in [m]$. If $P_i(a) \not= 0$, let $i^*$ be the smallest element
  in $[m]$ such that $P_{i^*}(a) \not= 0$, which must exist by
  hypothesis, and define $c_{a,i} = P(a)/P_i(a)$ for $i = i^*$ and
  $c_{a,i} = 0$ for $i \in ([m] \cup \{0\}) \setminus \{i^*\}$. This
  choice is well-defined over any field and
  guarantees~\eqref{eqn:weights}. The rest of the proof is the same.
\end{proof}

\subsection{Constraint satisfaction problem}\label{sec:csp}

There are many equivalent definitions of the constraint satisfaction problem. Here we use the definition in terms of homomorphisms. Below we introduce the necessary terminology. A concrete example will be developed in Section~\ref{sec:threecoloring} where we apply the method of reducibilities to the graph $k$-coloring problem for $k \geq 3$.

\paragraph{CSPs and homomorphisms.}
A \emph{relational vocabulary} $L$ is a set of symbols; each symbol has an associated natural number called its \emph{arity}. A \emph{relational structure} $\Bstruct$ over $L$ (or an $L$-\emph{structure}) is a set $B$, called a \emph{domain} together with a set of relations over~$B$. For each natural number $r$ and each relation symbol $R \in L$ of arity $r$, there is a relation in $\Bstruct$ of arity $r$ denoted $R(\Bstruct)$, i.e., $R(\Bstruct) \subseteq B^r$. Sometimes we call it an \emph{interpretation} of~$R$ in~$\Bstruct$. We say that a relational structure is finite if its domain is finite and it has finitely many non-empty relations. 

Let $\Bstruct$ and $\Bstruct'$ be $L$-structures, for some relational vocabulary $L$. A \emph{homomorphism} from $\Bstruct$ to $\Bstruct'$ is a function $h \colon B \rightarrow B'$, which preserves all the relations, that is, for every natural number $r$ and each relation symbol $R \in L$ of arity $r$, if $(b_1, \ldots, b_r) \in R(\Bstruct)$, then $(h(b_1), \ldots, h(b_r)) \in R(\Bstruct')$.

For a fixed $L$-structure $\Bstruct$, the constraint satisfaction
problem of $\Bstruct$, denoted CSP($\Bstruct$), is the following
computational problem: given a finite $L$-structure $\Astruct$, decide
whether there exists a homomorphism from $\Astruct$ to $\Bstruct$. If
the anwser is positive we call the instance $\Astruct$
\emph{satisfiable}; otherwise we call it \emph{unsatisfiable}. The
\emph{size} of an instance $\Astruct$ is the number of elements in its
domain plus the number of tuples in all its relations. Note that if
the vocabulary $L$ is fixed and finite, then the size of $\Astruct$ is
polynomial in the number of elements of its domain which we denote by
$|A|$.  In the context of CSP the structure $\Bstruct$ is often called
a \emph{constraint language} or a \emph{template}. We usually assume
that the constraint language $\Bstruct$ is finite.

\paragraph{Bounded-width.}

The existential $k$-pebble game is played on two relational structures
$\Astruct$ and $\Bstruct$ over the same vocabulary by two players
called Spoiler and Duplicator. The players are given two corresponding
sets of pebbles $\{a_1, \ldots, a_k \}$ and $\{b_1, \ldots, b_k
\}$. In each round Spoiler picks one of the $k$ pebbles $a_1, \ldots,
a_k$, say $a_i$, and puts it on an element of the structure
$\Astruct$. Duplicator responds by picking the corresponding pebble
$b_i$ and placing it on some element of the structure $\Bstruct$. For
simplicity, in any given configuration of the game let us identify a
pebble with the element of the structure that it is placed on. Spoiler
wins if at any point during the game the partial function $f : A
\rightarrow B$ defined by $f(a_i) = b_i$, for each pebbled element
$a_i$ of $\Astruct$, is either not well defined (because there exist
indices $i,j \in [k]$ of two pebbled elements such that $a_i = a_j$
but $b_i \not= b_j$), or is not a partial homomorphism. Otherwise, the
Duplicator wins.

We say that a finite relational structure $\Bstruct$ has \emph{width}
$k$ if, for every finite structure $\Astruct$ of the same vocabulary
as $\Bstruct$, if there is no homomorphism from $\Astruct$ to
$\Bstruct$, then Spoiler wins the existential $k$-pebble game on
$\Astruct$ and $\Bstruct$.  The structure $\Bstruct$ has \emph{bounded
  width} if it has width $k$ for some $k$. Structures of bounded width are exactly those structures for which CSP($\Bstruct$) can be solved by a local consistency algorithm~\cite{KolaitisVardi2000}.

\subsection{Propositional and polynomial encodings} \label{sec:encodings}

To reason about proof systems for CSPs we encode the fact that a finite
structure $\Astruct$ maps homomorphically to a finite
structure $\Bstruct$, over the same vocabulary, as a CNF or a system of polynomial inequalities
or/and equations. In the proofs we will use concrete fixed encodings
but our results hold for a whole class of encodings which we call
\emph{local}.

\paragraph{Local encodings.} First let us fix some notation. In the
context of propositional proof systems, for any sets $A$ and $B$ by
$V(A,B)$ we denote a set of propositional variables: for every $a \in
A$ and every $b \in B$ there is a variable $X(a,b)$ in the set
$V(A,B)$. Truth valuations of the variables in $V(A,B)$ and relations
on $A \times B$ have a natural one-to-one correspondence: a variable
$X(a,b)$ is assigned the truth value $1$ if and only if the pair
$(a,b)$ belongs to the relation.  Recall that a function $f$ from $A$
to $B$ is a relation $\{(a, f(a)) : a \in A\}$ on $A \times B$. Hence,
a homomorphism from an $L$-structure $\Astruct$ to an $L$-structure
$\Bstruct$ is a relation on $A \times B$.

Fix a finite relational vocabulary $L$ and a finite structure
$\Bstruct$ over $L$.

A \emph{propositional encoding scheme} $\E$ for $\CSP(\Bstruct)$ is a
mapping which assigns to every $L$-structure $\Astruct$ a set of
clauses $\E(\Astruct)$ over the variables in $V(A,B)$ in such a way
that there is a one-to-one correspondence between the truth valuations
of the variables in $V(A,B)$ satisfying $\E(\Astruct)$ and
the homomorphisms from $\Astruct$ to $\Bstruct$.  

In the context of algebraic and semi-algebraic proof systems we
additionally assume the presence of twin variables. For every $a \in
A$ and every $b \in B$ there is both the algebraic variable $X(a,b)$
and the algebraic variable $\bar{X}(a,b)$ in the set $V(A,B)$, and an
analogous bijective correspondence holds between relations of $A
\times B$ and those evaluations of the variables from $V(A,B)$ in
$\{0,1\}$ which satisfy the axioms from~(\ref{eqn:axioms}): a pair
$(a,b)$ belongs to the relation if and only if the variable $X(a,b)$
is assigned the value $1$ if and only if the variable $\bar{X}(a,b)$
is assigned the value $0$.

An \emph{algebraic encoding scheme} $\E$ over a field $F$ for
$\CSP(\Bstruct)$ is a mapping which assigns to every $L$-structure
$\Astruct$ a set of polynomial equations $\E(\Astruct)$ over the
variables in $V(A,B)$ in such a way that there is a one-to-one
correspondence between the evaluations of the variables form $V(A,B)$
in $\{0,1\}$ satisfying $\E(\Astruct)$ and the axioms
from~(\ref{eqn:axioms}) over $F$, and the homomorphisms
from~$\Astruct$ to~$\Bstruct$. Finally, a \emph{semi-algebraic
  encoding scheme} $\E$ for $\CSP(\Bstruct)$ is a mapping which
assigns to every $L$-structure $\Astruct$ a set of polynomial
inequalities $\E(\Astruct)$ over the variables in $V(A,B)$ in such a
way that there is a one-to-one correspondence between the evaluations
of the variables form $V(A,B)$ in $\{0,1\}$ satisfying $\E(\Astruct)$
and the axioms from~(\ref{eqn:axioms}) and~(\ref{eqn:axiomssemi}), and
the homomorphisms from~$\Astruct$ to~$\Bstruct$. Observe that every
algebraic encoding scheme over the real-field is also a semi-algebraic
encoding scheme.

An encoding scheme $\E$ is \emph{invariant under isomorphisms} if,
whenever $f : A \rightarrow A'$ is an isomorphism from an
$L$-structure $\Astruct$ to an $L$-structure $\Astruct'$, it holds
that $\E(\Astruct') = f(\E(\Astruct))$, where $f(\E(\Astruct))$ is
obtained from $\E(\Astruct)$ by substituting each variable $X(a,b)$ by
$X(f(a),b)$ (and each variable $\bar{X}(a,b)$ by $\bar{X}(f(a),b)$ if
necessary).

Next we define the key notion of \emph{local} encoding scheme. We need two pieces of notation. If the structure $\Astruct$
has a single element and each of its relations is empty, we
denote the encoding $\E(\Astruct)$ by $\E(a)$. If the structure
$\Astruct$ has a single non-empty relation $R({\Astruct})$ with a
single tuple $(a_1,\ldots,a_r)$ in it, and its domain is
$\{a_1,\ldots,a_r\}$, then we denote $\E(\Astruct)$ by
$\E(R(a_1,\ldots,a_r))$. Since the vocabulary $L$ is finite, up to
isomorphism there are only finitely many structures of one of the
above-mentioned two kinds. Therefore, for any relational structure
$\Bstruct$ over a finite vocabulary $L$ and any encoding scheme $\E$
that is invariant under isomorphisms, the size of encodings of the
form $\E(a)$ or $\E(R(a_1,\ldots,a_r))$ is bounded by a constant. We
call it the \emph{local bound} of the encoding scheme.

An encoding scheme $\E$ in \emph{local} if it is invariant under
isomorphisms and, for every $L$-structure $\Astruct$, the encoding
$\E(\Astruct)$ is a sum of $\E(a)$ over all $a \in A$ and
$\E(R(a_1,\ldots,a_r))$ over all $R \in L$ and $(a_1,\ldots,a_r) \in
R({\Astruct})$.  For our purposes all local encodings of the same kind
(i.e., propositional, algebraic or semi-algebraic) are essentially
equivalent, as formalized by the following result.

\begin{lemma}\label{lem:encodings}
  Let $\Bstruct$ be a finite structure over a finite vocabulary $L$,
  and let $(\E,\E')$, $(\F,\F')$ and $(\Gg,\Gg')$ be pairs of local
  encoding schemes for $\Bstruct$ that are propositional, algebraic
  and semi-algebraic, respectively. There exists a positive integer
  $p$ such that for every finite $L$-structure $\Astruct$ it holds
  that:
\begin{enumerate}  \itemsep=0pt
\item every clause in $\E'(\Astruct)$ has a resolution proof from $\E(\Astruct)$ of size bounded by $p$,
\item every equation in $\F'(\Astruct)$ has an NS proof from $\F(\Astruct)$ of size and degree bounded by $p$,
\item every inequality in $\Gg'(\Astruct)$ has an SA proof from $\Gg(\Astruct)$ of size and degree bounded by $p$.
\end{enumerate}
\end{lemma}

\begin{proof}
For 1, let~$s$ and~$s'$ be the local bounds of $E$ and $E'$, respectively. Take a clause $C$ from $\E'(\Astruct)$. The clause $C$ belongs to a subset of $\E'(\Astruct)$ of the form $\E'(a)$ or $\E'(R(a_1,\ldots,a_r))$, so the size of $C$ is bounded by $s'$. Without loss of generality suppose that $C$ belongs to a set $\E'(R(a_1,\ldots,a_r))$. The corresponding subset $\E(R(a_1,\ldots,a_r))$ of $\E(\Astruct)$ has size at most $s$. The satisfying truth valuations for $\E(R(a_1,\ldots,a_r))$ and $\E'(R(a_1,\ldots,a_r))$ are the same. Therefore, since $C$ is an element of $\E'(R(a_1,\ldots,a_r))$, we have that $\E(R(a_1,\ldots,a_r))$ logically implies $C$. It follows from the quantitative completeness theorem for resolution (cf. Theorem~\ref{thm:completeness}) that the clause $C$ has a resolution derivation from $\E(R(a_1,\ldots,a_r))$ of size bounded by a function of~$s$ and~$s'$.

The proofs of 2 and 3 are analogous. The completeness theorem for Nullstellensatz and Sherali-Adams (cf.~Theorem~\ref{thm:completeness_alg}) needs to be used instead of Theorem~\ref{thm:completeness}.
\end{proof}

\paragraph{Three specific examples.}
The results of this paper hold for arbitrary local encoding
schemes. However, in the proofs it is often convenient to be specific.
We now introduce three concrete encoding schemes that, in addition,
are defined uniformly with respect to the template $\Bstruct$.

For every structures $\Astruct$ and $\Bstruct$ over the same vocabulary, let
$\CNF(\Astruct,\Bstruct)$  be a set of clauses with:
\begin{enumerate} \itemsep=0pt
\item a clause $\bigvee_{b \in B} X(a,b)$ for each
$a \in A$,
\item a clause $\overline{X(a,b_0)} \vee
  \overline{X(a,b_1)}$ for each $a \in A$ and $(b_0,b_1) \in B^2$
  with $b_0 \not= b_1$,
\item a clause $\bigvee_{i \in [r]} \overline{X(a_i,b_i)}$ for each
  natural number $r$, each $R \in L$ of arity~$r$, each
  $(a_1,\ldots,a_r) \in R({\Astruct})$, and each $(b_1,\ldots,b_r) \in
  B^r \setminus R({\Bstruct})$.
\end{enumerate}
Note that the mapping that to an $L$-structure $\Astruct$ assigns
$\CNF(\Astruct,\Bstruct)$ is a local encoding scheme for
$\CSP(\Bstruct)$. Since this definition is uniform with respect to
$\Bstruct$ we call it simply the $\CNF$ encoding scheme. We use it to
reason about propositional proof systems for $\CSP(\Bstruct)$.

There are two standard ways of encoding a clause into a system of
inequalities: multiplicatively and additively. These give rise to two
local encoding schemes which we use to reason about algebraic and
semi-algebraic proof systems in the context of CSP.  Specifically, the
multiplicative and additive encodings of a clause $C = \overline{X_1}
\vee \cdots \vee \overline{X_\ell} \vee X_{\ell+1} \vee \cdots \vee
X_k$ are the following equation and inequality, respectively:
\begin{align*} 
X_1 \cdots X_\ell \bar{X}_{\ell+1}
\cdots \bar{X}_k = 0 \;\;\;\;\text{ and }\;\;\;\;
\bar{X}_1
+ \cdots + \bar{X}_\ell + X_{\ell+1} + \cdots + X_k - 1 \geq 0.
\end{align*}

Let $\EQ(\Astruct,\Bstruct)$ be the system of polynomial equations
that are multiplicative encodings of the clauses in
$\CNF(\Astruct,\Bstruct)$, that is:
\begin{enumerate} \itemsep=0pt
\item $\prod_{b \in B} \bar{X}(a,b) = 0$ for each $a \in A$,
\item $X(a,b_0) X(a,b_1) = 0$ for each $a \in A$ and $(b_0,b_1) \in B^2$
with $b_0 \not= b_1$,
\item $\prod_{i=1}^r X(a,b_i) = 0$ for each natural number $r$, each
  $R \in L$ of arity $r$, each $(a_1,\ldots,a_r) \in R({\mathbb{A}})$,
  and each $(b_1,\ldots,b_r) \in B^r \setminus R({\mathbb{B}})$.
\end{enumerate}
The mapping that to an $L$-structure $\Astruct$ assigns
$\EQ(\Astruct,\Bstruct)$ is a local encoding scheme for
$\CSP(\Bstruct)$. Note that this scheme makes sense over any field.
We call it the $\EQ$ encoding scheme.  It is used in
Section~\ref{sec:ppdefinitions} to reason both about algebraic and
semi-algebraic proof systems, and in Section~\ref{sec:lowerbound}
while discussing lower bounds for SOS.

Similarly, let $\INEQ(\Astruct,\Bstruct)$ be a system of of linear
inequalities that are additive encodings of the clauses in
$\CNF(\Astruct,\Bstruct)$, that is:
\begin{enumerate} \itemsep=0pt
\item $\sum_{b \in B} X(a,b) - 1 \geq 0$ for each $a \in A$,
\item $\bar{X}(a,b_0) + \bar{X}(a,b_1) - 1 \geq 0$ for each $a \in A$ and
  $(b_0,b_1) \in B^2$ with $b_0 \not= b_1$,
\item $\sum_{i=1}^r \bar{X}(a,b_i) - 1 \geq 0$ for each natural number $r$,
  each $R \in L$ of arity $r$, each $(a_1,\ldots,a_r) \in
  R({\mathbb{A}})$, and each $(b_1,\ldots,b_r) \in B^r \setminus
  R({\mathbb{B})}$.
\end{enumerate}
The mapping that to an $L$-structure $\Astruct$ assigns
$\INEQ(\Astruct,\Bstruct)$ is a local encoding scheme for
$\CSP(\Bstruct)$. We call it the $\INEQ$ encoding scheme.  It is used
in Section~\ref{sec:upperboundinLS} to reason about semi-algebraic
proof systems.

In Section~\ref{sec:threecoloring} we will discuss one more local
semi-algebraic encoding scheme that was used
in~\cite{DBLP:conf/coco/LauriaN17} to prove PC lower bounds for graph
coloring.

\section{General proof complexity facts} \label{sec:generalproofcomplexityfacts}

Substitutions will play a central role in showing that certain
propositional and semi-algebraic proof systems behave well with
respect to the classical CSP reductions. In the case of propositional
proof systems we will consider substitutions of variables by
bounded-DNF formulas with a bounded number of terms, and in the case
of algebraic and semi-algebraic proof systems we will use
substitutions by polynomials with bounded degree and a bounded number
of monomials. We now prove some key technical lemmas regarding such
substitutions.

\subsection{Substitutions in Frege}

In the case of propositional proof systems, a substitution is a mapping from variables to formulas. Applying a
substitution to a formula means replacing all variables by the
corresponding formulas, simultaneously all at once. Since our formulas
are in negation normal form, it is implicit that the
result of applying the substitution $X \mapsto F$ to a negative
literal $\overline{X}$ is the formula dual to $F$, i.e.,
$\overline{F}$.

\begin{lemma} \label{lem:substitution}
Let $k$, $d$ and $m$ be positive integers, let $A$ be a $k$-term and
let $A^+$ be the result of replacing each variable in $A$ by a
(possibly different) $d$-DNF with at most $m$ many terms. Then $A^+$
is logically equivalent to a $k(d+m)$-DNF with at most $m^k d^{km}$
many terms.
\end{lemma}

\begin{proof}
Let $p$ and $n$ be the numbers of positive and negative literals in
$A$, respectively.  After applying the substitution, the $k$-term
becomes a conjunction of $p$ many $d$-DNFs and $n$ many negations of
$d$-DNFs, where each $d$-DNF has at most $m$ many terms.  Applying the
De Morgan rules to the negated $d$-DNFs, what we get is a formula of
the following schematic form:
\begin{equation}
\left({\bigwedge^{p} \bigvee^m \bigwedge^d}\right) \wedge 
\left({\bigwedge^{n} \bigwedge^m \bigvee^d}\right).
\label{eqn:substituted}
\end{equation}
In the left subformula in~\eqref{eqn:substituted}, distributing the
outer conjunction over the disjunction gives a disjunction of at most
$m^{p}$ many $p d$-terms. In the right subformula
in~\eqref{eqn:substituted}, distributing the two outer conjunctions
over the disjunction gives a disjunction of at most $d^{nm}$ many
$nm$-terms. Schematically:
\begin{equation}
\left({\bigvee^{m^{p}} \bigwedge^{p} \bigwedge^d}\right)
\wedge
\left({\bigvee^{d^{nm}} \bigwedge^{nm}}\right).
\label{eqn:distributed}
\end{equation}
Finally, in formula~\eqref{eqn:distributed}, distributing the outer
conjunction over the disjunctions gives a disjunction of $m^p d^{nm}$
many $(pd+mn)$-terms:
\begin{equation}
\bigvee^{m^p d^{nm}} \bigwedge^{pd+mn}.
\end{equation}
Using $p+n \leq k$ we get the result.
\end{proof}

\begin{lemma} \label{lem:replace} 
Fix any positive integers $q$, $d$, $m$ and $p$. 
Let $F$ and $G$ be sets of clauses with at most $q$ variables each, and
let $\sigma$ be a substitution of the variables of $F$ into $d$-DNFs
with at most $m$ many terms on the variables of $G$. For any
positive integers $k$, $s$ and any $t \geq 2$, if $F$ has a Frege refutation of depth $t$, bottom fan-in $k$, and size $s$, and for each formula in $F$ its substitution is a logical
consequence of at most $p$ many clauses from $G$, then $G$ has a Frege refutation of depth $t$, bottom fan-in $k(d+m)$, and size polynomial in $2^k$ and $s$.
\end{lemma}

\begin{proof}
Fix some positive integers $q$, $d$, $m$ and $p$. Assume that $F$ has a $\Sigma_{t,k}$-Frege refutation, for some $k$ and $t \geq 2$. Let $\ell = k(d+m)$.

We now define an operator which maps formulas in $\Sigma_{t,k}$ to formulas in $\Sigma_{t,\ell}$. If a formula $D$ is a variable or the negation of a variable, then we define $D^+$ simply as the $d$-DNF or $d$-CNF obtained by applying the substitution $\sigma$ to~$D$. For a $k$-DNF $D$, we put $D^+$ to be the $\ell$-DNF that one gets from applying Lemma~\ref{lem:substitution} to each $k$-term in $D$ with the
substitution $\sigma$ and then taking the disjunction of the resulting
DNFs. For a $k$-CNF $D$, we define $D^+$ as the complement of $(\overline{D})^+$. In this case $D^+$ is an $\ell$-CNF. Clauses and terms are treated as $1$-DNFs and $1$-CNFs, respectively. Finally, if $D$ is a formula from $\Sigma_{t,k}$ of depth at least $3$, then we define $D^+$ to be the formula constructed by replacing each maximal subformula $E$ of $D$ of depth at most $2$ by $E^+$. By Lemma~\ref{lem:substitution}, the size of $D^+$ is at most polynomial in $2^k$ and the size of $D$.

If $D$ and $E$ are both formulas in $\Sigma_{t,k}$, then $(D \vee E)^+ = D^+ \vee
E^+$. Moreover, for any $D$ it holds that $(\overline{D})^+ = \overline{(D^+)}$. Hence also $(D \wedge E)^+ = D^+ \wedge
E^+$. This means that the result of applying our operator to the premises and conclusion of any of the rules of Frege is an instance of the same rule.

Let $D_1,D_2,\ldots,D_t$ be a $\Sigma_{t,k}$-Frege refutation of $F$ of
size $s$. In order to transform the sequence of formulas $D_1^+,D_2^+,\ldots,D_t^+$ into a valid $\Sigma_{t,\ell}$-Frege refutation of $G$ we need to prove that for each non-logical axiom $D_i$, the formula $D_i^+$ has constant size $\Sigma_{t,\ell}$-Frege proof from $G$.

Each non-logical axiom $D_i$ is a $q$-clause $C$ from
$F$. By assumption, the substitution $\sigma(C)$ and hence also $C^+$
is a logical consequence of at most $p$ many $q$-clauses of
$G$. Moreover, the size of $C^+$ is bounded by a function of $d$, $m$
and $q$, and the total size of the $p$ many $q$-clauses of $G$ that
imply $C^+$ is bounded by a function of $p$ and $q$. The
quantitative completeness theorem for $\Sigma_{t,\ell}$-Frege does the rest:
$D_i^+$ has a $\Sigma_{t,\ell}$-Frege derivation from $G$ of size
bounded by a function of $d$, $m$, $p$ and~$q$.
\end{proof}

\subsection{Substitutions in algebraic and semi-algebraic proof systems}

In the case of algebraic and semi-algebraic proof systems, a
substitution is a mapping from variables to polynomials. Applying a
substitution to an equation or inequality means replacing all
variables by the corresponding polynomials, simultaneously all at
once.

For every set of polynomial equations $F$, by $\mathrm{Eq}(F)$ we denote the
union of $F$ and all the axiom polynomial equations from~(\ref{eqn:axioms}) for the
variables in $F$, i.e., for each variable $X$ or $\bar{X}$ appearing
in one of the equations from $F$, we add to $F$ the polynomial equations $X^2-X = 0$, $\bar{X}^2
- \bar{X} = 0$ and $X+\bar{X}-1=0$.

\begin{lemma} \label{lem:replacePC} Fix any positive integers $d$,
  $m$, $p$, and $q$.  Let $F$ and $G$ be sets of polynomial equations
  of the form $P = 0$, where $P$ is a monomial of degree at most $q$
  with coefficient~$1$, and let $\sigma$ be a substitution of the
  variables of $F$ into polynomials on the variables of $G$ of degree
  at most $d$, with at most $m$ many monomials and every coefficient
  equal $1$. For $\mathcal{P}$ being the Nullstellensatz or Polynomial
  Calculus proof system over any field, and for any positive integers
  $k$ and~$s$, if $F$ has a $\mathcal{P}$ refutation of degree~$k$,
  size $s$, and for each equation in $\mathrm{Eq}(F)$ its substitution
  follows from at most $p$ many equations from $G$ on all evaluations
  of its variables in $\{0,1\}$ over the underlying field, then $G$
  has a $\mathcal{P}$ refutation of degree linear in $k$ and size
  polynomial in $2^k$ and $s$.
\end{lemma}

\begin{proof}
Let us fix some positive integers
  $d$, $m$, $p$, and $q$. 
  Let $F$ and $G$ be sets of polynomial equations of the form $P = 0$, where $P$ is a monomial of degree at most $q$ with coefficient~$1$, and let $\sigma$ be a
  substitution of the variables of $F$ into
  polynomials on the variables of $G$ of degree at most~$d$, with at most $m$
  many monomials and every coefficient equal $1$. 
  If for each equation in $\mathrm{Eq}(F)$
  its substitution follows from at most $p$ many equations from $G$ on all evaluations of its variables in $\{0,1\}$, then 
by Theorem~\ref{thm:completeness_alg} for every equation in $\mathrm{Eq}(F)$
  its substitution has an NS derivation from $G$. The size and degree of this derivation are bounded by some constants which depend on $d$, $m$, $p$, and $q$.

Suppose that $\mathcal{P}$ is the Nullstellensatz proof system and assume that for some positive integers $k$ and $s$, the set of equations $F$ has an NS refutation of degree $k$, size $s$.
The refutation of $F$ is of the form
\begin{equation}
\sum_{i=1}^r P_i \cdot c_i \prod_{j \in J_i} X_j \prod_{k \in K_i} \bar{X}_k = -1,
\end{equation}
where $P_1,\ldots,P_r$
are polynomials such that the equation $P_i = 0$ is in the
set $\mathrm{Eq}(F)$, and
$c_1,\ldots,c_r$ are scalars. 
We substitute the variables in the above equality according to $\sigma$ and substitute the polynomials from the set $\mathrm{Eq}(F)$ by their NS derivations. This way we obtain an NS refutation of $G$ of degree linear in $k$ and size polynomial in $2^k$ and $s$.

Suppose that $\mathcal{P}$ is the Polynomial Calculus proof system and assume that for some positive integers $k$ and $s$, the set of equations $F$ has a PC refutation of degree $k$, size~$s$.
The PC refutation of $G$ goes as follows: first for each equation in $\mathrm{Eq}(F)$ we derive its substitution in the Nullstellensatz proof system, and then we simulate the subsequent steps of the refutation of $F$. Applications of addition and multiplication by scalars remain as they were, and applications of multiplication by variables are simulated in several steps. Since after applying the substitution to the variables they become polynomials of degree at most $d$, with at most $m$ many monomials
  and every coefficient equal $1$, we can simulate multiplication by a variable by at most $md$ multiplication steps and at most $m-1$ additions. The substitution of variables causes a blow-up in size which is polynomial in $2^k$, and the simulation additionally increases the size by a constant factor. Altogether, the degree of the PC refutation of $G$ described above is linear in $k$ and its size is polynomial in $2^k$ and $s$.
\end{proof}

For every set of polynomial inequalities $F$, by $\mathrm{Ineq}(F)$ we denote the
union of $F$ and all the axiom polynomial inequalities and equations from~(\ref{eqn:axiomssemi}) and~(\ref{eqn:axioms}) for the
variables in $F$, i.e., for each variable $X$ or $\bar{X}$ appearing
in one of the equations from $F$, we add to $F$ the polynomial equations $X^2-X = 0$, $\bar{X}^2
- \bar{X} = 0$, $X+\bar{X}-1=0$, and inequalities $X \geq 0$, $\bar{X} \geq 0$, $1-X \geq 0$, $1-\bar{X} \geq 0$, $1 \geq 0$.

\begin{lemma} \label{lem:replaceLS}
Fix any positive integers
  $d$, $m$, $p$, and $q$.
  Let $F$ and $G$ be sets of polynomial equations of the form $P = 0$, where $P$ is a monomial of degree at most $q$ with coefficient~$1$, and let $\sigma$ be a
  substitution of the variables of $F$ into
  polynomials on the variables of $G$ of degree at most $d$, with at most $m$
  many monomials and every coefficient equal $1$. For $\mathcal{P}$ being the
Sherali-Adams, Positive Semidefinite
Sherali-Adams, Sums-of-Squares, Lov\'asz-Schrijver or Positive Semidefinite 
Lov\'asz-Schrijver
  proof system, and for any positive integers $k$ and $s$,
  if $F$ has a $\mathcal{P}$ refutation of degree~$k$, size $s$, and for each inequality and equation in $\mathrm{Ineq}(F)$
  its substitution follows from at most $p$ many equations from $G$ on all evaluations of its variables in $\{0,1\}$, then $G$ has
  a $\mathcal{P}$ refutation of degree linear in $k$ and size polynomial in $2^k$ and $s$.
\end{lemma}

\begin{proof}
Let us fix some positive integers
  $d$, $m$, $p$, and $q$. Let $F$ and $G$ be sets of polynomial equations of the form $P = 0$, where $P$ is a monomial of degree at most $q$ with coefficient~$1$, and let $\sigma$ be a
  substitution of the variables of $F$ into
  polynomials on the variables of $G$ of degree at most~$d$, with at most $m$
  many monomials and every coefficient equal $1$. If for an inequality or equation in $\mathrm{Ineq}(F)$
  its substitution follows from at most $p$ many equations from $G$ on all evaluations of its variables in $\{0,1\}$, then 
by Theorem~\ref{thm:completeness_alg} 
  such substitution has an SA derivation from $G$. Moreover, the size and degree of those derivations are bounded by some constants which depend on $d$, $m$, $p$, and $q$.

Suppose that $\mathcal{P}$ is the Sherali-Adams or Positive Semidefinite
  Sherali-Adams proof system and assume that for some positive integers $k$ and $s$, the set of equations $F$ has an SA (or SA$^+$) refutation of degree $k$, size $s$.
The refutation of $F$ is of the form
\begin{equation}
\sum_{i=1}^r P_i \cdot c_i \prod_{j \in J_i} X_j \prod_{k \in K_i} \bar{X}_k = -1,
\end{equation}
where $c_1,\ldots,c_r$ are reals and $P_1,\ldots,P_r$
are polynomials such that the equation $P_i = 0$ or the inequality $P_i \geq 0$ is in the
set $\mathrm{Ineq}(F)$, or they are squares of
polynomials when they are allowed (i.e., for SA$^+$). We substitute the variables in the above equality according to $\sigma$ and substitute the polynomials from the set $\mathrm{Ineq}(F)$ by their SA derivations. This way we obtain an SA (or SA$^+$) refutation of $G$ of degree linear in $k$ and size polynomial in $2^k$ and $s$.

Suppose that $\mathcal{P}$ is the Sum-of-Squares proof system and assume that for some positive integers $k$ and $s$,
the set of equations $F$ has an SOS refutation of degree $k$, size $s$.
The refutation of $F$ is of the form
\begin{equation}
\sum_{i=1}^r P_i \cdot S_i = -1,
\end{equation}
where $P_1,\ldots,P_r$
are polynomials such that either the equation $P_i = 0$ or the inequality $P_i \geq 0$ is in the
set $\mathrm{Ineq}(F)$, or they are
squares, and $S_1,\ldots,S_r$ are arbitrary polynomials. 
We substitute the variables in the above equality according to $\sigma$ and substitute the polynomials from the set $\mathrm{Ineq}(F)$ by their SA derivations. This way we obtain an SOS refutation of degree linear in $k$ and size polynomial in $2^k$ and $s$.

Suppose that $\mathcal{P}$ is the Lov\'asz-Schrijver or Positive Semidefinite 
Lov\'asz-Schrijver proof system and assume that for some positive integers $k$ and $s$, the set of equations $F$ has an LS (or LS$^+$) refutation of degree $k$, size $s$.
The refutation of $G$ goes as follows: first for each equation and inequality in $\mathrm{Ineq}(F)$ we derive its substitution in the Sherali-Adams proof system, and then we simulate the subsequent steps of the refutation of $F$. Applications of addition and multiplication by positive reals as well as positivity-of-squares when it is allowed (i.e. for LS$^+$) remain as they were, and applications of multiplication by variables are simulated in several steps. Since after applying the substitution to the variables they become polynomials of degree at most $d$, with at most $m$ many monomials
  and every coefficient equal $1$, we can simulate multiplication by a variable by at most $md$ multiplication steps and at most $m-1$ additions. The substitution of variables causes a blow-up in size which is polynomial in $2^k$, and the simulation additionally increases the size by a constant factor. Altogether, the degree of the LS (or LS$^+$) refutation of $G$ described above is linear in $k$ and its size is polynomial in $2^k$ and $s$.
\end{proof}

\subsection{Simulations}\label{sec:simulations}

At a later section we will need to use the known fact that both Polynomial Calculus and
Sherali-Adams efficiently simulate resolution. If $C$ is the clause $\bigvee_{i
  \in I} \overline{X_i} \vee \bigvee_{j \in J} X_j$, let 
\begin{equation}
M(C) := \prod_{i \in I} X_i \prod_{j \in J} \bar{X}_j.
\end{equation}
Note that, under the axioms~\eqref{eqn:axioms}, the clause $C$ is
encoded by the equation $M(C) = 0$ or, in the context of
semi-algebraic proofs, by the pair of inequalities $M(C) \geq 0$ and
$-M(C) \geq 0$. In Section~\ref{sec:encodings} we called this the
multiplicative encoding of~$C$.

\begin{lemma}\label{lem:simulation}
  If $C$ is a clause that has a resolution derivation of width $k$ and
  size $s$ from clauses $C_1,\ldots,C_m$, then the equation $M(C) = 0$
  has a PC proof over any field from $M(C_1) = 0,\ldots,M(C_m) =
  0$ and an SA proof from $M(C_1) = 0,\ldots,M(C_m) = 0$ of degree
  linear in $k$ and size polynomial in $s$ and $k$.
\end{lemma}

\begin{proof}
  Assume that $C$ has a resolution derivation of width $k$ and size
  $s$.  Before we describe the conversions we need to apply a light
  pre-processing to the resolution derivation. Convert each resolution
  step deriving $D \vee E$ from $D \vee X$ and $E \vee \overline{X}$
  into a \emph{symmetric} resolution step in which first $D \vee E
  \vee X$ and $D \vee E \vee \overline{X}$ are derived by weakenings
  from $D \vee X$ and $E \vee \overline{X}$, respectively, and then $D
  \vee E$ is derived from these by resolving on $X$. Let
  $D_1,D_2,\ldots,D_t$ be the resulting resolution derivation. The
  proofs for Polynomial Calculus and for Sherali-Adams are quite
  different because the latter one is a static proof system while the
  former one is not.

  For Polynomial Calculus, we derive the equation $M(D_i) = 0$ for $i
  = 1,\ldots,t$, by induction on~$i$.  When $D_i$ is a clause from the
  set $\{C_1,\ldots,C_m\}$, there is nothing to do. Assume now that
  $D_i$ is derived by a symmetric resolution step from $D_j = D_i \vee
  X$ and $D_k = D_i \vee \overline{X}$, where $j,k < i$.  By induction
  hypothesis the equations $M(D_i) \bar{X} = 0$ and $M(D_i) X = 0$
  have already been derived. Add these equations to the lift of the
  axiom $1 - X - \bar{X} = 0$ by $M(D_i)$ to get the equation $M(D_i)
  = 0$. Next assume that $D_i$ is derived by a weakening step
  from~$D_j$, say $D_i = D_j \vee X$ or $D_i = D_j \vee
  \overline{X}$. By induction hypothesis the equation $M(D_j) = 0$ has
  already been derived. Lift this equation by $X$ or $\bar{X}$ as
  appropriate to get $M(D_i) = 0$. Clearly, the degree of this proof
  is linear in $k$ and the size is polynomial in $s$ and $k$.

  For Sherali-Adams the proof is quite different. For each $D_i$ in
  the resolution derivation we produce an inequality $Q_i \geq 0$ as
  follows. If $D_i$ is an initial clause, let $Q_i := -M(D_i)$ so that
  $Q_i \geq 0$ is one of the given inequalities.  If $D_i$ is obtained
  as a one-variable weakening step deriving $D_j \vee X$ from $D_j$,
  let $Q_i := M(D_j) - M(D_j \vee X)$ and derive $Q_i \geq 0$ by
  lifting the axiom $1-\bar{X} \geq 0$. If $D_i$ is obtained as a
  one-variable weakening step deriving $D_j \vee \overline{X}$ from
  $D_j$, let $Q_i := M(D_j) - M(D_j \vee \overline{X})$ and derive
  $Q_i \geq 0$ by lifting the axiom $1-X \geq 0$.  If $D_i$ is
  obtained as a symmetric resolution step deriving $D_i$ from $D_i
  \vee X$ and $D_i \vee \overline{X}$, let $Q_i :=  M(D_i \vee
  X) + M(D_i \vee \overline{X}) - M(D_i)$ and derive $Q_i \geq 0$ by lifting
  the axiom $\bar{X} + X - 1 \geq 0$. Next consider the DAG of the
  resolution derivation oriented from the initial clauses towards the
  conclusion $D_t$. We assign a weight $c_i$ to each $D_i$ in this DAG
  inductively: the conclusion $D_t$ gets weight $1$, and if all
  immediate successors of $D_i$ have already been assigned weights,
  then $D_i$ gets as weight the sum of the weights of its immediate
  successors. Next multiply each inequality $Q_i \geq 0$ by its weight
  $c_i$ and add them together. This could cause the coefficients in
  the SA proof to go exponentially big, but their bitsize is still
  polynomial. The result is an SA proof of $- M(C) \geq 0$ since the
  only monomial that survives is the conclusion. The reverse
  inequality $M(C) \geq 0$ follows from lifting the axiom $1 \geq
  0$. This gives an SA proof of $M(C) = 0$ as required. The degree of
  this proof is linear in $k$ and its size is polynomial in $s$ and $k$.
\end{proof}

\section{Closure under reductions} \label{sec:ppdefinitions}

Three types of reductions are often considered in the context of constraint satisfaction problems: a) pp-interpretability, b) homomorphic equivalence, c) addition of constants to a core. In this section we give their precise definitions and show that many proof systems behave well with respect to those types of reductions. 

\subsection{Reductions.} 

Let $\Bstruct$ and $\Bstruct'$ be finite relational structures over
finite vocabularies $L$ and $L'$, respectively. We say that the
structure $\Bstruct'$ is \emph{pp-definable} in the structure
$\Bstruct$ if it has the same domain and for every relation symbol $T
\in L'$ the relation $T(\Bstruct')$ is definable in~$\Bstruct$ by a
pp-formula. Recall that a \emph{primitive positive formula} over $L$,
or \emph{pp-formula}, is a first-order formula which uses only symbols
from $L$, equality, conjunction, and first-order existential
quantification. A relation $T \subseteq B^r$ is \emph{definable in
  $\Bstruct$ by a pp-formula}, or \emph{pp-definable in $\Bstruct$},
if there exists a pp-formula $\phi(x_1, \ldots, x_r)$ over $L$, with
free variables $x_1,\ldots,x_r$, such that
\begin{equation*}
T = \{ (b_1, \ldots, b_r) \in B^r : \Bstruct \models \phi(x_1/b_1,
\ldots, x_r/b_r)\}.
\end{equation*}

Pp-interpretability is a generalization of pp-definability which allows for changing the domain of a CSP language. Given two relational structures~$\Bstruct$ and~$\Bstruct'$ in finite
vocabularies~$L$ and~$L'$, respectively, we say that~$\Bstruct'$ is
\emph{pp-interpretable} in~$\Bstruct$ if there exist a positive
integer $n$ and a surjective partial function $f \colon B^n
\rightarrow B'$ such that the preimages of all relations in $\Bstruct'$
(including the equality relation) and the domain of $f$ are
pp-definable in $\Bstruct$. Showing that a CSP over a language $\Bstruct'$ pp-interpretable in the language $\Bstruct$ is not harder than the CSP of the language $\Bstruct$ itself~\cite{BulatovJeavonsKrokhin2005} is one of the fundamental results of the so-called \emph{algebraic approach} to constraint satisfaction problem, which led to many break-through results in the area.

Probably the simplest of all the constructions is the homomorphic equivalence. Structures~$\Bstruct$ and~$\Bstruct'$ over a vocabulary $L$ are \emph{homomorphically equivalent} if there exists a homomorphism from $\Bstruct$ to $\Bstruct'$ and a homomorphism from $\Bstruct'$ to $\Bstruct$. Obviously, if $L$-structures $\Bstruct$ and $\Bstruct'$ are homomorphically equivalent, then any $L$-structure $\Astruct$ maps homomorphically to $\Bstruct$ if and only if it maps homomorphically to $\Bstruct'$. So the CSP problems over both languages are the same.

Homomorphic equivalence allows us to focus on studying constraint satisfaction problems of well-behaved structures which in this context turn out to be those exhibiting little symmetry. A finite relational structure is called a \emph{core} if all its endomorphisms are surjective. It is known that every relational structure has a homomorphically equivalent substructure that is a core. Core structures can be extended by one-element unary relations which we refer to as \emph{constants}, without increasing the complexity of the language~\cite{BulatovJeavonsKrokhin2005}.

The importance of the constructions a), b) and c) follows from the fact that classes of
constraint languages closed under those constructions can be studied via the corresponding
algebras of polymorphisms, that is algebras of operations which preserve all the relations in the
language (for details see e.g. the survey~\cite{Barto2017PolymorphismsAH}). Here we show that bounded-DNF Frege, bounded-depth Frege, Frege, Polynomial Calculus, Sherali-Adams, Sums-of-Squares and Lov\'asz-Schrijver of bounded and unbounded degree behave well with respect to those three types of reductions. This 
allows us to apply (in Section~\ref{sec:lowerbound}) strong results based on the algebraic approach to CSP.

\subsection{Results}

Let us fix relational structures $\Bstruct$ and $\Bstruct'$ over finite vocabularies $L$ and $L'$, respectively, such that $\Bstruct'$ is obtained from $\Bstruct$ by a finite sequence of
constructions a), b) and c). In the following we recall the known polynomial-time
computable transformation that maps instances $\Astruct$ of
$\CSP(\Bstruct')$ to instances $\Astruct'$ of $\CSP(\Bstruct)$ such
that $\Astruct'$ is satisfiable if and only if $\Astruct$ is
satisfiable, and the size of $\Astruct'$ is linear in the size of
$\Astruct$. The notation is supposed to remind the reader that 
once a template $\Bstruct'$ is constructed from a template $\Bstruct$,
the transformation of instances goes in the other direction: from an
instance $\Astruct$ of
$\CSP(\Bstruct')$ we build an instance $\Astruct'$ of $\CSP(\Bstruct)$ 
satisfying the above mentioned conditions.

We prove that if $\E$ and $\E'$ are any local propositional encoding schemes for $\CSP(\Bstruct)$ and $\CSP(\Bstruct')$, respectively, then this transformation satisfies the
following:

\begin{theorem}\label{thm:closureboundeddepthFrege}
  For any positive integers $t$, $k$ and $s$, and any $L'$-structure
  $\Astruct$, if there is a Frege refutation of
  $\E(\Astruct')$ of depth~$t$, bottom fan-in $k$, and size
  $s$, then there is a Frege refutation of $\E'(\Astruct)$
  of depth $t$, bottom fan-in polynomial in $k$, and size polynomial
  in~$2^k$, $s$ and the size of $\Astruct$.
\end{theorem}

As a special case of the above theorem, by taking $t=1$ we obtain that bounded-DNF Frege behaves well with respect to the classical CSP reductions:

\begin{corollary}
  For any positive integers $k$ and $s$, and any $L'$-structure
  $\Astruct$, if there is a $k$-DNF Frege refutation of
  $\E(\Astruct')$ of size
  $s$, then there is an $\ell$-DNF Frege refutation of $\E'(\Astruct)$
  of size polynomial
  in~$2^k$, $s$ and the size of $\Astruct$, where $\ell$ is polynomial in $k$.
\end{corollary}

Notice also that a Frege refutation of depth $t$ and bottom fan-in $k$ can be seen as a Frege refutation of depth $t+1$ and bottom fan-in $1$. Therefore, Theorem~\ref{thm:closureboundeddepthFrege} implies the following statement, which will be crucial for obtaining lower bounds in Section~\ref{sec:lowerbound}:

\begin{corollary}\label{col:bdFrege}
 For any positive integers $t$ and $s$, and any $L'$-structure
  $\Astruct$, if there is a Frege refutation of
  $\E(\Astruct')$ of depth $t$ and size
  $s$, then there is a Frege refutation of $\E'(\Astruct)$
  of depth $t+1$ and size polynomial
  in~$s$ and the size of $\Astruct$.
\end{corollary}

One more consequence of Theorem~\ref{thm:closureboundeddepthFrege} concerns proofs in Frege proof system without any bounds on the depth. Corollary~\ref{col:bdFrege} above immediately implies that  Frege is well-behaved with respect to the classical CSP reductions, that is:

\begin{corollary}
 For any positive integer $s$, and any $L'$-structure
  $\Astruct$, if there is a Frege refutation of
  $\E(\Astruct')$ of size
  $s$, then there is a Frege refutation of $\E'(\Astruct)$
  of size polynomial
  in~$s$ and the size of $\Astruct$.
\end{corollary}

In the case of algebraic proof systems, if $\E$ and $\E'$ are any
local algebraic encoding schemes over a field $F$ for $\CSP(\Bstruct)$
and $\CSP(\Bstruct')$, respectively, we show that:

\begin{theorem}\label{thm:closurealgebraic}
  For any positive integers $k$ and $s$, and any $L'$-structure
  $\Astruct$, if there is a PC refutation over $F$ of
  $\E(\Astruct')$ of degree~$k$ and size $s$, then there is a
  PC refutation over $F$ of $\E'(\Astruct)$ of degree
  linear in $k$ and size polynomial in~$2^k$, $s$ and the size of
  $\Astruct$.
\end{theorem}

Finally, if $\E$ and $\E'$ are any local semi-algebraic encoding schemes for $\CSP(\Bstruct)$ and $\CSP(\Bstruct')$, respectively, then:

\begin{theorem}\label{thm:closuresemialgebraic}
  For any positive integers $k$ and $s$, and any $L'$-structure $\Astruct$,
  if there is an SA, SA$^+$, SOS, LS or LS$^+$
  refutation of $\E(\Astruct')$ of degree~$k$ and size $s$,
  then there is, respectively, 
  an SA, SA$^+$, SOS, LS or LS$^+$
  refutation of $\E'(\Astruct)$ of
  degree linear in $k$ and size polynomial in~$2^k$, $s$ and the size of $\Astruct$.
\end{theorem}

We point out that Theorem~\ref{thm:closuresemialgebraic} in the case of the Sherali-Adams and Sums-of-Squares proof systems and the $\EQ$ encoding scheme can be extracted from~\cite{DBLP:journals/siamcomp/ThapperZ17} and~\cite{DBLP:conf/lics/ThapperZ17}.

The main idea in proving the above theorems for all the proof systems
under consideration is the same. The refutation for an instance
$\Astruct'$ of $\CSP(\Bstruct)$ is transformed into a refutation for an
instance $\Astruct$ of $\CSP(\Bstruct')$ by substituting the
variables of $\E(\Astruct')$
by DNFs with a bounded number of terms and a bounded
number of literals in each term,
or by
polynomials with bounded degree, a bounded number of monomials and all
coefficients equal $1$. The additional
condition we need to ensure is that each element of
$\E(\Astruct')$ after applying the substitution is a
logical consequence of a subset of
$\E'(\Astruct)$ of a bounded size. This way we can use Lemmas~\ref{lem:replace}, \ref{lem:replacePC} and~\ref{lem:replaceLS} from Section~\ref{sec:generalproofcomplexityfacts} to control the growth of the size and depth/degree of the refutations. This argument, however, fails if one of the steps in constructing $\Bstruct'$ from $\Bstruct$ is adding the equality relation (which is a special case of a pp-definition). We deal with this by showing that equality propagation can be done in bounded-width resolution.

We prove Theorems~\ref{thm:closureboundeddepthFrege}, \ref{thm:closurealgebraic} and~\ref{thm:closuresemialgebraic} for CNF and EQ encoding schemes in a series of lemmas below. It follows from Lemma~\ref{lem:encodings} that this suffices to obtain the theorems in full generality. Let us see how to argue this for propositional proof systems. The reasoning in the case of algebraic and semi-algebraic proof systems is analogous.

\begin{proof}[Proof of Theorem~\ref{thm:closureboundeddepthFrege}]
Assume that the statement of the theorem holds for $\E$ and $\E'$ being the CNF encoding scheme. Let now $\E$ and $\E'$ be arbitrary local propositional encoding schemes for $\CSP(\Bstruct)$ and $\CSP(\Bstruct')$, respectively. 

By Lemma~\ref{lem:encodings} there exist positive integers $p$ and $p'$ such that for each $L$-structure $\Astruct'$, every clause in $\E(\Astruct')$ has a resolution proof from $\CNF(\Astruct',\Bstruct)$ of size bounded by $p$, and for each $L'$-structure $\Astruct$, every clause in $\CNF(\Astruct,\Bstruct')$ has a resolution proof from $\E(\Astruct)$ of size bounded by $p'$.

Take an $L'$-structure
  $\Astruct$, and assume that there is a Frege refutation of
  $\E(\Astruct')$ of depth~$t$, bottom fan-in $k$, and size
  $s$. Since every clause in $\E(\Astruct')$ has a resolution proof from $\CNF(\Astruct',\Bstruct)$ of size bounded by $p$, it follows that $\CNF(\Astruct',\Bstruct)$ has a Frege refutation of depth $t$, bottom fan-in $k$, and size linear in $s$. The statement of the theorem holds for the CNF encoding schemes, so $\CNF(\Astruct,\Bstruct')$ has a Frege refutation of depth~$t$, bottom fan-in polynomial in $k$, and size polynomial
  in~$2^k$, $s$ and the size of $\Astruct$. Since every clause in $\CNF(\Astruct,\Bstruct')$ has a resolution proof from $\E(\Astruct)$ of size bounded by $p'$, it follows that $\E(\Astruct)$ has a Frege refutation of depth $t$, bottom fan-in polynomial in $k$, and size polynomial
  in~$2^k$, $s$ and the size of $\Astruct$.
\end{proof}

In the subsequent sections we consider one by one the cases when $\Bstruct'$ is constructed from $\Bstruct$ using a), b) and c). We begin with pp-definability, with which we deal in three steps: by considering the equality relation, pp-formulas using conjunction only and existential quantification only.

\subsection{Equality}

Suppose that none of the relation symbols in $L$ interprets in $\Bstruct$ as the equality relation.
For a binary relation symbol $E$ not in $L$, let $L' = L \cup \{E\}$.
Assume that $\Bstruct'$ is the $L'$-structure with domain $B$, all relation symbols from $L$ interpreted as in $\Bstruct$, i.e., $R(\Bstruct') = R(\Bstruct)$ for every $R \in L$,
and the relation symbol $E$ interpreted as the equality relation over $B$, i.e., $E(\Bstruct') = \{(b,b) : b \in B\}$.

For every instance of the CSP of the language $\Bstruct'$, that is for every finite $L'$-structure~$\Astruct$, there is a natural corresponding instance $\Astruct'$ of the CSP over the language $\Bstruct$. If $\equiv$ is the smallest equivalence relation on $A$ which contains $E(\Astruct)$, then define $\Astruct'$ to be the $L$-structure whose domain $A'$ is the set of the equivalence classes of the relation $\equiv$ and every relation symbol $R \in L$ is interpreted as $\{ ([a_1]_\equiv, \ldots, [a_r]_\equiv) : (a_1, \ldots, a_r) \in R(\Astruct) \}$, where $r$ is the arity of~$R$. It is not difficult to see that~$\Astruct$ maps homomorphically to~$\Bstruct'$ if and only if~$\Astruct'$ maps homomorphically to~$\Bstruct$.

\begin{lemma}\label{lem:eq}
  There exists a positive integer $c$ such that the following holds. For any positive integers $t$, $k$ and $s$ and
  every finite $L'$-structure $\Astruct$, if there is a Frege refutation of $\CNF(\Astruct',\Bstruct)$ of depth $t$, bottom fan-in $k$ and size $s$, then
  there is a Frege refutation of $\CNF(\Astruct,\Bstruct')$ of depth~$t$, bottom fan-in $k$ and
  size at most $(c|A|+1)s$.
\end{lemma}

\begin{proof}
  Let $F$ denote
  $\CNF(\Astruct',\Bstruct)$ and let $G$ denote
  $\CNF(\Astruct,\Bstruct')$. For each $[a]_\equiv \in A'$, we
  choose one element of $[a]_\equiv$ denoted by $a^*$, and consider
  the substitution $\sigma$ of variables in $F$ defined by:
\begin{equation*}
X([a]_\equiv,b) := X(a^*,b),
\end{equation*}
for every $[a]_\equiv \in A'$ and $b \in B$. We show that, for a
constant $c$ to be determined later, for every clause $C$ from $F$,
the substituted formula $\sigma(C)$ has a resolution proof from $G$ of size at
most $c|A|$. It
follows that $G$ has a Frege refutation of depth $t$, bottom fan-in $k$ and
  size at most $(c|A|+1)s$.
 
Let $C$ be any of the clauses in $F$. Note that if $C$ is of type~1
or~2 then by applying the substitution we obtain a clause in $G$, so
there is nothing to be proved.  Now, let us assume that $C$ is a
clause of type~3, i.e., $C = \overline{X([a_1]_\equiv,b_1)} \vee
\cdots \vee \overline{X([a_r]_\equiv,b_r)}$ for some $R \in L$ of
arity~$r$, $([a_1]_\equiv, \ldots, [a_r]_\equiv) \in R(\Astruct')$
and $(b_1, \ldots, b_r) \in B^r \setminus R(\Bstruct)$. The following
claim will finish the proof. We state the bound on width because it
will be useful later. By $q$ we denote the number of elements in $B$.

\begin{claim} \label{lem:propagation}
  There are constants $c$ and $d$ such that the clause
  $\sigma(C)$ has a resolution derivation from $G$ of width at most
  $d$ and size at most $c|A|$.
\end{claim}

\noindent To prove this claim the following observation will be helpful.
   
\begin{claim}\label{lem:equality}
  There is a constant $e$ such that for every $a,a' \in A$ such
  that $a \equiv a'$, and every $b' \in B$, there is a resolution
  proof of $\overline{X(a,b')}$ from $\overline{X(a',b')}$ and clauses
  in $G$ of width at most $e$, length at most $(2q+1)|A|$ and size at
  most $(2q+1)^2|A|$.
\end{claim}
 	
We use Claim~\ref{lem:equality} to prove
Claim~\ref{lem:propagation}. Note that for every $i \in [r]$ there
exists $a'_i \in [a_i]_{\equiv}$ such that $(a'_1, \ldots, a'_r) \in
R(\Astruct)$. Therefore, the clause $\overline{X(a'_1,b_1)} \vee
\cdots \vee \overline{X(a'_r,b_r)}$ belongs to~$G$. Now, since $a^*_1
\equiv a'_1$, it follows from Claim~\ref{lem:equality} that there is a
resolution derivation of width at most $e$, length at most $(2q+1)|A|$
and size bounded by $(2q+1)^2|A|$ of $\overline{X(a^*_1,b_1)}$ from
$\overline{X(a'_1,b_1)}$ and clauses in $G$.  If we reproduce exactly
the same derivation starting with the clause $\overline{X(a'_1,b_1)}
\vee \cdots \vee \overline{X(a'_r,b_r)}$ instead of
$\overline{X(a'_1,b_1)}$, what we get is a valid resolution derivation
of $\overline{X(a^*_1,b_1)} \vee \overline{X(a'_2,b_2)} \vee \cdots
\vee \overline{X(a'_r,b_r)}$ of width at most $e+r$ and size at most
$(2q+1)^2 |A|+(2q+1)(2r-2)|A|$. We repeat the same construction $r-1$
more times starting with the last clause derived and get a resolution
derivation of $\sigma(C)$ whose width is bounded by $re$ and whose size
is bounded by $((2q+1)^2r+(2q+1)(2r-2)r)|A|$. It
remains to prove Claim \ref{lem:equality}.
 	
\medskip 	
 	
\noindent \textit{Proof of Claim \ref{lem:equality}.}  First let us
show that for every $a,a' \in A$ such that $(a,a') \in \ E(\Astruct)$
or $(a',a) \in \ E(\Astruct)$, and every $b' \in B$ there is a
resolution proof of width at most $q$, length at most $2q+1$ and
size bounded by $(2q+1)^2$ of $\overline{X(a,b')}$ from $
\overline{X(a',b')}$ and the clauses in $G$. Indeed, the cut rule
applied to $\overline{X(a',b')}$ and the formula $\bigvee_{b \in B}
X(a',b)$ from $G$ gives $\bigvee_{b \in B'} X(a',b)$, where $B' = B
\setminus \{ b' \}$. Then by a sequence of $q-1$ cuts with formulas
$\overline{X(a,b')} \vee \overline{X(a',b)}$, for $b \in B'$, we
derive $\overline{X(a,b')}$. The total number of formulas in this
sequence is $2q+1$, and each has width at most $q$ and size at most
$2q+1$.

Now, let $a=a_1, \ldots, a_m=a'$ be a sequence of elements of $A$
such that $(a_i,a_{i+1}) \in E(\Astruct)$ or $(a_{i+1},a_{i}) \in
 E(\Astruct)$, and let us assume that this is one of the shortest
sequences with this property.  The statement of the claim then follows
from the fact that $m \leq |A|$.
 \end{proof}

\begin{lemma}\label{lem:eqLS}
  Let $\mathcal{P}$ be Polynomial Calculus,
  Sherali-Adams, Positive Semidefinite Sherali-Adams, Sums-of-Squares, Lov\'asz-Schrijver or Positive Semidefinite Lov\'asz-Schrijver.
  For any positive integers $k$ and $s$, and every finite
  $L'$-structure $\Astruct$, if there is a $\mathcal{P}$ refutation of
  $\EQ(\Astruct',\Bstruct)$ of degree $k$ and size $s$, then there is
  a $\mathcal{P}$ refutation of $\EQ(\Astruct,\Bstruct')$ of degree
  linear in $k$ and size polynomial in $|A|$ and $s$.
\end{lemma}

\begin{proof}
  Let $F$ denote
  $\EQ(\Astruct',\Bstruct)$ and let $G$ denote
  $\EQ(\Astruct,\Bstruct')$.  Analogously as in the proof of Lemma~\ref{lem:eq} above, for each $[a]_\equiv \in A'$, we
  choose one element of $[a]_\equiv$ denoted by $a^*$, and consider
  the substitution $\sigma$ of the variables in $F$ defined by:
\begin{align*}
X([a]_\equiv,b) := X(a^*,b), \ \ \ \ \ \ \ \ \ \
\bar{X}([a]_\equiv,b) := \bar{X}(a^*,b),
\end{align*}
for every $[a]_\equiv \in A'$ and $b \in B$. 

We show that every equation from $\mathrm{Eq}(F)$ after applying the substitution $\sigma$
has a PC derivation from $\mathrm{Eq}(G)$ of constant degree and size polynomial in $|A|$. Once we have this, the proof for $\mathcal{P}$ being the Polynomial Calculus proof system follows the same lines as the proof of Lemma~\ref{lem:replacePC}.
Similarly, for (Positive Semidefinite) Sherali-Adams, Sums-of-Squares or (Positive Semidefinite) Lov\'asz-Schrijver proof systems, the proof follows the same lines as the proof of Lemma~\ref{lem:replaceLS} once we show that every inequality from $\mathrm{Ineq}(F)$ after applying the substitution $\sigma$
has an SA derivation from $\mathrm{Ineq}(G)$ of constant degree and size polynomial in~$|A|$.

Note that by applying the substitution to equations of type 1 and
2, and to the axiom equations and inequalities we obtain equations and inequalities from $\mathrm{Eq}(G)$ and $\mathrm{Ineq}(G)$ so
there is nothing to be proved.  Now, consider an equation of type~3
from $F$, i.e., $X([a_1]_\equiv,b_1) \cdot \ldots \cdot
X([a_r]_\equiv,b_r) = 0$ for some $R \in L$ of arity $r$,
$([a_1]_\equiv, \ldots, [a_r]_\equiv) \in R(\Astruct')$ and $(b_1,
\ldots, b_r) \in B^r \setminus R(\Bstruct)$. In the
terminology of Section~\ref{sec:simulations}, this equality is the
multiplicative encoding $M(C) = 0$ of the corresponding clause $C$ from
CNF$(\Astruct',\Bstruct)$. Note that $\sigma(C)$ is again a
clause because $\sigma$ is a substitution by variables. Let us call this clause $D$. The substitution of the equation we are
considering is the multiplicative encoding $M(D) = 0$ of the clause $D$. Now, by
Claim~\ref{lem:propagation} that we proved inside the proof of
Lemma~\ref{lem:eq}, there is a resolution derivation of $D$ from $G$ of some constant width $d$ and size linear in $|A|$. Since $G$ is the set of multiplicative encodings of the clauses in
$\CNF(\Astruct,\Bstruct')$,
Lemma~\ref{lem:simulation} applies and we get both PC and SA
proofs of $M(D) = 0$ from $G$ of degree linear in
$d$ and size polynomial in $|A|$. This is a constant degree, and the proof is complete.
\end{proof}

\subsection{Conjunction}

We now consider the case when the structure $\Bstruct'$ is
pp-definable from $\Bstruct$ by adding a single relation pp-definable
using conjunction only. Let $S$ and $P$ be relation symbols in $L$,
let $T$ be a relation symbol not in $L$, let $L' = L \cup \{T\}$, and
assume that $\Bstruct'$ is the expansion of $\Bstruct$ with the
relation $T(\Bstruct')$ defined using a pp-formula
$\phi(x_1,\ldots,x_r)$, where $r$ is the arity of $T$, that is made of
a conjunction of one atom on $S$ and one atom on $P$.  That is,
$R(\Bstruct') = R(\Bstruct)$ for every $R \in L$, and $T(\Bstruct') =
\{ (b_1,\ldots,b_r) \in B^r : \Bstruct \models
\phi(x_1/b_1,\ldots,x_r/b_r) \}$. To focus our attention let us assume
that $S$ and $P$ are binary, $T$ is ternary, and that the pp-formula
that defines $T$ is $\phi(x_1,x_2,x_3) = S(x_1,x_2) \wedge
P(x_2,x_3)$. The proof of the general case will be the same.

For a finite $L'$-structure $\Astruct$ the corresponding $L$-structure
$\Astruct'$ has the same domain and all the relation symbols except
for $S$ and $P$ interpreted the same as in $\Astruct$.  Moreover,
$S(\Astruct') = S(\Astruct) \cup \{(a_1,a_2) : (a_1,a_2,a_3) \in
T(\Astruct) \text{ for some } a_3\}$, and $P(\Astruct') = P(\Astruct)
\cup \{(a_2,a_3) : (a_1,a_2,a_3) \in T(\Astruct) \text{ for some } a_1
\}$.  It is easy to see that $\Astruct$ maps homomorphically
to~$\Bstruct'$ if and only if~$\Astruct'$ maps homomorphically
to~$\Bstruct$.

\begin{lemma}\label{lem:conjunction}
For any positive integers $t$, $k$ and $s$, and every finite $L'$-structure $\Astruct$, if there
  is a Frege refutation of $\CNF(\Astruct',\Bstruct)$ of depth $t$, bottom fan-in $k$ and size $s$,
  then there is a Frege refutation of
  $\CNF(\Astruct,\Bstruct')$ of depth $t$, bottom fan-in $k$ and size linear in $s$.
\end{lemma}

\begin{proof}
  Let $F$ denote
  $\CNF(\Astruct',\Bstruct)$ and let $G$ denote
  $\CNF(\Astruct,\Bstruct')$.  Observe that the variables of $F$
  and $G$ as well as the clauses of type 1 and 2 are the same.  
  Below we show that every clause $C$ of
  type 3 in $F$ is a logical consequence of a bounded number of
  clauses of $G$. It follows that every clause $C$ of
  type 3 in $F$ has a resolution derivation from $G$ of size bounded by some constant $c$, hence $G$ has a Frege refutation of depth $t$, bottom fan-in $k$ and
  size at most~$cs+s$.

Let $C$ be a clause $\overline{X(a_1,b_1)} \vee \cdots
\vee \overline{X(a_r,b_r)}$ for some natural number $r$, $R \in L$ of
arity~$r$, $(a_1,\ldots,a_r) \in R({\Astruct'})$, and $(b_1,\ldots,b_r)
\in B^r \setminus R({\Bstruct})$. If $R \not \in \{S,P\}$ then $C$ is
also a clause of $G$ and there is nothing to be proved.  Without loss
of generality let us assume that $R=S$ and hence $C$ is of the form
$\overline{X(a_1,b_1)} \vee \overline{X(a_2,b_2)}$ where $(a_1,a_2)
\in S({\Astruct'})$, and $(b_1,b_2) \in B^2 \setminus S({\Bstruct})$.
Now if $(a_1,a_2) \in S(\Astruct)$ then $C$ is a clause of $G$ and
we are done.  Otherwise, there exists $a_3 \in A'$ such that
$(a_1,a_2,a_3) \in T(\Astruct)$ and for every $b_3 \in B$ there is a
clause $\overline{X(a_1,b_1)} \vee \overline{X(a_2,b_2)} \vee
\overline{X(a_3,b_3)}$ in $G$. Indeed, since $(b_1,b_2) \in B^2
\setminus S({\Bstruct})$ we have that $(b_1,b_2,b_3) \in B^3 \setminus
T({\Bstruct}')$.  The number of such clauses is bounded by $q^{\ell}$,
where $q$ is the number of elements in $B$ and $\ell$ is the arity of $T$.  Those clauses together with the clause
of type 1 for $a_3$ logically imply~$C$.
\end{proof}

\begin{lemma}\label{lem:conjunctionLS}
Let $\mathcal{P}$ be Polynomial Calculus,
  Sherali-Adams, Positive Semidefinite Sherali-Adams, Sums-of-Squares, Lov\'asz-Schrijver or Positive Semidefinite Lov\'asz-Schrijver.
  For any positive integers $k$ and $s$, and every finite
  $L'$-structure $\Astruct$, if there is a $\mathcal{P}$ refutation of
  $\EQ(\Astruct',\Bstruct)$ of degree $k$ and size $s$, then there is
  a $\mathcal{P}$ refutation of $\EQ(\Astruct,\Bstruct')$ of degree
  linear in $k$ and size linear in $s$.
\end{lemma}

\begin{proof}
  Let $F$ denote
  $\EQ(\Astruct',\Bstruct)$ and let $G$ denote
  $\EQ(\Astruct,\Bstruct')$. Observe that the variables of $F$
  and $G$ as well as the equations of type 1 and 2, and the axiom equations and inequalities are the same.  
   Below we show that each equation of
  type 3 in $F$ follows from a bounded number of equations~in~$\mathrm{Eq}(G)$ on all evaluations of its variables in $\{0,1\}$. The way to show this is analogous as in Lemma~\ref{lem:conjunction} above. It follows that every equation in $F$ has NS and SA derivations from $G$ of degrees and sizes bounded by some constants. Similarly as in the proofs of Lemmas~\ref{lem:replacePC} and~\ref{lem:replaceLS}, this implies that $G$ has a $\mathcal{P}$ refutation of degree linear in $k$ and size linear in $s$.

Let $P = 0$ be an equation $X(a_1,b_1) \cdot \ldots
\cdot X(a_r,b_r) = 0$ for some natural number $r$, $R \in L$ of
arity~$r$, $(a_1,\ldots,a_r) \in R({\Astruct'})$, and $(b_1,\ldots,b_r)
\in B^r \setminus R({\Bstruct})$. If $R \not \in \{S,P\}$ then the equation $P = 0$ is
also in $G$ and there is nothing to be proved.  Without loss
of generality let us assume that $R=S$ and hence $P = 0$ is of the form
$X(a_1,b_1) X(a_2,b_2) = 0$ where $(a_1,a_2)
\in S({\Astruct'})$, and $(b_1,b_2) \in B^2 \setminus S({\Bstruct})$.
Now if $(a_1,a_2) \in S(\Astruct)$ then the equation $P = 0$ is in~$G$ and
we are done. Otherwise, there exists $a_3 \in A'$ such that
$(a_1,a_2,a_3) \in T(\Astruct)$ and for every $b_3 \in B$ there is an equation $X(a_1,b_1)  X(a_2,b_2) 
X(a_3,b_3) = 0$ in~$G$.  The number of such equations is bounded by $q^{\ell}$,
where $q$ is the number of elements in the domain of $\Bstruct$ and $\ell$ is the arity of~$T$.  The equation in question follows on all evaluations of its variables in $\{0,1\}$ from the set $\mathrm{Eq}(G')$, where $G'$ is the set of those at most $q^{\ell}$ equations together with the equation
of type~1 for $a_3$.
\end{proof}

\subsection{Existential quantification}

We now consider the case when the structure $\Bstruct'$ is
pp-definable from $\Bstruct$ by adding a single relation definable
using existential quantification only.  Let $S$ be a relation
symbol in $L$, let~$T$ be a relation symbol not in $L$, let
$L' = L \cup \{T\}$, and assume that $\Bstruct'$ is the expansion
of~$\Bstruct$ with the relation $T(\Bstruct')$ defined using
a pp-formula $\phi(x_1,\ldots,x_r)$, where $r$ is the arity of~$T$,
that is made of the existential quantification of one variable
over an atom on $S$. That is, $R(\Bstruct') = R(\Bstruct)$ for
every $R \in L$, and $T(\Bstruct') = \{ (b_1,\ldots,b_r) \in B^r :
\Bstruct \models \phi(x_1/b_1,\ldots,x_r/b_r) \}$. To focus
our attention let us assume that $S$ is ternary, $T$ is binary,
and that the pp-formula that defines $T$ is $\phi(x_1,x_2) = \exists  y
S(x_1,x_2,y)$.

For a finite $L'$-structure $\Astruct$ the corresponding $L$-structure $\Astruct'$ has domain $A$ extended by a set of witnesses for $S$. For each $(a_1,a_2) \in T(\Astruct)$, we add to $A$ a new point $y(a_1,a_2)$ so the domain $A'$ is equal to $A \cup \{ y(a_1,a_2) : (a_1,a_2) \in T(\Astruct) \}$. All the relation symbols from $L$ except for $S$ are interpreted in $\Astruct'$ the same as in $\Astruct$, and $S(\Astruct') = S(\Astruct) \cup
\{(a_1,a_2,y(a_1,a_2)) : (a_1,a_2) \in T(\Astruct) \}$.
It is not difficult to see that~$\Astruct$ maps homomorphically to~$\Bstruct'$ if and only if~$\Astruct'$ maps homomorphically to~$\Bstruct$.

\begin{lemma}\label{lem:exist}
  For any positive integers $t$, $k$ and $s$, and every finite $L'$-structure $\Astruct$, if there
  is a Frege refutation of $\CNF(\Astruct',\Bstruct)$ of depth $t$, bottom fan-in $k$ and size $s$,
  then there is a  Frege refutation of
  $\CNF(\Astruct,\Bstruct')$ of depth $t$, bottom fan-in polynomial in $k$ and size polynomial in~$2^k$ and $s$.
\end{lemma}

\begin{proof}
  Let $F$ denote
  $\CNF(\Astruct',\Bstruct)$ and let $G$ denote
  $\CNF(\Astruct,\Bstruct')$.  Assume for this proof that $B =
  [q]$.  For each $b \in [q]$ we define a subset $F(b)$ of
  $T(\Bstruct')$ inductively as follows:
\begin{equation*}
F(b) = \big\{{ (i,j) \in [q]^2 : (i,j,b) \in S(\Bstruct') }\big\} \setminus (F(1)
\cup \cdots \cup F(b-1)).
\end{equation*}
Note that $F(1),\ldots,F(q)$ cover $T(\Bstruct')$ and are pairwise
disjoint. In other words, they partition $T(\Bstruct')$; note however
that some $F(b)$'s may be empty.

Consider the substitution $\sigma$ defined by the identity on all
variables of $G$ and defined as follows for every variable in $F$ that
is not in $G$:
\begin{equation*}
X(y(a_1,a_2),b) := \bigvee_{(b_1,b_2) \in F(b)} X(a_1,b_1) \wedge X(a_2,b_2),
\end{equation*}
for every $(a_1,a_2) \in T(\Astruct)$ and $b \in [q]$. Note that
this is an $\ell$-DNF with at most $q^{\ell}$ many terms, where $\ell$ is the
arity of $T$. By Lemma~\ref{lem:replace} it suffices to check that,
for each clause $C$ of $F$, the substituted formula $\sigma(C)$ is a
logical consequence of a bounded number of clauses~of~$G$.

To argue this, let $C$ be any of the clauses in $F$, say $\bigvee_{b \in [q]}
X(a,b)$ for $a$ in the domain of~$\Astruct'$. If $a$ is not of
the form $y(a_1,a_2)$, then the clause is left untouched by the substitution.
Since the same clause is also in $G$, there is nothing to prove.
Suppose now that $a$ is $y(a_1,a_2)$. The substituted formula is
then the following:
\begin{equation*}
\bigvee_{b \in [q]} \bigvee_{(b_1,b_2) \in F(b)} X(a_1,b_1) \wedge X(a_2,b_2).
\end{equation*}
Since the sets $F(1),\ldots,F(q)$ cover $T(\Bstruct')$, this is
indeed equivalent to 
\begin{equation*}
\bigvee_{(b_1,b_2) \in T(\Bstruct')} X(a_1,b_1) \wedge X(a_2,b_2).
\end{equation*}
Note now that this formula is a logical consequence of the following
clauses of $G$: those of type 1 for $a_1$ and $a_2$, and all those of
type 3 for $(a_1,a_2)$ and the relation symbol $T$. These are at most $q^{\ell} + 2$ many clauses, where $\ell$ is the
arity of $T$,
and we are done for this case.

Suppose now that $C$ is the clause $\overline{X(a,b)} \vee
\overline{X(a,b')}$ for $a$ in the domain of $\Astruct'$ and $(b,b')
\in B^2$ with $b \not= b'$. As in the previous case, if $a$ is not of
the form $y(a_1,a_2)$, then the clause is left untouched by the
substitution and there is nothing to prove.  Suppose now that $a$ is
$y(a_1,a_2)$. By applying the substitution $\sigma$ and converting the
resulting formula into negation normal form we obtain the following:
\begin{equation*}
\Big({\bigwedge_{(b_1,b_2) \in F(b)} \overline{X(a_1,b_1)} \vee \overline{X(a_2,b_2)}}\Big)
\vee
\Big({\bigwedge_{(b_1,b_2) \in F(b')} \overline{X(a_1,b_1)} \vee \overline{X(a_2,b_2)}}\Big).  
\end{equation*}
This formula says that the tuple $(a_1,a_2)$ is either not mapped to
any tuple in $F(b)$ or not mapped to any tuple in $F(b')$.  Since the
sets $F(b)$ and $F(b')$ are disjoint, this is a logical consequence
of at most $2q^2$ many clauses of $G$: those of type 2 for $a_1$ and
$a_2$. Indeed, those clauses imply that the tuple $(a_1,a_2)$ can be
mapped to at most one tuple from $B^2$.

Now, let $C$ be the clause $\overline{X(a_1,b_1)} \vee \cdots \vee
\overline{X(a_r,b_r)}$ for some natural number $r$, $R \in L$ of
arity~$r$, $(a_1,\ldots,a_r) \in R({\Astruct'})$, and
$(b_1,\ldots,b_r) \in B^r \setminus R({\Bstruct})$. If
$(a_1,\ldots,a_r) \in A^r$ then the same argument as above shows
that there is nothing to be proved.  Observe that the only other case
is when $R=S$ and $C$ is of the form $\overline{X(a_1,b_1)} \vee
\overline{X(a_2,b_2)} \vee \overline{X(y(a_1,a_2),b_3)}$, where
$(a_1,a_2) \in T(\Astruct)$ and $(b_1,b_2,b_3) \in B^3 \setminus
S({\Bstruct})$.  The substituted formula (after converting to negation
normal form) is then the following:
\begin{equation*}
\overline{X(a_1,b_1)} \vee \overline{X(a_2,b_2)} \vee {\bigwedge_{(b,b') \in F(b_3)} \overline{X(a_1,b)} \vee \overline{X(a_2,b')}}.  
\end{equation*}  
There are two possibilities. If $(b_1,b_2) \in B^2 \setminus
T({\Bstruct'})$, then the formula above is the logical consequence of
the clause $\overline{X(a_1,b_1)} \vee \overline{X(a_2,b_2)}$ from
$G$.  Otherwise, we have that $(b_1,b_2) \in T({\Bstruct'})$, but
$(b_1,b_2,b_3) \in B^3 \setminus S({\Bstruct})$ which means that
$(b_1,b_2) \not \in F(b_3)$. Observe that the substituted formula says
that the tuple $(a_1,a_2)$ is not mapped to the tuple $(b_1,b_2)$ or
it is not mapped to any tuple from $F(b_3)$. Similarly to the previous
case, this is a logical consequence of at most $2q^2$ many clauses of
$G$: those of type 2 for $a_1$ and $a_2$. This is because those
clauses imply that the tuple $(a_1,a_2)$ can be mapped to at most one
tuple from $B^2$.
\end{proof}

\begin{lemma}\label{lem:existLS}
Let $\mathcal{P}$ be Polynomial Calculus,
  Sherali-Adams, Positive Semidefinite Sherali-Adams, Sums-of-Squares, Lov\'asz-Schrijver or Positive Semidefinite Lov\'asz-Schrijver.
  For any positive integers $k$ and $s$, and every finite
  $L'$-structure $\Astruct$, if there is a $\mathcal{P}$ refutation of
  $\EQ(\Astruct',\Bstruct)$ of degree $k$ and size $s$, then there is
  a $\mathcal{P}$ refutation of $\EQ(\Astruct,\Bstruct')$ of degree
  linear in $k$ and size polynomial in~$2^k$ and $s$.
\end{lemma}

\begin{proof}
 Let $F$ denote
  $\EQ(\Astruct',\Bstruct)$ and let $G$ denote
  $\EQ(\Astruct,\Bstruct')$. 
Assume that $B =
  [q]$. For each $b \in [q]$ define a subset $F(b)$ of
  $T(\Bstruct')$ as in Lemma~\ref{lem:exist} above.
 Consider the substitution $\sigma$ defined by the identity on all
variables of $G$ and defined as follows for every variable in $F$ that
is not in $G$:
\begin{align*}
X(y(a_1,a_2),b) &:= \sum_{(b_1,b_2) \in F(b)} X(a_1,b_1)  X(a_2,b_2), \\
\bar{X}(y(a_1,a_2),b) &:= \sum_{(b_1,b_2) \in B^2 \setminus F(b)} X(a_1,b_1)  X(a_2,b_2),
\end{align*}
for every $(a_1,a_2) \in T(\Astruct)$ and $b \in [q]$. Note that
those are polynomials of degree $m$ with at most $q^m$ many monomials and all coefficients equal $1$, where $m$ is the
arity of $T$. We will show that for each equation in $F$ and for each axiom inequality and equation, its substitution follows on all evaluations of its variables in $\{0,1\}$ from a bounded number of equations in $\mathrm{Eq}(G)$. By Lemmas~\ref{lem:replacePC} and~\ref{lem:replaceLS} this implies the statement of the lemma.
The way to show this is analogous as in Lemma~\ref{lem:exist} above.

Let $P = 0$ be any of the equations in $F$, say $\prod_{b \in B}
\bar{X}(a,b) = 0$ for $a$ in the domain of~$\Astruct'$. If $a$ is not of
the form $y(a_1,a_2)$, then the equation is left untouched by the substitution and there is nothing to prove.
Suppose now that $a$ is $y(a_1,a_2)$. The substituted equation is
then the following:
\begin{equation*}
\prod_{b \in B} \big( \sum_{(b_1,b_2) \in B^2 \setminus F(b)} X(a_1,b_1)  X(a_2,b_2) \big) = 0.
\end{equation*}
This equation follows on all evaluations of its variables in $\{0,1\}$ from the
set $\mathrm{Eq}(G')$, where $G'$ contains the following
equations of $G$: those of type 1 and 2 for $a_1$ and $a_2$, and all those of
type 3 for $(a_1,a_2)$ and the relation symbol $T$. Indeed, take any evaluation satisfying $\mathrm{Eq}(G')$. It corresponds to a mapping from $\{a_1,a_2\}$ to $B$, where $(a_1,a_2)$ is mapped to a pair $(b_1,b_2)$ in $T(\Bstruct')$. Since the sets $F(1),\ldots,F(q)$ form a partition of $T(\Bstruct')$, there is $b \in B$ such that $(b_1,b_2) \in F(b)$. For such $b$ it holds that $X(a_1,b_1)  X(a_2,b_2) = 0$ whenever $(b_1,b_2) \in B^2 \setminus F(b)$. There are at most $q^{\ell} + 2q^2 + 2$ many equations in $G'$, where $\ell$ is the
arity of $T$,
so we are done for this case.

Suppose now that $P = 0$ is the equation $X(a,b)  X(a,b') = 0$ for $a$ in the domain of~$\Astruct'$ and $(b,b')
\in B^2$ with $b \not= b'$. As in the previous case, if $a$ is not of
the form $y(a_1,a_2)$, then the equation is left untouched by the
substitution and there is nothing to prove.  Suppose now that $a$ is
$y(a_1,a_2)$. By applying the substitution $\sigma$ we obtain the following:
\begin{equation*}
\big( {\sum_{(b_1,b_2) \in F(b)} X(a_1,b_1)  X(a_2,b_2)} \big)
\cdot
\big( {\sum_{(b_1,b_2) \in F(b')} X(a_1,b_1)  X(a_2,b_2)} \big) = 0.  
\end{equation*}
Since the
sets $F(b)$ and $F(b')$ are disjoint, this equation follows on all evaluations of its variables in $\{0,1\}$ from the set of equations of type 2 for $a_1$ and
$a_2$. Indeed, those equations imply that at most one of the pairs $(b_1,b_2) \in B^2$ the product $X(a_1,b_1) X(a_2,b_2)$ is $1$.

Now, let $P = 0$ be the equation $X(a_1,b_1) \cdot \ldots \cdot
X(a_r,b_r) = 0$ for some natural number~$r$, $R \in L$ of
arity~$r$, $(a_1,\ldots,a_r) \in R({\Astruct'})$, and
$(b_1,\ldots,b_r) \in B^r \setminus R({\Bstruct})$. If
$(a_1,\ldots,a_r) \in A^r$ then the same argument as above shows
that there is nothing to be proved.  Observe that the only other case
is when $R=S$ and the equation is of the form $$X(a_1,b_1) 
X(a_2,b_2)  X(y(a_1,a_2),b_3) = 0,$$ where
$(a_1,a_2) \in T(\Astruct)$ and $(b_1,b_2,b_3) \in B^3 \setminus
S({\Bstruct})$.  The substituted equation is then the following:
\begin{equation*}
X(a_1,b_1)  {X}(a_2,b_2) \ \ \cdot \big({\sum_{(b,b') \in F(b_3)} {X(a_1,b)} {X(a_2,b')}}\big) = 0.  
\end{equation*}  
There are two possibilities. If $(b_1,b_2) \in B^2 \setminus
T({\Bstruct'})$, then the equation above follows on all evaluations of its variables in $\{0,1\}$ from the equation ${X}(a_1,b_1)  {X}(a_2,b_2) = 0$ from~$G$.  Otherwise, we have that $(b_1,b_2) \in T({\Bstruct'})$, but
$(b_1,b_2,b_3) \in B^3 \setminus S({\Bstruct})$ which means that
$(b_1,b_2) \in B^2 \setminus F(b_3)$. In this case, the substituted equation follows on all evaluations of its variables in $\{0,1\}$ from the set of all equations of type 2 for $a_1$ and $a_2$, which imply that for at most one of the pairs $(b_1,b_2) \in B^2$ the product $X(a_1,b_1)  X(a_2,b_2)$ is $1$.

Let us consider the axiom equation $X(a,b)^2-X(a,b)=0$ for $a$ in the domain of~$\Astruct'$ and $b
\in B$. If $a$ is not of
the form $y(a_1,a_2)$, then the equation is left untouched by the
substitution and there is nothing to prove. Suppose now that $a$ is
$y(a_1,a_2)$. By applying the substitution~$\sigma$ we obtain the following:
\begin{equation*}
\big(\sum_{(b_1,b_2) \in F(b)} X(a_1,b_1)  X(a_2,b_2)\big)^2 - \sum_{(b_1,b_2) \in F(b)} X(a_1,b_1)  X(a_2,b_2) = 0.  
\end{equation*}
This equation follows on all evaluations of its variables in $\{0,1\}$ from $\mathrm{Eq}(G')$ where $G'$ is the set of
equations of type 1 and 2 for~$a_1$ and~$a_2$. 

Let us consider the axiom equation $X(a,b)+\bar{X}(a,b)-1=0$ for $a$ in the domain of~$\Astruct'$ and $b \in B$. If $a$ is not of
the form $y(a_1,a_2)$, then the equation is left untouched by the
substitution and there is nothing to prove. Suppose now that $a$ is
$y(a_1,a_2)$. By applying the substitution $\sigma$ we obtain the following:
\begin{equation*}
\sum_{(b_1,b_2) \in B^2} X(a_1,b_1)  X(a_2,b_2) - 1 = 0.  
\end{equation*}
This equation follows on all evaluations of its variables in $\{0,1\}$ from $\mathrm{Eq}(G')$ where $G'$ is the set of
equations of type 1 and 2 for~$a_1$ and~$a_2$.

Let us consider the axiom inequality $1-X(a,b) \geq 0$, for $a$ in the domain of~$\Astruct'$ and $b \in B$. If $a$ is not of
the form $y(a_1,a_2)$, then the inequality is left untouched by the
substitution and there is nothing to prove. Suppose now that $a$ is
$y(a_1,a_2)$. By applying the substitution $\sigma$ we obtain the following:
\begin{equation*}
1 \ - \sum_{(b_1,b_2) \in F(b)} X(a_1,b_1)  X(a_2,b_2) \geq 0.  
\end{equation*}
This inequality follows on all evaluations of its variables in $\{0,1\}$ from $\mathrm{Eq}(G')$ where $G'$ is the set of
equations of type 1 and 2 for~$a_1$ and~$a_2$.  They imply that at most one of the products $X(a_1,b_1)  X(a_2,b_2)$ for $(b_1,b_2) \in F(b)$ is equal $1$. The same way we deal with the case, when the inequality in question is the axiom inequality $1-\bar{X}(a,b) \geq 0$, for $a$ in the domain of~$\Astruct'$ and $b \in B$.

Finally, the axiom inequalities 
$X(a,b) \geq 0$ and $\bar{X}(a,b) \geq 0$, for $a$ in the domain of~$\Astruct'$ and $b \in B$, after applying the substitution $\sigma$ are always satisfied on evaluations of their variables in $\{0,1\}$. 
\end{proof}

\subsection{All together: pp-interpretations}

Let $\Bstruct'$ be a finite $L'$-structure pp-interpretable in
  $\Bstruct$, and let $f \colon B^n
  \rightarrow B'$ be a surjective partial function such that the domain of $f$ is defined by a pp-formula $\delta(x_1, \ldots, x_n)$ in the language $L$, i.e,. $f^{-1}(B') = \{ (b_1,\ldots,b_n) \in B^n : \Bstruct \models \delta(x_1/b_1,\ldots,x_n/b_n) \}$,
 the preimage of the equality relation on $B'$ is defined by a
   pp-formula $\epsilon(x_1,\ldots, x_{2n})$ in the language $L$,
   i.e., $f^{-1}(\{(b',b') : b' \in B'\}) = \{ ((b_1,\ldots,
   b_n),(b_{n+1}, \ldots, b_{2n})) \in (B^{n})^2 : \Bstruct \models
   \epsilon(x_1/b_1,\ldots,x_{2n}/b_{2n}) \}$,
 and for every relation symbol $R \in L'$ of arity $r$, the preimage
   of the relation $R(\Bstruct')$ is defined by a pp-formula
   $\varphi_R(x_1, \ldots, x_{rn})$ in the vocabulary $L$, i.e.,
   $$f^{-1}(R(\Bstruct')) = \{ ((b_1, \ldots,b_n),\ldots,(b_{nr-n+1},
   \ldots, b_{nr})) \in (B^n)^r : \Bstruct \models
   \varphi_R(x_1/b_1,\ldots,x_{rn}/b_{rn}) \}.$$

Consider the set $f^{-1}(B') \subseteq B^n$ quotiented by the
equivalence relation $f^{-1}(\{(b',b') : b' \in B'\}) \subseteq
(B^n)^2$. For every equivalence class $[(b_1, \ldots,b_n)]$ we choose
a representative $(b_1, \ldots,b_n)^*$. The $L'$-structure whose
domain is the set of all representatives and for each $R \in L'$
of arity $r$ the relation $R$ interpreted as $\{ ((b_1,
\ldots,b_n)^*,\ldots,(b_{nr-n+1}, \ldots, b_{nr})^*) : \Bstruct
\models \varphi_R(x_1/b_1,\ldots,x_{rn}/b_{rn}) \}$ is isomorphic to
$\Bstruct'$. From now on whenever we talk about the structure $\Bstruct'$
we mean the structure that we have just defined.

We now define a structure $\Bstruct''$ pp-definable in $\Bstruct$ and show intuitively that small
 refutations for~$\Bstruct''$ imply small refutations for~$\Bstruct'$. By
the results of previous sections it follows that small 
 refutations for $\Bstruct$ imply small refutations for $\Bstruct'$. To this end, for every relation symbol $R \in L'$ of arity $r$, let
  $\hat{R}$ be a relation symbol of arity $nr$, and let $L'' =
  \{ \hat{R} : R \in L'\}$. We define $\Bstruct''$ to be the finite $L''$-structure with domain $B$ and relations defined by 
$\hat{R}(\Bstruct'') = \{ (b_1,\ldots,b_{rn}) \in B^{rn} :
  \Bstruct \models \varphi_R(x_1/b_1,\ldots,x_{rn}/b_{rn}) \}$, for
  each $\hat{R} \in L''$ of arity $rn$.

For every instance $\Astruct$ of the CSP of the language $\Bstruct'$, that is,
for every finite $L'$-structure $\Astruct$, the corresponding instance of the CSP of the language $\Bstruct''$ is the $L''$-structure
$\Astruct''$ whose domain $A''$ is $A \times [n]$
and whose relations are defined by $$\hat{R}(\Astruct'') = \{ ((a_1,1), \ldots, (a_1,n),(a_2,1),  \ldots \ldots \ldots, (a_r,n)) : (a_1, \ldots, a_r) \in R(\Astruct)\},$$ for each $\hat{R} \in L''$ of arity~$rn$.
It is not difficult to see that~$\Astruct$ maps homomorphically to~$\Bstruct'$ if and only if~$\Astruct''$ maps homomorphically to $\Bstruct''$.

\begin{lemma}\label{lem:interpretation}
 For any positive integers $t$, $k$ and $s$, and every finite $L'$-structure $\Astruct$, if
  there is a Frege refutation of
  $\CNF(\Astruct'',\Bstruct'')$ of depth $t$, bottom fan-in $k$ and size $s$, then there is a
  Frege refutation of $\CNF(\Astruct,\Bstruct')$ of depth $t$, bottom fan-in polynomial in $k$ and size polynomial in~$2^k$ and $s$.
\end{lemma}

\begin{proof}
  Let $F$ denote
  $\CNF(\Astruct'',\Bstruct'')$, let $G$ denote
  $\CNF(\Astruct,\Bstruct')$, and let $q$ denote the number of elements in $B$. For each $i \in [n]$ and each $b \in B$
  we define $F(b,i)$ to be the set of those tuples in $B' \subseteq
  B^n$ which have $b$ on their $i$-th coordinate, i.e., $$F(b,i) = \{
  (b_1, \ldots, b_n) \in B' : b_i = b\}.$$ Observe that for a fixed
  $i$ the sets $F(b,i)$ are disjoint subsets of $B'$ and they cover
  the whole~$B'$. In other words, they partition $B'$; note however
  that some $F(b,i)$'s may be empty.

Consider the following substitution $\sigma$ of the variables of $F$:
\begin{equation*}
X((a,i),b) := \bigvee_{(b_1, \ldots, b_n) \in F(b,i)} X(a,(b_1, \ldots, b_n)).
\end{equation*}
Note that this is a clause with at most $q^{n-1}$ many literals, and
hence a $1$-DNF with at most $q^{n-1}$ many terms.  By
Lemma~\ref{lem:replace} it suffices to check that, for each clause $C$
of $F$, the substituted formula $\sigma(C)$ is a logical consequence
of a bounded number of clauses of $G$.

To argue this, let $C$ be any of the clauses in $F$, say $\bigvee_{b \in B}
X((a,i),b)$ for $(a,i) \in A \times [n]$. 
The substituted formula is
then the following:
\begin{equation*}
\bigvee_{b \in B} \  \bigvee_{(b_1, \ldots, b_n) \in F(b,i)} X(a,(b_1, \ldots, b_n)).
\end{equation*}
Since for each $i \in [n]$ the sets $F(b,i)$ partition $B'$, this is
equivalent to
\begin{equation*}
\bigvee_{(b_1, \ldots, b_n) \in B'} X(a,(b_1, \ldots, b_n)),
\end{equation*}
which is the clause of type 1 for $a$ in $G$. Hence,
we are done for this case.

Suppose now that $C$ is the clause $\overline{X((a,i),b)} \vee
\overline{X((a,i),b')}$ for $(a,i) \in A \times [n]$ and $(b,b') \in
B^2$ with $b \not= b'$. If either of the sets $F(b,i)$ or $F(b',i)$ is
empty, then $X((a,i),b)$ (or $X((a,i),b')$ respectively) is
substituted by the empty formula and $\sigma(C)$ is true so there is
nothing to be proved. Otherwise, the substituted formula (after
converting to negation normal form) is the following:
\begin{equation*}
\Big({\bigwedge_{(b_1, \ldots, b_n) \in F(b,i)} \overline{X(a,(b_1, \ldots, b_n))}}\Big) \vee
\Big({\bigwedge_{(b_1, \ldots, b_n) \in F(b',i)} \overline{X(a,(b_1, \ldots, b_n))}}\Big),  
\end{equation*}
and it says that either $a$ is not mapped to any of the elements in
$F(b,i)$, or it is not mapped to any of the elements in
$F(b',i)$. Since the sets $F(b,i)$ and $F(b',i)$ are disjoint, this
formula is a logical consequence of $q^n (q^n-1) / 2$ clauses of $G$:
those of type 2 for $a$. Indeed, those clauses imply that the element
$a$ can be mapped to at most one tuple from $B'$.

Now, let $C$ be the clause $\overline{X((a_1,1),b_1)} \vee \cdots \vee
\overline{X((a_r,n),b_{nr})}$ for some $\hat{R} \in L''$ of
arity~$nr$, $(a_1,\ldots,a_r) \in R({\Astruct})$, and
$(b_1,\ldots,b_{nr}) \in B^{nr} \setminus \hat{R}({\Bstruct''})$.  If
for some $j \in [r]$ and some $i \in [n]$ the set $F(b_{nj-n+i},i)$ is
empty, then the variable $X((a_j,i),b_{nj-n+i})$ is substituted by the
empty formula, in which case $\sigma(C)$ is true, and there is nothing
to be proved.  Otherwise, the substituted formula is a $q^{n-1}$-DNF:
for each $j \in [r]$ and each $i \in [n]$ there is a term which says
that $a_j$ is not mapped to any tuple from $F(b_{nj-n+i},i)$, that is,
$a_j$ is not mapped to any tuple in $B'$ which has $b_{nj-n+i}$ on the
$i$-th coordinate.  There are two cases: either the tuple $((b_1,
\ldots,b_n),\ldots,(b_{nr-n+1}, \ldots, b_{nr}))$ belongs to $(B')^r$
or not. In the second case without loss of generality let us assume
that $(b_1, \ldots, b_n) \not \in B'$. In particular this means that
$\bigcap_{i\in [n]}F(b_i,i) = \emptyset$. Then we argue that the
formula
  \begin{equation*}
\Big({\bigwedge_{(b'_1, \ldots, b'_n) \in F(b_1,1)} \overline{X(a_1,(b'_1, \ldots, b'_n))}}\Big) \vee \cdots \vee
\Big({\bigwedge_{(b'_1, \ldots, b'_n) \in F(b_n,n)} \overline{X(a_1,(b'_1, \ldots, b'_n))}}\Big),
\end{equation*}
and hence also the substituted formula $\sigma(C)$ is a logical
consequence of $q^n(q^n-1)/2 +1$ clauses of $G$: those of type 1 and 2
for $a_1$. Indeed, those $q^n(q^n-1)/2 +1$ clauses imply that $a_1$ is
mapped to exactly one element from $B'$. Since $\bigcap_{i\in
  [n]}F(b_i,i) = \emptyset$, this in turn implies that there exist $i$
such that $a_1$ is not mapped to any tuple from $F(b_i,i)$ and we are
done. Otherwise, if the tuple $((b_1, \ldots,b_n),\ldots,(b_{nr-n+1},
\ldots, b_{nr}))$ belongs to $(B')^r$, then the substituted formula is
a logical consequence of at most $rq^n(q^n-1)/2+1$ clauses of $G$: the
clauses of type 2 for $a_1, \ldots, a_r$ and the clause of type 3 for
$R \in L'$, $(a_1,\ldots,a_r) \in R({\Astruct})$, and $((b_1,
\ldots,b_n),\ldots,(b_{nr-n+1}, \ldots, b_{nr})) \in (B')^r \setminus
R(\Bstruct')$. This is not very difficult to see since those
$rq^n(q^n-1)/2 +1$ clauses imply that the tuple $(a_1,\ldots,a_r)$ is
mapped to at most one tuple from $(B')^r$ and is not mapped to $((b_1,
\ldots,b_n),\ldots,(b_{nr-n+1}, \ldots, b_{nr}))$. This in turn
implies that for some $j \in [r]$ and some $i \in [n]$, $a_j$ is not
mapped to any tuple in $B'$ which has $b_{nj-n+i}$ on the $i$-th
coordinate, and we are done. More formally, the $rq^n(q^n-1)/2 +1$
clauses in question imply that for every $j \in [r]$ at most one of
the variables in
  \begin{equation*}
    \bigcup_{i \in [n]} \{ X(a_j,(b'_1,\ldots,b'_n)) : (b'_1,\ldots,b'_n) \in F(b_{nj-n+i},i) \}
\end{equation*}
is true. Since for every $j \in [r]$, we have that $\bigcap_{i \in
  [n]} F(b_{nj-n+i},i) = \{(b_{nj-n+1}, \ldots, b_{nj}) \}$, and at
least one of the variables $X(a_1,(b_1, \ldots, b_n)), \ldots,
X(a_r,(b_{nr-n+1}, \ldots, b_{nr}))$ is false, this implies that there
exists $j \in [r]$ such that for some $i \in [n]$ each of the
variables in $\{ X(a_j,(b'_1,\ldots,b'_n)) : (b'_1,\ldots,b'_n) \in
F(b_{nj-n+i},i) \}$ is false, which finishes the proof in this case. 
\end{proof}

\begin{lemma}\label{lem:interpretationLS}
Let $\mathcal{P}$ be Polynomial Calculus,
  Sherali-Adams, Positive Semidefinite Sherali-Adams, Sums-of-Squares, Lov\'asz-Schrijver or Positive Semidefinite Lov\'asz-Schrijver.
  For any positive integers $k$ and $s$, and every finite
  $L'$-structure $\Astruct$, if there is a $\mathcal{P}$ refutation of
  $\EQ(\Astruct'',\Bstruct'')$ of degree $k$ and size $s$, then there
  is a $\mathcal{P}$ refutation of $\EQ(\Astruct,\Bstruct')$ of degree
  linear in $k$ and size polynomial in~$2^k$ and $s$.
\end{lemma}

\begin{proof}
  Let $F$ denote
  $\EQ(\Astruct'',\Bstruct'')$ and let $G$ denote
  $\EQ(\Astruct,\Bstruct')$. For each $i \in [n]$ and each $b \in B$
  we define $F(b,i)$ as in the proof of Lemma~\ref{lem:interpretation} above, and we consider the following substitution $\sigma$ of the variables of $F$:
\begin{align*}
X((a,i),b) &:= \sum_{(b_1, \ldots, b_n) \in F(b,i)} X(a,(b_1, \ldots, b_n)), \\
\bar{X}((a,i),b) &:= \sum_{(b_1, \ldots, b_n) \in B' \setminus F(b,i)} X(a,(b_1, \ldots, b_n)).
\end{align*}
Those are polynomials of degree $1$ with at most $q^n$ many monomials and all coefficients equal~$1$.
We will show that for each equation in $F$ and for each axiom inequality and equation, its substitution follows on all evaluations of its variables in $\{0,1\}$ from a bounded number of equations in $\mathrm{Eq}(G)$. By Lemmas~\ref{lem:replacePC} and~\ref{lem:replaceLS} this implies the statement of the lemma.

Let $P = 0$ be any of the equations in $F$, say $\prod_{b \in B}
\bar{X}((a,i),b) = 0$ for $(a,i) \in A \times [n]$. 
The substituted equation is
then the following:
\begin{equation*}
\prod_{b \in B} \ \big(  \sum_{(b_1, \ldots, b_n) \in B' \setminus F(b,i)} X(a,(b_1, \ldots, b_n)) \big) = 0.
\end{equation*}
Since for each $i \in [n]$ the sets $F(b,i)$ partition $B'$, this
equation follows on all evaluations of its variables in $\{0,1\}$ from the set of equations of type 2 for $a$. This set of equations implies that $X(a,(b_1, \ldots, b_n))$ is $1$ for at most one element of~$B'$.

Suppose now that $P = 0$ is the equation ${X}((a,i),b) 
{X}((a,i),b') = 0$ for $(a,i) \in A \times [n]$ and $(b,b') \in
B^2$ with $b \not= b'$. Since the sets $F(b,i)$ and $F(b',i)$ are disjoint, 
the substituted equation
follows on all evaluations of its variables in $\{0,1\}$ from the set
of equations of type 2 for $a$ in $G$.

Now, let $P = 0$ be the equation ${X}((a_1,1),b_1) \cdot \ldots \cdot
{X}((a_r,n),b_{nr}) = 0$ for some $\hat{R} \in L''$ of
arity~$nr$, $(a_1,\ldots,a_r) \in R({\Astruct})$, and
$(b_1,\ldots,b_{nr}) \in B^{nr} \setminus \hat{R}({\Bstruct''})$.  If
for some $j \in [r]$ and some $i \in [n]$ the set $F(b_{nj-n+i},i)$ is
empty, then the variable ${X}((a_j,i),b_{nj-n+i})$ is substituted by~$0$ and the substituted equation is always satisfied.
Otherwise, there are two cases: either the tuple $((b_1,
\ldots,b_n),\ldots,(b_{nr-n+1}, \ldots, b_{nr}))$ belongs to $(B')^r$
or not. In the second case without loss of generality let us assume
that $(b_1, \ldots, b_n) \not \in B'$. In particular this means that
$\bigcap_{i\in [n]}F(b_i,i) = \emptyset$. Hence, the set of equations of type~2 for~$a_1$ in $G$ imply  
  $$\sum_{(b'_1, \ldots, b'_n) \in F(b_1,1)} {X(a_1,(b'_1, \ldots, b'_n))} \cdot \ldots \cdot
{\sum_{(b'_1, \ldots, b'_n) \in  F(b_n,n)} {X(a_1,(b'_1, \ldots, b'_n))}} = 0.$$
 Otherwise, if the tuple $((b_1, \ldots,b_n),\ldots,(b_{nr-n+1},
\ldots, b_{nr}))$ belongs to $(B')^r$, then the substituted equation follows on all evaluations of its variables in $\{0,1\}$ from the
equations of type 2 for $a_1, \ldots, a_r$, the equation of type 3 for
$R \in L'$, $(a_1,\ldots,a_r) \in R({\Astruct})$, and $((b_1,
\ldots,b_n),\ldots,(b_{nr-n+1}, \ldots, b_{nr})) \in (B')^r \setminus
R(\Bstruct')$.

For the axiom inequalities and equations the proof is the same as in Lemma~\ref{lem:existLS}.
\end{proof}

\subsection{Homomorphic equivalence}

Now let $\Bstruct'$ be a finite $L$-structure homomorphically equivalent to $\Bstruct$. Any $L$-structure $\Astruct$ maps homomorphically to $\Bstruct'$ if and only if it maps homomorphically to $\Bstruct$. 
We fix some homomorphism from~$\Bstruct'$ to~$\Bstruct$ and denote it by~$h$. 

\begin{lemma}\label{lem:homomorphism}
 For any positive integers $k$, $t$ and $s$, and every finite $L$-structure $\Astruct$, if
  there is a Frege refutation of
  $\CNF(\Astruct,\Bstruct)$ of depth $t$, bottom fan-in $k$ and size $s$, then there is a
  Frege refutation of $\CNF(\Astruct,\Bstruct')$ of depth $t$, bottom fan-in polynomial in $k$ and size polynomial in~$2^k$ and $s$.
\end{lemma}

\begin{proof}
  Let $F$ denote
  $\CNF(\Astruct,\Bstruct)$ and let $G$ denote
  $\CNF(\Astruct,\Bstruct')$.  Consider the substitution $\sigma$
  defined as
  follows for every variable in $F$:
\begin{equation*}
X(a,b) := \bigvee_{b' \in h^{-1}(b)} X(a,b'),
\end{equation*}
for every $a \in A$ and $b \in B$. By Lemma~\ref{lem:replace} it suffices to check that,
for each clause $C$ of $F$, the substituted formula $\sigma(C)$ is a
logical consequence of a bounded number of clauses~of~$G$.

To argue this, let $C$ be any of the clauses in $F$, say $\bigvee_{b \in B}
X(a,b)$ for $a$ in the domain of~$\Astruct$. Observe that $\sigma(C)$ is $\bigvee_{b \in B'}
X(a,b)$, which is a clause that belongs to $G$.

Suppose now that $C$ is the clause $\overline{X(a,b)} \vee
\overline{X(a,b')}$ for $a \in A$ and $(b,b')
\in B^2$ with $b \not= b'$. The substituted clause says that $a$ is either not mapped to any of the elements in $h^{-1}(b)$ or it is not mapped to any of the elements in $h^{-1}(b')$. If any of those sets is empty, then $\sigma(C)$ is true. Otherwise, since the sets $h^{-1}(b)$ and $h^{-1}(b')$ are disjoint, $\sigma(C)$ is a consequence of the clauses of type 2 for $a$, which imply that $a$ can be mapped to at most one element in $B'$. 

Now, let $C$ be the clause $\overline{X(a_1,b_1)} \vee \cdots \vee
\overline{X(a_r,b_r)}$ for some natural number $r$, $R \in L$ of
arity~$r$, $(a_1,\ldots,a_r) \in R({\Astruct})$, and
$(b_1,\ldots,b_r) \in B^r \setminus R({\Bstruct})$.
Since $h$ is a homomorphism, all the tuples in $h^{-1}(b_1) \times \ldots \times h^{-1}(b_r)$ belong to $(B')^r \setminus R({\Bstruct'})$. Therefore, $\sigma(C)$ is a logical consequence of the clauses of type 3 in $G$ for the relation symbol $R$, $(a_1,\ldots,a_r) \in R({\Astruct})$ and all tuples $(b'_1,\ldots,b'_r) \in h^{-1}(b_1) \times \ldots \times h^{-1}(b_r)$. 
\end{proof}

\begin{lemma}\label{lem:homomorphismLS}
Let $\mathcal{P}$ be Polynomial Calculus,
  Sherali-Adams, Positive Semidefinite Sherali-Adams, Sums-of-Squares, Lov\'asz-Schrijver or Positive Semidefinite Lov\'asz-Schrijver.
  For any positive integers $k$ and $s$, and every finite
  $L$-structure $\Astruct$, if there is a $\mathcal{P}$ refutation of
  $\EQ(\Astruct,\Bstruct)$ of degree $k$ and size $s$, then there is a
  $\mathcal{P}$ refutation of $\EQ(\Astruct,\Bstruct')$ of degree
  linear in $k$ and size polynomial in~$2^k$ and $s$.
\end{lemma}

\begin{proof}
  Let $F$ denote
  $\EQ(\Astruct,\Bstruct)$ and let $G$ denote
  $\EQ(\Astruct,\Bstruct')$.  Consider the substitution~$\sigma$
  defined as
  follows for every variable in $F$:
\begin{align*}
X(a,b) &:= \sum_{b' \in h^{-1}(b)} X(a,b'), \\
\bar{X}(a,b) &:= \sum_{b' \in B' \setminus h^{-1}(b)} X(a,b'),
\end{align*}
for every $a \in A$ and $b \in B$. By Lemmas~\ref{lem:replacePC} and~\ref{lem:replaceLS} it suffices to check that for each equation in $F$ and for each axiom inequality and equation, its substitution follows on all evaluations of its variables in $\{0,1\}$ from a bounded number of equations in $\mathrm{Eq}(G)$.

Let $P = 0$ be the equation $\prod_{b \in B}
\bar{X}(a,b) = 0$ for $a$ in the domain of $\Astruct$. Since the union of the sets $h^{-1}(b)$ for $b \in B$ is $B'$, the substituted equality follows on all valuations of its variables in $\{0,1\}$ from the set of equations of type 2 for $a$ in $G$.

Suppose now that $P = 0$ is the equation ${X}(a,b)
{X}(a,b') = 0$ for $a \in A$ and $(b,b')
\in B^2$ with $b \not= b'$. Since the sets $h^{-1}(b)$ and $h^{-1}(b')$ are disjoint, the substituted equation follows on all valuations of its variables in $\{0,1\}$ from the set of equations of type 2 for~$a$. 

Now, let $P = 0$ be the equation ${X}(a_1,b_1) \cdot \ldots \cdot
{X}(a_r,b_r) = 0$ for some natural number~$r$, $R \in L$ of
arity~$r$, $(a_1,\ldots,a_r) \in R({\Astruct})$, and
$(b_1,\ldots,b_r) \in B^r \setminus R({\Bstruct})$.
Since $h$ is a homomorphism, all the tuples in $h^{-1}(b_1) \times \ldots \times h^{-1}(b_r)$ belong to $(B')^r \setminus R({\Bstruct'})$. Therefore, the substituted equation follows on all valuations of its variables in $\{0,1\}$ from the set of
equations of type 2 for $a_1, \ldots, a_r$ and of type 3 for the relation symbol $R$, $(a_1,\ldots,a_r) \in R({\Astruct})$ and all tuples $(b'_1,\ldots,b'_r) \in h^{-1}(b_1) \times \ldots \times h^{-1}(b_r)$. 

The argument for the axiom equation and inequalities is the same as in the proof of Lemma~\ref{lem:existLS}.
\end{proof}

\subsection{Adding constants}

Finally we consider the extension by unary one-element relations under the assumption of $\Bstruct$ being a core.
For each $b \in B$, let $R_b$ be a unary relation symbol, not in $L$, and let $L' = L \cup \{R_b : b \in B\}$. We assume that $\Bstruct$ is a core and
$\Bstruct'$ is the $L'$-structure with domain $B$, each relation symbol from $L$ interpreted as in $\Bstruct$, and $R_b(\Bstruct') = \{ b \}$, for every $b \in B$.

For every finite $L'$-structure $\Astruct$ the corresponding $L$-structure $\Astruct'$ has domain
$A' = A \cup B $ (we assume that the sets $A$ and $B$ are disjoint),
and every relation symbol $R \in L$ interpreted as $\bigcup \{ R(\Bstruct)_{b:=a} : b \in B, a \in R_b(\Bstruct)  \} \cup R(\Astruct) \cup R(\Bstruct)$, where $R(\Bstruct)_{b:=a}$ is the relation $R(\Bstruct)$ with every occurrence of $b$ in a tuple substituted by $a$. It follows from the proof of Lemma~23 in~\cite{AtseriasBulatovDawar2009} that $\Astruct$ maps homomorphically to $\Bstruct'$ if and only if $\Astruct'$ maps homomorphically to $\Bstruct$.

\begin{lemma}\label{lem:constants}
  For every two positive integers $k$ and $s$ and every finite
  $L'$-structure $\Astruct$, if there is a Frege refutation of
  $\CNF(\Astruct',\Bstruct)$ of depth $t$, bottom fan-in $k$ and size $s$, then there is a
  Frege refutation of $\CNF(\Astruct,\Bstruct')$ of depth $t$, bottom fan-in $k$ and size at
  most $s$.
\end{lemma}

\begin{proof}
  Let $F$ denote
  $\CNF(\Astruct',\Bstruct)$ and let $G$ denote
  $\CNF(\Astruct,\Bstruct')$.  Consider the substitution $\sigma$
  defined by the identity on all variables of $G$ and defined as
  follows for every variable in $F$ that is not in $G$:
\begin{equation*}
X(b,b') := \begin{cases}
    \emptyformula, & \text{if } b \neq b',\\
    \fullformula, & \text{otherwise},
  \end{cases}
\end{equation*}
for every $(b,b') \in B^2$. By Lemma~\ref{lem:replace} it suffices to check that,
for each clause $C$ of $F$, the substituted formula $\sigma(C)$ is a
logical consequence of a bounded number of clauses~of~$G$.

To argue this, let $C$ be any of the clauses in $F$, say $\bigvee_{b \in B}
X(a,b)$ for $a$ in the domain of $\Astruct'$. If $a \in A$, then the clause is left untouched by the substitution.
Since the same clause is also in $G$, there is nothing to prove.
Suppose now that $a = b' \in B$. One of the variables in $C$ is then $X(b',b')$. This variable is substituted by the true formula so $\sigma(C)$ is true, which finishes the proof in this case.

Suppose now that $C$ is the clause $\overline{X(a,b)} \vee
\overline{X(a,b')}$ for $a$ in the domain of $\Astruct'$ and $(b,b')
\in B^2$ with $b \not= b'$. As in the previous case, if $a \in A$, then the clause is left untouched by the
substitution and there is nothing to prove. Suppose now that $a = b'' \in B$. Then either $b \neq b''$ or $b' \neq b''$. Therefore, either the variable $X(b'',b)$ or the variable $X(b'',b')$ gets substituted by the empty formula and $\sigma(C)$ is true.

Now, let $C$ be the clause $\overline{X(a_1,b_1)} \vee \cdots \vee
\overline{X(a_r,b_r)}$ for some natural number $r$, $R \in L$ of
arity~$r$, $(a_1,\ldots,a_r) \in R({\Astruct'})$, and
$(b_1,\ldots,b_r) \in B^r \setminus R({\Bstruct})$. If
$(a_1,\ldots,a_r) \in A^r$ then the same argument as above shows
that there is nothing to be proved. If $(a_1,\ldots,a_r) \in B^r$ then $C$ is of the form
$\overline{X(b'_1,b_1)} \vee \ldots \vee \overline{X(b'_r,b_r)}$,
where $(b'_1,\ldots,b'_r) \in R(\Bstruct)$ and $(b_1,\ldots,b_r) \in B^r
\setminus R(\Bstruct)$. Then there exists $i \in [r]$ such that $b'_i
\neq b_i$ and the variable $X(b'_i,b_i)$ is substituted by the empty
formula, so once again $\sigma(C)$ is true. The only remaining case is when $(a_1,\ldots,a_r) \in R(\Bstruct)_{b:=a}$, where $a \in R_b(\Bstruct')$ and $a_j = a$ for some (possibly more than one) $j \in [r]$. If there exists $i \in [r]$ such that $a_i = b'_i \in B$ and $b'_i
\neq b_i$ then the variable $X(b'_i,b_i)$ is substituted by the empty
formula, so once again $\sigma(C)$ is true. Otherwise, there exists $j \in [r]$ such that $a_j = a$ and $b_j \neq b$. Then the substituted formula is (possibly a weakening of) the formula $\overline{X(a,b_j)}$ where $b_j \neq b$. This formula
belongs to $G$: it is the clause of type 3 for $a \in R_b(\Astruct')$
and $b_j \in B \setminus R_b(\Bstruct')$.  
\end{proof}

\begin{lemma}\label{lem:constantsLS}
Let $\mathcal{P}$ be Polynomial Calculus,
  Sherali-Adams, Positive Semidefinite Sherali-Adams, Sums-of-Squares, Lov\'asz-Schrijver or Positive Semidefinite Lov\'asz-Schrijver.
  For every two positive integers $k$ and $s$, and every
  finite $L'$-structure $\Astruct$, if there is a $\mathcal{P}$
  refutation of $\EQ(\Astruct',\Bstruct)$ of degree $k$ and size $s$,
  then there is a $\mathcal{P}$ refutation of
  $\EQ(\Astruct,\Bstruct')$ of degree $k$ and size at most $s$.
\end{lemma}

\begin{proof}
  Let $F$ denote
  $\EQ(\Astruct',\Bstruct)$ and let $G$ denote
  $\EQ(\Astruct,\Bstruct')$. Consider the substitution~$\sigma$
  defined by the identity on all variables of $G$ and defined as
  follows for every variable in $F$ that is not in $G$:
\begin{equation*}
X(b,b') := \begin{cases}
    0, & \text{if } b \neq b',\\
    1, & \text{otherwise},
  \end{cases} \;\;\;\;
  \bar{X}(b,b') := \begin{cases}
    1, & \text{if } b \neq b',\\
    0, & \text{otherwise},
  \end{cases}
\end{equation*}
for every $(b,b') \in B^2$. By Lemmas~\ref{lem:replacePC} and~\ref{lem:replaceLS} it suffices to check that for each equation in $F$ and for each axiom inequality and equation, its substitution follows on all evaluations of its variables in $\{0,1\}$ from a bounded number of equations in $\mathrm{Eq}(G)$. We show this analoguosly as in the proof of Lemma~\ref{lem:constants}.

Let $P = 0$ be any of the equations in $F$, say $\prod_{b \in B}
\bar{X}(a,b) = 0$ for $a$ in the domain of $\Astruct'$. If $a \in A$, then the equation is left untouched by the substitution.
Since the same equation is also in $G$, there is nothing to prove.
Suppose now that $a = b' \in B$. One of the variables in $\prod_{b \in B}
\bar{X}(a,b)$ is then $\bar{X}(b',b')$. This variable is substituted by $0$ so the substituted equation holds for all valuations of its variables in $\{0,1\}$.

Suppose now that $P = 0$ is the equation ${X}(a,b) 
{X}(a,b') = 0$ for $a$ in the domain of $\Astruct'$ and $(b,b')
\in B^2$ with $b \not= b'$. As in the previous case, if $a \in A$, then the equation is left untouched by the
substitution and there is nothing to prove. Suppose now that $a = b'' \in B$. Then either $b \neq b''$ or $b' \neq b''$. Therefore, either the variable ${X}(b'',b)$ or the variable ${X}(b'',b')$ gets substituted by $0$ so the substituted equation holds for all valuations of its variables in~$\{0,1\}$.

Now, let $P = 0$ be the equation ${X}(a_1,b_1) \cdot \ldots \cdot
{X}(a_r,b_r) = 0$ for some natural number~$r$, $R \in L$ of
arity~$r$, $(a_1,\ldots,a_r) \in R({\Astruct'})$, and
$(b_1,\ldots,b_r) \in B^r \setminus R({\Bstruct})$. If
$(a_1,\ldots,a_r) \in A^r$ then the same argument as above shows
that there is nothing to be proved. Otherwise, if $(a_1,\ldots,a_r) \in B^r$ then $P = 0$ is of the form
${X}(b'_1,b_1) \cdot \ldots \cdot {X}(b'_r,b_r) = 0$,
where $(b'_1,\ldots,b'_r) \in R(\Bstruct)$ and $(b_1,\ldots,b_r) \in B^r
\setminus R(\Bstruct)$. Hence, there exists $i \in [r]$ such that $b'_i
\neq b_i$ and the variable ${X}(b'_i,b_i)$ is substituted by $0$, so once again the substituted equation holds for all valuations of its variables in $\{0,1\}$. The only remaining case is when $(a_1,\ldots,a_r) \in R(\Bstruct)_{b:=a}$, where $a \in R_b(\Bstruct')$ and $a_j = a$ for some (possibly more than one) $j \in [r]$. If there exists $i \in [r]$ such that $a_i = b'_i \in B$ and $b'_i
\neq b_i$ then the variable ${X}(b'_i,b_i)$ is substituted by $0$, so once again the substituted equation holds for all valuations of its variables in $\{0,1\}$. Otherwise, there exists $j \in [r]$ such that $a_j = a$ and $b_j \neq b$. Then the substituted equation follows on all valuations of its variables in $\{0,1\}$ from the equation ${X}(a,b_j) = 0$ where $b_j \neq b$. This equation
belongs to $G$: it is the equation of type 3 for $a \in R_b(\Astruct')$
and $b_j \in B \setminus R_b(\Bstruct')$. 

All the axiom equations and inequalities from $\mathrm{Ineq}(F)$ after applying the substitution $\sigma$ either become true or are axiom equations and inequalities for the variables of $G$. 
\end{proof}

\section{Upper bound} \label{sec:upperbound}


In this section we show that templates of bounded width (cf. Section~\ref{sec:csp}) admit efficient refutations in
resolution. It immediately follows that the bounded width property ensures efficient refutations in bounded depth Frege, as well as in Polynomial Calculus over the reals, Sherali-Adams and Sums-of-Squares proof systems (cf. Lemma~\ref{lem:simulation}). Together with matching lower bounds obtained in the next section, this will complete the proof of Theorem~\ref{thm:stronggap}.

Let $k(n)$ be a function. Let $\Bstruct$ be a finite relational structure over a finite vocabulary and let~$E$ be a propositional encoding scheme for $\CSP(\Bstruct)$. We say that a finite relational structure~$\Bstruct$ has \emph{resolution refutations of width $k(n)$ with respect to the encoding scheme $E$} if, for
every finite structure $\Astruct$ over the same vocabulary as
$\Bstruct$ with $n$ elements, if there is no homomorphism from
$\Astruct$ to $\Bstruct$, then $E(\Astruct)$ has a
resolution refutation of width $k(n)$. We say that $\Bstruct$ has
\emph{resolution refutations of constant width} if there exist a local encoding $E$ and a function $k(n) = O(1)$
such that $\Bstruct$ has resolution refutations of width $k(n)$ with respect to $E$. Lemma~\ref{lem:encodings} implies that a structure $\Bstruct$ has resolution refutations of constant width if and only if it has resolution refutations of constant width with respect to any local encoding scheme. In this section we use the CNF encoding scheme. The
goal is to prove the following:

\begin{theorem} \label{thm:boundedwidth}
Let $\Bstruct$ be a finite relational structure. The following are
equivalent:
\begin{enumerate} \itemsep=0pt
\item $\Bstruct$ has bounded width,
\item $\Bstruct$ has resolution refutations of
constant width.
\end{enumerate}
\end{theorem}

\noindent In preparation for the proof we revisit the characterization of
resolution width in terms of existential pebble games from
\cite{AtseriasDalmau2008}.

Let $L = \{R_0, \ldots, R_q \}$ be a finite relational vocabulary
consisting of $q+1$ symbols of arity~$q$. Let $\Sstruct_q$ be an
$L$-structure with two-element domain $\{0, 1\}$, where each relation
$R_i(\Sstruct_q)$ encodes the set of valuations that satisfy a
$q$-clause with~$i$ negated variables. More precisely, for $0 \leq i
\leq q$, let $R_i(\Sstruct_q) = \{0,1 \}^q \setminus \{(x_1, \ldots,
x_q)\}$ where $(x_1, \ldots, x_q) \in \{0,1\}^q$ is the vector defined
by $x_j = 0$ for $j > i$ and $x_j = 1$, otherwise. Now for every
$q$-CNF $F$, we define an $L$-structure $\Astruct_F$. Its domain is
the set of variables in $F$, and the relation $R_i(\Astruct_F)$ is the set of all
tuples $(X_1, \ldots, X_q)$ such that the clause $\overline{X_1} \vee
\ldots \vee \overline{X_i} \vee X_{i+1} \vee \ldots \vee X_q$ belongs
to~$F$. We allow the variables in the clauses to repeat, so the
definition covers clauses with less than~$q$ literals. Observe that
partial homomorphisms from $\Astruct_F$ to $\Sstruct_q$ correspond to
partial truth assignments to the variables of $F$ that do not falsify
any clause from $F$.  Hence, for every $q$-CNF $F$, it holds that $F$
is satisfiable if and only if there is a homomorphism from~$\Astruct_F$ to~$\Sstruct_q$.

\begin{theorem}[\cite{AtseriasDalmau2008}]\label{thm:width}
Let $k$ and $q$ be positive integers such that $k \geq q$ and let $F$
be $q$-CNF. Then $F$ has a resolution refutation of width $k$ if and
only if Spoiler wins the existential $(k+1)$-pebble game on
$\Astruct_F$ and $\Sstruct_q$.
\end{theorem}

In this section we use the above theorem to establish a similar
correspondence between existential pebble games on arbitrary
structures $\Astruct$ and $\Bstruct$ and bounded width resolution
refutations of $\CNF(\Astruct,\Bstruct)$. In the following, let the
notation $\Astruct \leq^k \Bstruct$ mean that Duplicator wins the
existential $k$-pebble game on $\Astruct$ and $\Bstruct$.

\begin{lemma}\label{lem:games}
Let $\Astruct$ and $\Bstruct$ be relational structures over the same
vocabulary of maximum arity $r$, and let $k$ be an integer such that
$k \geq |B|$ and $k \geq r$. Then:
\begin{enumerate} \itemsep=0pt
\item if $\Astruct \not\leq^{k+2} \Bstruct$,
  then $\CNF(\Astruct,\Bstruct)$ has a
  resolution refutation of width $k+|B|$,
\item if $\Astruct \leq^{k+2} \Bstruct$,
  then $\CNF(\Astruct,\Bstruct)$ does not
  have a resolution refutation of width $k+1$.
\end{enumerate}
\end{lemma}

\begin{proof}
Let $F$ denote $\CNF(\Astruct,\Bstruct)$. Let $q$ be the maximum of
the number of elements in $B$ and the arity of relation symbols in the
vocabulary of $\Astruct$ and $\Bstruct$. Observe that $F$ is a
$q$-CNF. Lemma~\ref{lem:games} follows from Theorem~\ref{thm:width}
together with the following facts.

\begin{claim}\label{lem:pebbles1}
If $\Astruct_F \leq^{k+|B|} \Sstruct_q$, 
then $\Astruct \leq^{k+1} \Bstruct$.
\end{claim}

\begin{claim}\label{lem:pebbles2}
If $\Astruct \leq^k \Bstruct$,
then $\Astruct_F \leq^k \Sstruct_q$.
\end{claim}

\noindent Indeed, if Spoiler wins the existential $(k+2)$-pebble game
on $\Astruct$ and $\Bstruct$ then, by Claim~\ref{lem:pebbles1},
Spoiler wins the existential $(k+|B|+1)$-pebble game on $\Astruct_F$
and $\Sstruct_q$ and, by Theorem~\ref{thm:width}, $F$ has a resolution
refutation of width $k+|B|$. On the other hand, if Duplicator
wins the existential $(k+2)$-pebble game on $\Astruct$ and $\Bstruct$
then, by Claim~\ref{lem:pebbles2}, Duplicator wins the
existential $(k+2)$-pebble game on $\Astruct_F$ and $\Sstruct_q$ and,
by Theorem~\ref{thm:width}, $F$ does not have a resolution refutation
of width $k+1$. It remains to prove Claims~\ref{lem:pebbles1}
and~\ref{lem:pebbles2}.

\medskip 	
 	
\noindent \textit{Proof of Claim \ref{lem:pebbles1}.} We prove
the contrapositive. Suppose that
Spoiler wins the existential $(k+1)$-pebble game on $\Astruct$ and
$\Bstruct$. We give a winning strategy for Spoiler in the existential
$(k+|B|)$-pebble game on $\Astruct_F$ and $\Sstruct_q$. We simulate
each move of Spoiler in the game on $\Astruct$ and $\Bstruct$ by $|B|$
moves in the game on $\Astruct_F$ and $\Sstruct_q$. Suppose that Spoiler
puts the $i$-th pebble on an element $a$ of $\Astruct$. We simulate
this by pebbling elements $X(a,b)$ of $\Astruct_F$, for each $b \in
B$. There are two possibilities. If the answer of Duplicator falsifies
any of the clauses of types 1 or 2 in $F$ then Spoiler wins
immediately. Otherwise, the answer of Duplicator is $1$ for exactly
one element $X(a,b')$. We simulate this by putting a pebble on the
element $b'$ of $\Bstruct$ in the game on $\Astruct$ and
$\Bstruct$. Now, in the game on $\Astruct_F$ and $\Sstruct_q$ the
pebble which lies on the element $X(a,b')$ stays there until Spoiler
picks up the $i$-th pebble form the element $a$ in $\Astruct$. The
other $|B|-1$ pebbles which lie on elements $X(a,b)$ for $b \neq b'$
can be used to simulate subsequent moves. Therefore, to simulate the
existential $(k+1)$-pebble game on $\Astruct$ and $\Bstruct$ we need
only $|B|-1$ extra pebbles.

If during the course of the game Spoiler does not win by falsifying
any of the clauses of types 1 or 2 then the simulation of the game on
$\Astruct$ and $\Bstruct$ continues. Since in the simulated game
Spoiler has a winning strategy, after a finite number of rounds the
partial assignment $f : A \rightarrow B$ defined by $f(a_i) = b_i$ is
not a partial homomorphism. This means that there exist a natural
number $r$, a relation symbol $R \in L$ of arity $r$, a tuple $(a'_1,
\ldots, a'_r) \in R(\Astruct)$ and a tuple $(b'_1, \ldots, b'_r) \in
B^r \setminus R(\Bstruct)$, such that for every $i \in [r]$ the pairs
of elements $(a'_i,b'_i)$ are pebbled by pairs of corresponding
pebbles. It follows from the construction that in the simulation game
on $\Astruct_F$ and $\Sstruct_q$ the pairs of elements
$(X(a'_i,b'_i),1)$ are pebbled by pairs of corresponding pebbles. This
means that the partial assignment defined by the current configuration
of the game falsifies one of the clauses of type 3 in $F$ and Spoiler
wins.

\medskip

\noindent \textit{Proof of Claim \ref{lem:pebbles2}.} We prove
the contrapositive. Suppose that
Spoiler wins the existential $k$-pebble game on $\Astruct_F$ and
$\Sstruct_q$. We give a winning strategy for Spoiler in the
existential $k$-pebble game on $\Astruct$ and $\Bstruct$. We simulate
each move of Spoiler in the game on $\Astruct_F$ and $\Sstruct_q$ by a
single move in the game on $\Astruct$ and $\Bstruct$. Suppose that
Spoiler puts a pebble on an element $X(a,b)$ of $\Astruct_F$. We
simulate this by pebbling the element $a$ of $\Astruct$. If Duplicator
responds by putting the corresponding pebble on the element $b$ of
$\Bstruct$, then we simulate this by pebbling~$1$ in $\Sstruct_q$,
otherwise we pebble $0$ in $\Sstruct_q$. 

It is not difficult to see that this is indeed a winning strategy for Spoiler in the
existential $k$-pebble game on $\Astruct$ and $\Bstruct$. 
Since in the simulated game Spoiler has a winning strategy, after a
finite number of rounds the partial assignment $f : A_F \rightarrow
\{0,1\}$ corresponding to the current configuration of the game on
$\Astruct_F$ and $\Sstruct_q$ is not a partial homomorphism. Observe
that it is not possible to falsify any of the clauses of type 1 in
$F$. If for some $a \in A$ and some $(b, b') \in B^2$ such that $b
\neq b'$, the partial assignment $f$ falsifies
$\overline{X(a,b)} \vee \overline{X(a,b')}$, it means that pairs of
corresponding pebbles lie on pairs of elements $(X(a,b),1)$ and
$(X(a,b'),1)$. Hence, in the simulation game on $\Astruct$
and $\Bstruct$ pairs of corresponding pebbles lie on pairs of elements
$(a,b)$ and $(a,b')$ and the partial assignment is not well
defined. Finally, if for some natural number $r$, a relation symbol $R
\in L$ of arity $r$, a tuple $(a'_1, \ldots, a'_r) \in R(\Astruct)$
and a tuple $(b'_1, \ldots, b'_r) \in B^r \setminus R(\Bstruct)$, the
partial assignment $f$ falsifies the clause $\overline{X(a'_1,b'_1)}
\vee \cdots \vee \overline{X(a'_r,b'_r)}$, it means that in the game
on $\Astruct$ and $\Bstruct$ for every $i \in [r]$, the pairs of
elements $(a'_i,b'_i)$ are pebbled by pairs of corresponding pebbles,
and the partial assignment given by the current configuration of the
game is also not a partial homomorphism, which ends the proof.
\end{proof}

We are ready to wrap-up:

\begin{proof}[Proof of Theorem~\ref{thm:boundedwidth}]
For the implication 1 to 2, assume that $\Bstruct$ has bounded width,
say width~$l$, and let $k = \max\{|B|,r,l\}$, where $r$ is the maximum
arity of the vocabulary of $\Bstruct$. Let $\Astruct$ be a structure
over the same vocabulary as $\Bstruct$ and assume that there is no
homomorphism from~$\Astruct$ to $\Bstruct$.  Then Spoiler wins the
existential $l$-pebble game on $\Astruct$ and $\Bstruct$, and hence
also the existential $(k+2)$-pebble game on $\Astruct$ and $\Bstruct$,
since $k+2 \geq l$.  The hypotheses of Lemma~\ref{lem:games} hold, so
by part \emph{1.} in that lemma, $\CNF(\Astruct,\Bstruct)$ has a
resolution refutation of width $k+|B|$. This shows that $\Bstruct$ has
resolution refutations of width $k+|B|$, and hence resolution
refutations of constant width.

For the implication 2 to 1, assume that
$\Bstruct$ has resolution refutations of width $l$. Again let $k =
\max\{|B|,r,l\}$ where $r$ is the maximum arity of the relations in
the vocabulary of $\Bstruct$.  Let~$\Astruct$ be a structure over the
same vocabulary as $\Bstruct$ and assume that there is no homomorphism
from $\Astruct$ to $\Bstruct$. Then $\CNF(\Astruct,\Bstruct)$ has a
resolution refutation of width $l$, and hence of width $k+1$ since
$k+1 \geq l$. The hypotheses of Lemma~\ref{lem:games} hold, so by part
\emph{2.} in that lemma, Spoiler wins the existential $(k+2)$-pebble
game on $\Astruct$ and $\Bstruct$. This shows that $\Bstruct$ has
width $k+2$, and hence bounded width.
\end{proof}

\section{Lower bounds} \label{sec:lowerbound}

Let $d(n)$, $k(n)$ and $s(n)$ be functions. Let $\Bstruct$ be a finite relational structure over a finite vocabulary and let $E$ be a propositional encoding scheme for $\CSP(\Bstruct)$. We say that the structure $\Bstruct$ has \emph{Frege refutations of depth $d(n)$, bottom fan-in $k(n)$,
and size $s(n)$ with respect to the encoding scheme $E$} if, for every finite structure $\Astruct$ over the
same vocabulary as~$\Bstruct$ with $n$ elements, if there is no
homomorphism from $\Astruct$ to $\Bstruct$, then
$\E(\Astruct)$ has a Frege refutation of depth $d(n)$, bottom fan-in $k(n)$, and
size $s(n)$. We say that $\Bstruct$ has \emph{bounded-depth Frege
refutations of subexponential size} if there exist a local encoding scheme $E$ and functions $d(n) = O(1)$, $k(n) = O(1)$ and
$s(n) = 2^{n^{o(1)}}$ such that the structure $\Bstruct$ has Frege refutations of
depth $d(n)$, bottom fan-in $k(n)$, and size $s(n)$ with respect to $E$. Due to Lemma~\ref{lem:encodings} the structure $\Bstruct$ has bounded-depth Frege
refutations of subexponential size if and only if it has bounded-depth Frege
refutations of subexponential size with respect to any local propositional encoding scheme.

Similarly, for any field $F$, if $E$ is an algebraic encoding scheme
over $F$, we say that the structure $\Bstruct$ has \emph{PC refutations over~$F$ of degree $d(n)$ with respect to the
  encoding scheme $E$} if, for every finite structure $\Astruct$ over
the same vocabulary as~$\Bstruct$ with $n$ elements, if there is no
homomorphism from $\Astruct$ to $\Bstruct$, then $E(\Astruct)$ has a
PC refutation over $F$ of degree $d(n)$. We say that
$\Bstruct$ has \emph{PC refutations over~$F$ of
  sublinear degree} if there exist a local encoding scheme $E$ over
$F$ and a function $d(n) = o(n)$ such that the structure $\Bstruct$
has PC refutations over~$F$ of degree $d(n)$ with
respect to $E$. Due to Lemma~\ref{lem:encodings} the structure
$\Bstruct$ has PC refutations over~$F$ of sublinear
degree if and only if it has PC refutations over~$F$
of sublinear degree with respect to any local algebraic encoding scheme.

Finally, if $E$ is a semi-algebraic encoding scheme, we say that the
structure $\Bstruct$ has \emph{SOS refutations of degree
  $d(n)$ with respect to the encoding scheme $E$} if, for every finite
structure~$\Astruct$ over the same vocabulary as~$\Bstruct$ with $n$
elements, if there is no homomorphism from~$\Astruct$ to~$\Bstruct$,
then $E(\Astruct)$ has a SOS refutation of degree
$d(n)$. We say that $\Bstruct$ has \emph{SOS refutations
  of sublinear degree} if there exist a local encoding scheme $E$ and
a function $d(n) = o(n)$ such that the structure $\Bstruct$ has
SOS refutations of degree~$d(n)$ with respect to $E$. Due
to Lemma~\ref{lem:encodings} the structure $\Bstruct$ has
SOS refutations of sublinear degree if and only if it has
SOS refutations of sublinear degree with respect to any
local semi-algebraic encoding scheme.

The goal of this section is to prove the following:

\begin{theorem} \label{thm:boundeddepth}
Let $\Bstruct$ be a finite relational structure. The following are
equivalent:
\begin{enumerate} \itemsep=0pt
\item $\Bstruct$ has bounded width,
\item $\Bstruct$ has bounded-depth Frege refutations of subexponential
  size,
\item $\Bstruct$ has PC refutations over the reals of
  sublinear degree,
\item $\Bstruct$ has SOS refutations of sublinear degree.
\end{enumerate}
\end{theorem}

The equivalence of \emph{1}\ and \emph{4}\ is
known~\cite{DBLP:conf/lics/ThapperZ17}.  Here we provide an
alternative proof.  The implication \emph{1}\ to \emph{2}\ follows
from Theorem~\ref{thm:boundedwidth}: every resolution refutation is a
Frege refutation of depth one, and if the refutation has bounded
width, then it has polynomial size and hence subexponential size. The
implications \emph{1}\ to \emph{3}\ and \emph{1}\ to \emph{4}\ follow
from Theorem~\ref{thm:boundedwidth} via the fact that bounded-degree
Polynomial Calculus and bounded-degree Sherali-Adams simulate bounded-width resolution (cf.
Lemma~\ref{lem:simulation}); note that the simulation by
bounded-degree Sherali-Adams implies also the simulation by bounded-degree Sums-of-Squares
and, for both Polynomial Calculus and Sums-of-Squares, bounded-degree implies constant, and hence
sublinear, degree. For implications \emph{2}\ to \emph{1},
\emph{3}\ to \emph{1}\ and \emph{4}\ to \emph{1} we use an algebraic
characterization of unbounded width. We begin with some definitions.

\subsection{Algebraic characterization of unbounded width}

Let $G = (G,+,0)$ be a finite Abelian group. For every positive
integer $n$, each $g \in G$ and every $(z_1, \ldots ,z_n) \in
\mathbb{Z}^n$, we define a relation $R_{(g,z_1, \ldots ,z_n)} = \{
(g_1, \ldots,g_n) \in G^n : z_1g_1 + \ldots + z_ng_n = g \}$, where
$z_ig_i$ is a shortcut for the sum of $|z_i|$ copies of
$\operatorname{sign}(z_i)g_i$. Let $\sim$ be the equivalence relation on the set
$\bigcup_{n > 0}G\times\mathbb{Z}^n$ that identifies tuples defining
the same relation, i.e., $(g,z_1,\ldots,z_n) \sim
(g',z'_1,\ldots,z'_{n'})$ if and only if $n = n'$ and
$R_{(g,z_1,\ldots,z_n)} = R_{(g',z'_1,\ldots,z'_{n'})}$.
Let $L(G)$ be the infinite relational vocabulary that for every equivalence class $[(g,z_1,\ldots,z_n)]$ has one $n$-ary relation symbol $E_{[(g,z_1,\ldots,z_n)]}$, and let $\Bstruct(G)$ be the
$L(G)$-structure that has domain~$G$ and where each relation symbol $E_{[(g,z_1,\ldots,z_n)]}$ is interpreted as $R_{(g,z_1,\ldots,z_n)}$.  The CSP of
$\Bstruct(G)$ is called $\mathrm{LIN}(G)$.
One should think about instances of $\mathrm{LIN}(G)$ as systems of linear equations over the group $G$. For simplicity, for any instance $\Astruct$ of $\mathrm{LIN}(G)$ we denote the fact that a tuple $(a_1,\ldots,a_n) \in A^n$ belongs to the relation $E_{[(g,z_1,\ldots,z_n)]}(\Astruct)$ 
 by $z_1a_1 + \ldots + z_na_n
= g$.

Observe that, since there are only finitely many relations of a fixed arity $k$ on the finite set $G$, the equivalence relation $\sim$ restricted to $G\times\mathbb{Z}^k$ has finitely many equivalence classes. For every positive integer $k$, by $L(G,k)$ we denote the finite relational vocabulary which is the subset of $L(G)$ containing all symbols of arity $k$, and by $\Bstruct(G,k)$ we denote the $L(G,k)$-structure obtained from $\Bstruct(G)$ by removing all relations of arity different than~$k$. The CSP problem
over $\Bstruct(G,k)$ is called $k\mathrm{LIN}(G)$. Instances of $k\mathrm{LIN}(G)$ correspond to systems of linear equations over the group $G$ with $k$ variables per equation.

\begin{theorem}[\cite{BartoKozik2014,Bulatov09}] \label{thm:bartokozik}
Let $\Bstruct$ be a finite relational structure. The following
are equivalent:
\begin{enumerate} \itemsep=0pt
\item $\Bstruct$ does not have bounded width,
\item there exists a non-trivial finite Abelian group $G$ such that
  $\Bstruct(G,3)$ is pp-interpretable in $\Bstruct^+$, where $\Bstruct^+$ is
  the expansion of the core of $\Bstruct$ with all constants.
\end{enumerate}
\end{theorem}

\noindent Thus, in view of Theorems~\ref{thm:closureboundeddepthFrege} and~\ref{thm:closuresemialgebraic},
in order to prove that \emph{2.} implies \emph{1.}, \emph{3.} implies \emph{1.}, and \emph{4.} implies \emph{1.} in
Theorem~\ref{thm:boundeddepth}, it suffices to prove lower bounds for
$3{\rm LIN}(G)$, for every non-trivial finite Abelian group $G$.

\subsection{Lower bound for bounded-depth Frege}

In~\cite{Ben02}, an exponential lower bound on the size of
bounded-depth Frege proofs of the so-called \emph{Tseitin formulas}
was obtained by reduction from the pigeonhole principle formulas; the
latter are known to be hard for bounded-depth Frege by the so-called Jewel
Theorem of Proof Complexity~\cite{A88,BIKPPW92,KPW95}. The Tseitin
formulas encode certain systems of linear equations over~$\mathbb{Z}_2$ that are derived from expander graphs. Here we adapt
the formulas to encode systems of linear equations over arbitrary
finite Abelian groups, and then show that the reduction
in~\cite{Ben02} can be generalised to our formulas. We use the CNF encoding scheme.

\begin{theorem}\label{thm:lowerbound}
  For every integer $d$ and every non-trivial finite Abelian group $G$
  there exists a positive constant~$\delta$ and a family of unsatisfiable instances 
  $(\Astruct_n)_{n \geq 1}$ of $3\mathrm{LIN}(G)$, where $\Astruct_n$ has $\Theta(n)$
   variables and $\Theta(n)$ equations, such
that for every sufficiently large integer $n$ every Frege refutation of
$\CNF(\Astruct_n,\Bstruct(G,3))$ of depth $d$ has
  size at least $2^{n^\delta}$.
\end{theorem}

The rest of this section is devoted to the proof of
Theorem~\ref{thm:lowerbound}. We provide a proof for the special case
when $G$ is the cyclic group $\mathbb{Z}_q$ of integers under addition
modulo $q$, for some $q \geq 2$.  Lemma~\ref{lem:abeliangroups} at the
end of this section shows that, thanks to the Fundamental Theorem of
Finite Abelian Groups, the special case of $G = \mathbb{Z}_q$ implies
Theorem~\ref{thm:lowerbound} in full generality. The proof of the
general case would actually be the same, however we believe that by
focusing on simpler groups we make the arguments easier to follow.

\paragraph{Linear equations over Abelian groups.} 
For the rest of this section, let us fix $G$ to be the cyclic group $\mathbb{Z}_q$ of integers under addition modulo $q$, for some $q \geq 2$. Whenever we talk about an element $z$ of the group $G$, where $z$ is some integer, we mean the unique element corresponding to $z$ modulo $q$. 
The instances of $3\mathrm{LIN}(G)$ that we show to be hard for bounded-depth Frege are special cases of so-called \emph{Tseitin graph tautologies} for $\mathbb{Z}_q$ as defined in~\cite{BussGIP01}. Before defining them we need to introduce some terminology.

For a graph $H=(V(H),E(H))$ (directed or undirected) and a set of
vertices $W \subseteq V(H)$ by $\partial(W)$ we denote the
\emph{boundary} of $W$ which is the set of all edges incident with a
vertex in $W$ and with a vertex in $V(H) \setminus W$. If the graph
$H$ is directed, then by $\partial_{-}(W)$ we denote the set of edges
with the head in $W$ and the tail in $V(H) \setminus W$, and by
$\partial_{+}(W)$ we denote the set of edges with the head in $V(H)
\setminus W$ and the tail in $W$. For single vertices $v$ we write
$\partial(v)$ instead of $\partial(\{v\})$. The same convention
explains the notation $\partial_+(v)$ and $\partial_-(v)$.

Consider a directed graph $H = (V(H),E(H))$ and a labelling $\sigma \colon V(H) \rightarrow G$ of the vertices of the graph $H$ by elements of $G$. The \emph{Tseitin graph tautology} $\Astruct(H,\sigma)$ is the following system of linear equations over the group $G$:
\begin{itemize} \itemsep=0pt
\item the set of variables is the set $E(H)$ of the edges of the graph;
\item for every vertex $v \in V(H)$ there is an equation
$$\sum_{e \in \partial_{+}(v)} e - \sum_{e \in \partial_{-}(v)} e = \sigma(v).$$
\end{itemize}
The system $\Astruct(H,\sigma)$ can be seen as an instance of $\mathrm{LIN}(G)$. The formula $\CNF(\Astruct(H,\sigma),\Bstruct(G))$ is called a \emph{Tseitin formula}. If the graph $H$ is obtained from directing the edges of a $k$-\emph{regular} undirected graph, i.e., a graph in which each vertex has degree $k$, then $\Astruct(H,\sigma)$ is an instance of~$k\mathrm{LIN}(G)$.

It is easy to see that if $\sum_{v \in V(H)}\sigma(v) \neq 0$, then the instance $\Astruct(H,\sigma)$ is unsatisfiable. Indeed, since every variable $e$ appears positively on the left-hand side of exactly one equation and negatively on the left-hand side of exactly one equation, by summing up all the equations we get $0$ on the left-hand side and $\sum_{v \in V(H)}\sigma(v)$ on the right-hand side. If $\sum_{v \in V(H)}\sigma(v) \neq 0$ we obtain a contradiction. It is not difficult to show that for a connected graph $H$, the converse statement holds as well.

\begin{lemma}\label{lem:satisfiability}
If $H = (V(H),E(H))$ is a connected directed graph and $\sigma \colon V(H) \rightarrow G$ is a labelling, then the system $\Astruct(H,\sigma)$ is satisfiable and if and only if $\sum_{v \in V(H)}\sigma(v) = 0$.
\end{lemma}

\begin{proof}
The left-to-right direction is clear.  For the opposite direction we define a solution to $\Astruct(H,\sigma)$ by assigning values to edges of the graph $H$ one by one while keeping two invariants: none of the equations gets falsified and the graph induced by unassigned edges is connected. Below we formalize this intuition.

Since $H$ is connected, we can enumerate its edges $e_1, \ldots, e_m$ in such a way that for every $i \in [m]$ the graph $H_{i+1}$ obtained from $H$ by removing the edges $e_1, \ldots, e_i$ and then deleting all isolated vertices, is connected.
Let us denote the system $\Astruct(H,\sigma)$ by $\Astruct_1$. We assign values to the edges of $H$ in the order specified above. Additionally, for each $i \in [m-1]$, after assigning a value to $e_i$ we substitute the variable $e_i$ in the system $\Astruct_i$ with this value, and next move all constants to the right-hand side of the equations. We denote the obtained system by $\Astruct_{i+1}$. The variables of $\Astruct_{i+1}$ are $e_j$, for $j > i$. Observe, that for every $i \in [k]$, the sum of group elements that appear on the right-hand side of all the equations in $\Astruct_i$ is $0$. 

Assume that we have already assigned values to the edges $e_j$ for $j <i $ without falsifying any of the equations in $\Astruct(H,\sigma)$ (this is true for $i=1$). The variable $e_i$ appears in exactly two equations in $\Astruct_{i}$. There are two possibilities:
\begin{itemize} \itemsep=0pt
\item if $i \neq m$ then at least one of the two equations has at least one more variable whose value has not yet been assigned. This is because the graph $H_{i}$ is connected. Then we can assign a value to $e_i$ in such a way that none of the equations in $\Astruct_{i}$ gets falsified: the value is either forced by the other equation, or can be assigned arbitrarily. 
\item if $i = m$ then the two equations which mention the variable $e_k$ are of the form $e_m = g$ and $-e_m = h$, for some elements $g$ and $h$ of the group. All the other equations in the system $\Astruct_{m}$ are of the form $0=0$. Since the sum of the group elements on the right-hand side of the equations in $\Astruct_{m}$ is $0$, we have that $g = -h$ and we can assign the value $g$ to $e_m$, satisfying the last two equations.
\end{itemize}
This finishes the construction of a solution to $\Astruct(H,\sigma)$.
\end{proof}

Sometimes we want to consider subsystems of $\Astruct(H,\sigma)$ induced by some subset of vertices. Let $W \subseteq V(H)$. By 
$\Astruct(W,\sigma)$ we denote the system of linear equations obtained from $\Astruct(H,\sigma)$ by removing all equations corresponding to vertices in $V(H) \setminus W$. In particular, $\Astruct(V(H),\sigma) = \Astruct(H,\sigma)$. Moreover, by $\Astruct(\partial(W),\sigma)$ we denote the system of linear equations consisting of the single equation $\sum_{e \in \partial_{+}(W)} e - \sum_{e \in \partial_{-}(W)} e = \sum_{v \in W}\sigma(v)$ which is the sum of all the equations in $\Astruct(W,\sigma)$. It turns out that whenever the subgraph induced by $W$ is connected, $\Astruct(\partial(W),\sigma)$ carries the essential information about the satisfiablity of 
$\Astruct(W,\sigma)$. The following is an easy generalization of well-known properties of Tseitin tautologies over~$\mathbb{Z}_2$.

\begin{lemma}\label{lem:encoding}
Let $H = (V(H),E(H))$ be a directed graph, let $W \subseteq V(H)$ be a subset of its vertices such that the subgraph induced by $W$ is connected, and let $\sigma \colon W \rightarrow G$ be a labelling of the vertices in $W$. An assignment $f : \partial(W) \rightarrow G$ extends to a solution to $\Astruct(W,\sigma)$ if and only if it satisfies $\Astruct(\partial(W),\sigma)$.
\end{lemma}

\begin{proof}
Since $\Astruct(\partial(W),\sigma)$ is the sum of the equations in $\Astruct(W,\sigma)$, the left-to-right direction is obvious.  For the opposite direction, let $f : \partial(W) \rightarrow G$ be an assignment satisfying $\Astruct(\partial(W),\sigma)$. Let us denote by $H'$ the graph induced by $W$. Observe that by assigning values to the variables in $\partial(W)$ according to $f$ and moving all the constants in the system $\Astruct(W,\sigma)$ to the right we obtain a system $\Astruct(H',\sigma')$ for some labelling $\sigma' \colon W \rightarrow G$ of the vertices in $H'$ which satisfies $\sum_{v \in W}\sigma'(v) = 0$. By Lemma~\ref{lem:satisfiability} there exists a solution $g$ to the system $\Astruct(H',\sigma')$. By extending $f$ with $g$ we obtain a solution to $\Astruct(W,\sigma)$.
\end{proof}

Hard Tseitin graph tautologies are usually based on graphs that are \emph{expanders}. For an undirected graph $H = (V(H),E(H))$ the \emph{expansion constant} is:
$$e(H) = \min \left\{{ {|\partial(W)|\over |W|} \colon W \subseteq V(H), |W| \leq {|V(H)| \over 2} }\right\}.$$
We call a family $\mathcal{H}$ of undirected graphs a \emph{family of expander graphs} if it is infinite and there exists a positive constant $e$ such that $e(H) \geq e$ for every graph $H$ in $\mathcal{H}$. For more information on expanders see e.g. the survey~\cite{Hoory06expandergraphs}. Here we only need the well known fact that expander families exist.

\begin{fact}
For every integer $l \geq 3$ there exists a family of connected $l$-regular undirected expander graphs $(H_n)_{n \geq 1}$, where the graph $H_n$ has $\Theta(n)$ vertices.
\end{fact}

Our proof strategy is now the following. We take a family $(\bar{H}_n)_{n \geq 1}$ of connected $3$-regular undirected expander graphs and show that for every sufficiently large integer $n$, one can specify edge directions in $\bar{H}_n$ obtaining a directed graph $H_n$, and choose an appropriate labelling $\sigma_n$ of the vertices of $H_n$ with elements of $G$, such that every Frege refutation of
  $\CNF(\Astruct(H_n,\sigma_n),\Bstruct(G,3))$ of depth $d$ has size at least $2^{n^\delta}$. To this end, following the lines of~\cite{Ben02}, we reduce the onto-pigeonhole principle, which states that there is no bijection between sets of size $m$ and $m+1$, to a Tseitin formula over a complete bipartite graph and further reduce the latter to the Tseitin formula over $H_n$. Let us begin with the second reduction.

\paragraph{Reducing Tseitin formulas over a complete bipartite graph.} We now define the graphs $H_n$ together with labellings $\sigma_n$ and show a reduction of a Tseitin formula over a complete bipartite graph to the Tseitin formula over $H_n$.

The following is a special case of Theorem 4.2 in~\cite{Ben02}.

\begin{theorem}[\cite{Ben02}]\label{thm:bubbles}
If $\mathcal{H}$ is a family of connected $3$-regular expander graphs, then there exists a positive constant $c$ such that for every graph $H$ in $\mathcal{H}$, the set of its vertices $V(H)$ can be partitioned into $V_1, \ldots, V_h$ where $h \geq c |V(H)|^{1/3}$ and:
\begin{itemize} \itemsep=0pt
\item For every $i \in [h]$ the subgraph induced by $V_i$ is connected.
\item For any $1 \leq i < j \leq h$, there is at least one edge incident to some vertex in $V_i$ and to some vertex in $V_j$.
\end{itemize}
\end{theorem}

Let us fix a family $(\bar{H}_n)_{n \geq 1}$ of connected $3$-regular undirected expander graphs, where the graph $\bar{H}_n$ has $\Theta(n)$ vertices. Let $c$ be the constant whose existence follows from the theorem above. 
Consider the graph $\bar{H}_n$, for some $n \geq ({5/ c})^3$, and take a partition of the set of vertices $V(\bar{H}_n)$ into at least $h \geq cn^{1/3} \geq 5$ subsets satisfying the conditions given in Theorem~\ref{thm:bubbles}. Let us call them \emph{bubbles}.  Without loss of generality we can assume that the number of bubbles is odd, i.e., $h=2m+1$, otherwise we remove the bubbles $V_{h-1}$ and $V_{h}$ and substitute them by a single bubble $V_{h-1} \cup V_{h}$. The set of bubbles is denoted $\mathcal{W}$.

Based on the undirected expander graph $\bar{H}_n$, together with the partition of the set of its vertices into $2m+1$ bubbles we define a directed graph $H_n$ and a labelling $\sigma_n$. The idea is to simulate the complete bipartite graph $K_{m,m+1}$. To this end, let us fix a partition of the set of bubbles into two disjoint sets: $\{V_1, \ldots, V_{m}  \}$ and $\{ W_1, \ldots, W_{m+1} \}$.  For each $i \in [m]$ and each $j \in [m+1]$ let us fix an edge $e_{i,j}$ incident to some vertex in $V_i$ and to some vertex in $W_j$. For future reference, let us say that we paint those edges red. We fix the direction of each of those edges from $W_j$ to $V_i$. The directions of the rest of the edges in the graph $H_n$ are set arbitrarily. Now, for each $i \in [m]$ we fix one vertex $v_i$ in the bubble $V_i$, paint it blue and label it with $-1$; similarly for each $j \in [m+1]$ we fix one vertex $v_j$ in the bubble $W_i$, paint it green and label it with~$1$. The rest of the vertices of the graph $H_n$ are labelled with~$0$. This finishes the definition of the directed graph $H_n$ and the labelling $\sigma_n$.

Observe that the Tseitin tautology $\Astruct(H_n,\sigma_n)$ is unsatisfiable. Indeed, the sum of all labels of the vertices of $H_n$ is $(m+1)\cdot 1 - m\cdot 1 = 1 \neq 0$. 

We now show that by assigning truth values to some variables in $\CNF(\Astruct(H_n,\sigma_n),\Bstruct(G,3))$ we obtain an encoding of $\CNF(\Astruct(K_{m,m+1}, \sigma),\Bstruct(G))$, where $K_{m,m+1}$ is a complete bipartite graph, and $\sigma$ is some labelling of its vertices with elements of the group $G$. Let us first make some general comments about partial assignments for instances of $\mathrm{LIN}(G)$ and variable substitutions for corresponding CNFs.  

Let $\Astruct$ be any system of linear equations over the group $G$. For a partial assignment $\rho$ which maps some of the variables of $\Astruct$ to elements of the group, there is a natural corresponding substitution of the variables in $\CNF(\Astruct,\Bstruct(G))$: if the partial assignment $\rho$ maps a variable $a$ to $g$, then the variable $X(a,g)$ is substituted by $1$ and the variables $X(a,g')$, for $g' \neq g$, are substituted by $0$; if the partial assignment $\rho$ leaves the value of some variable~$a$ unassigned, then for all the variables $X(a,g)$, the substitution is defined by the identity. It is not difficult to see that the result of applying this substitution to $\CNF(\Astruct,\Bstruct(G))$ is $\CNF(\Astruct|_\rho,\Bstruct(G))$, where $\Astruct|_\rho$ is the system of linear equations obtained by applying $\rho$ to the variables of $\Astruct$. For simplicity, we denote the above defined substitution by $\rho$. Hence, $\CNF(\Astruct,\Bstruct(G))|_\rho = \CNF(\Astruct|_\rho,\Bstruct(G))$.

Coming back to the graph $H_n$, let us consider a partial assignment $\rho$ which, for each of the bubbles $V \in \mathcal{W}$, maps the non-red edges in $\partial(V)$ to the group element $0$, and leaves the value of the rest of the edges in $H_n$ unassigned.
Observe that $\rho$ does not falsify any of the equations in $\Astruct(H_n,\sigma_n)$. Indeed, since every subgraph induced by a single bubble is connected, for every vertex $v$, the value of at least one variable that appears in the equation $\sum_{e \in \partial_{+}(v)} e - \sum_{e \in \partial_{-}(v)} e = \sigma_n(v)$ is left unassigned. Moreover, for every bubble $V \in \mathcal{W}$ the equation $\Astruct(\partial(V),\sigma_n)|_\rho$ says that the sum of the red edges in $\partial(V)$ is $1$. This is clear for the bubbles in $\{ W_1, \ldots, W_{m+1} \}$, and for the bubbles in $\{V_1, \ldots, V_{m}  \}$ one only needs to multiply the corresponding equations by $-1$.

Consider a complete bipartite graph $K_{m,m+1}$ with $m$ blue vertices, $m+1$ green vertices and a directed red edge from every green vertex to every blue vertex. Let the labelling $\sigma$ assign $-1$ to the blue vertices and $1$ to the green ones. The Tseitin tautology $\Astruct(K_{m,m+1}, \sigma)$ is up to renaming of variables the same as the set of equations in:
$$\bigcup_{i \in [m]}\Astruct(\partial(V_i),\sigma_n)|_\rho \ \cup \bigcup_{j \in [m+1]}\Astruct(\partial(W_j),\sigma_n)|_\rho \ = \bigcup_{V \in \mathcal{W}} \Astruct(\partial(V),\sigma_n)|_\rho.$$
Therefore, from now on we denote the above system of linear equations by $\Astruct(K_{m,m+1}, \sigma)$. Note that, for each of the vertices $v$ of the graph $K_{m,m+1}$, the corresponding equation says that the sum of the variables in $\partial(v)$ is $1$.

Let $r = \max(3,|G|)$. For $m \leq l$, an $r$-CNF $F$ over variables $X_1, \ldots, X_l$ is called an \emph{implicit encoding}~\cite{Ben02} of a propositional formula $\psi$ over variables $X_1, \ldots, X_m$ if the following holds: a truth assignment to the variables of $\psi$ satisfies $\psi$ if and only if it can be extended to a truth assignment to the variables of $F$ which satisfies $F$. The variables $X_{m+1}, \ldots, X_l$ are called \emph{auxiliary variables}.

It follows from Lemma~\ref{lem:encoding} that, for each bubble $V \in \mathcal{W}$, the formula $\CNF(\Astruct(V,\sigma_n),\Bstruct(G))$ is an implicit encoding of the formula $\CNF(\Astruct(\partial(V),\sigma_n),\Bstruct(G))$, with the set of auxiliary variables being the set of edges of the subgraph induced by $V$. Since on this set of edges the substitution $\rho$ is defined as the identity, it is not difficult to see that, for each bubble $V \in \mathcal{W}$, the substituted formula $\CNF(\Astruct(V,\sigma_n),\Bstruct(G))|_\rho$ is an implicit encoding of the substituted formula $\CNF(\Astruct(\partial(V),\sigma_n),\Bstruct(G))|_{\rho}$. Moreover, the sets of auxiliary variables in those implicit encodings are pairwise disjoint, hence the formula $$\bigcup_{V \in \mathcal{W}}\CNF(\Astruct(V,\sigma_n),\Bstruct(G,3))|_\rho = \CNF(\Astruct(H_n,\sigma_n),\Bstruct(G,3))|_\rho$$ is an implicit encoding of the formula $$\bigcup_{V \in \mathcal{W}}\CNF(\Astruct(\partial(V),\sigma_n),\Bstruct(G))|_{\rho} = \CNF(\Astruct(K_{m,m+1}, \sigma),\Bstruct(G)).$$

This way we have reduced an implicit encoding of a Tseitin formula over a complete bipartite graph $K_{m,m+1}$ to the Tseitin formula over the expander graph $H_n$, where $m > Cn^{1/3}$, and $C$ is a constant which does not depend on $n$.

\paragraph{Reducing the pigeonhole principle.} We now use the
technique for removing auxiliary variables without significantly
increasing the proof size introduced in~\cite{Ben02} to reduce the
\emph{onto-pigeonhole principle} formula $\OPHP(m,m+1)$, as defined
below, to the formula
$\CNF(\Astruct(H_n,\sigma_n),\Bstruct(G,3))|_\rho$.

For a positive integer $l$ and a set of variables $X_1, \ldots, X_l$, we denote by $\mathcal{U}({X_1, \ldots, X_l})$ the CNF which has a clause $\bigvee_{i \in [l]} X_i$, and for every $1 \leq i < i' \leq l$, a clause $\overline{X_i} \vee
  \overline{X_{i'}}$. For a complete bipartite graph $K_{l,l+1}$, the onto-pigeonhole principle $\OPHP(l,l+1)$ is the CNF which is the union of $\mathcal{U}(\partial(v))$ over the set of all vertices $v$ of the graph.

Let us consider the following substitution of the variables in $\CNF(\Astruct(K_{m,m+1}, \sigma),\Bstruct(G))$ and its implicit encoding
$\CNF(\Astruct(H_n,\sigma_n),\Bstruct(G,3))|_\rho$: for every red edge $e$ the variable $X(e,1)$ is substituted by $e$, the variable $X(e,0)$ is substituted by $\overline{e}$, and for every $g \in G$ such that $g \not \in \{0,1\}$, the variable $X(e,g)$ is substituted by $0$. On the auxiliary variables of the implicit encoding the substitution is defined by the identity. For simplicity, let us consider this substitution as an extension of the substitution $\rho$, and let us denote it by $\rho'$. Intuitively, the substituted formula $\CNF(\Astruct(K_{m,m+1}, \sigma),\Bstruct(G))|_{\rho'}$ encodes those assignments to the variables of $\Astruct(K_{m,m+1}, \sigma)$ that map each variable either to the group element $0$ or to the group element~$1$. Setting the truth value of the variable $e$ to $1$ corresponds to mapping $e$ to $1$, and setting it to $0$ corresponds to mapping $e$ to $0$.

Observe that for every vertex $v \in K_{m,m+1}$, we have $\mathcal{U}(\partial(v)) \models \CNF(\Astruct(v, \sigma),\Bstruct(G))|_{\rho'}$. Indeed, $\mathcal{U}(\partial(v))$ is satisfied if and only if exactly one of the variables in $\partial(v)$ is assigned a truth value~$1$. This truth assignment corresponds to mapping exactly one of the red edges incident to $v$ to the group element $1$ and mapping the rest of the red edges incident to $v$ to the identity element $0$. It is not difficult to see that such an assignment satisfies the equation~$\Astruct(v, \sigma)$.

For a CNF $F$ with variables $X_1, \ldots, X_l$,  by $\DNF(F)$ we denote the $l$-DNF formula which, for every truth assignment satisfying $F$, has an $l$-term representing this assignment, i.e., the unique $l$-term which is satisfied by this assignment and no other.

We now have all the ingredients necessary to remove the auxiliary
variables using the technique from~\cite{Ben02}. We remark that the
Frege system studied therein differs from the system considered in the
present paper. The formulas are formed from variables using negation
and disjunction only, and there is no introduction of conjunction
rule. However, it follows from the theorem of~\cite{ReckhowPhD} that
those two Frege systems polynomially simulate each other up to a
constant factor loss in depth.  Therefore, since the lower bound we
aim at is exponential and for all constant depths, we can apply the
technique from~\cite{Ben02}.

We have the following:
\begin{itemize} \itemsep=0pt
\item $\CNF(\Astruct(H_n,\sigma_n),\Bstruct(G,3))|_{\rho'}$ is an implicit encoding of $\CNF(\Astruct(K_{m,m+1}, \sigma),\Bstruct(G))|_{\rho'}$;
\item for every vertex $v \in K_{m,m+1}$, we have that $\mathcal{U}(\partial(v)) \models \CNF(\Astruct(v, \sigma),\Bstruct(G))|_{\rho'}$;
\item it follows from Lemma 5.7 in~\cite{Ben02} that, for every vertex $v \in K_{m,m+1}$, the size of a depth~$4$ Frege derivation of the formula
$\overline{\mathcal{U}(\partial(v))} \vee \DNF(\mathcal{U}(\partial(v)))$
is bounded by $O(m^2)$.
\end{itemize}
Hence, by Theorem 5.5 of~\cite{Ben02} if $\CNF(\Astruct(H_n,\sigma_n),\Bstruct(G))|_{\rho'}$ has a Frege refutation of depth~$d$ and size~$s$, then there exists a Frege refutation of $\OPHP(m,m+1) = \bigcup_{v \in V(K_{m,m+1})} \mathcal{U}(\partial(v))$ of depth $d+10$ and size at most polynomial in $m^4 s$, that is of size at most polynomial in~$n^{4/3}s$.

To complete the proof in the case of $G = \mathbb{Z}_q$ it suffices to
refer to the following theorem proved independently in~\cite{BIKPPW92}
and~\cite{KPW95} as an exponential improvement over~\cite{A88}.

\begin{theorem}[The Jewel Theorem of Proof Complexity \cite{BIKPPW92,KPW95,A88}]
For every integer $d$ there exists a constant $\delta$  such that for every sufficiently large integer $m$ every Frege refutation of $\OPHP(m,m+1)$ of depth $d$ has size at least~$2^{m^{\delta}}$.
\end{theorem}

It remains to show that thanks to the Fundamental Theorem of Finite Abelian Groups, the special case of $G = \mathbb{Z}_q$ implies Theorem~\ref{thm:lowerbound} in full generality.

\begin{lemma}\label{lem:abeliangroups}
  Let $G = \bigoplus_{q \in \mathcal{Q}}\mathbb{Z}_q$ be a finite
  Abelian group, and let $n$, $d$ and $s$ be positive integers.  If
  for some $q \in \mathcal{Q}$ there is an unsatisfiable instance
  $\Astruct$ of $3\mathrm{LIN}(\mathbb{Z}_q)$ with $n$ variables such
  that every Frege refutation of
  $\CNF(\Astruct,\Bstruct(\mathbb{Z}_q,3))$ of depth $d$ has size at least $s$,
  then there is an unsatisfiable instance $\Astruct'$ of
  $3\mathrm{LIN}(G)$ with $n$ variables such that every
  Frege refutation of $\CNF(\Astruct',\Bstruct(G,3))$ of depth $d$ has size at
  least $s$.
\end{lemma}

\begin{proof}
  Let $\Astruct$ be an unsatisfiable instance of
  $3\mathrm{LIN}(\mathbb{Z}_q)$ with $n$ variables, and assume that every
  Frege refutation of
  $\CNF(\Astruct,\Bstruct(\mathbb{Z}_q,3))$ of depth $d$ has size at least
  $s$. The instance $\Astruct$ is a system of linear equations over
  the group $\mathbb{Z}_q$. Since $\mathbb{Z}_q$ is a subgroup of $G$,
  we can think of the same system of linear equations as a system of
  linear equations over the group $G$. Let $\Astruct'$ be the
  corresponding instance of $3\mathrm{LIN}(G)$. It is not difficult to
  see that it is unsatisfiable. Moreover, every Frege
  refutation of $\CNF(\Astruct',\Bstruct(G,3))$ of depth $d$ has size at least
  $s$. Indeed, by applying a substitution $\rho$ which for every $a
  \in A'$ and every $g \in G \setminus \mathbb{Z}_q$ substitutes the
  variable $X(a,g)$ with $0$ and on all other variables is defined by
  identity, we transform a Frege refutation of
  $\CNF(\Astruct',\Bstruct(G,3))$ of depth $d$ to a Frege refutation of
  $\CNF(\Astruct,\Bstruct(\mathbb{Z}_q,3))$ of depth~$d$.
\end{proof}

\subsection{Lower bound for Polynomial Calculus}

The original motivation in \cite{BussGIP01} for defining the Tseitin
graph tautologies for Abelian groups beyond~$\mathbb{Z}_2$ was to
compare the strength of Polynomial Calculus over different
fields. Here we use their results with the different purpose of
getting lower bounds for Polynomial Calculus over the real-field
for all CSPs of unbounded width. Along the lines of the previous
section for bounded-depth Frege, this will be a consequence of
Theorem~\ref{thm:closurealgebraic}, Theorem~\ref{thm:bartokozik}, and
the following lower bound (for which we use the EQ encoding scheme).

\begin{theorem} \label{thm:lbpc} For every non-trivial finite Abelian
  group $G$ there exists a positive constant~$\delta$ and a family of unsatisfiable instances 
  $(\Astruct_n)_{n \geq 1}$ of $3\mathrm{LIN}(G)$, where $\Astruct_n$ has $\Theta(n)$
   variables and $\Theta(n)$ equations, such that for
  every sufficiently large $n$ every
  PC refutation over the reals of
  $\EQ(\mathbb{A}_n,\mathbb{B}(G,3))$ has degree at least~$\delta n$.
\end{theorem}

\noindent By the same argument as in the previous section,
Theorem~\ref{thm:lbpc} will follow from the special case for Abelian
groups of the form $\mathbb{Z}_m$ proved in \cite{BussGIP01}. Let us
note that the statement in \cite{BussGIP01} is made only for fields of
prime characteristic and for prime $m$, but the same proof goes
through for arbitrary fields whose characteristic does not divide~$m$.

Strictly speaking, the form of the Tseitin system of equations that we
defined in the previous section is slightly more general than the
original one from~\cite{BussGIP01}. In~\cite{BussGIP01}, the definition
starts with an undirected graph $H$ and, given a labelling $\sigma :
V(H) \rightarrow \mathbb{Z}_m$, the system of equations
$\hat{\mathbb{A}}(H,\sigma)$ over $\mathbb{Z}_m$ is defined as
follows:
\begin{itemize} \itemsep=0pt
\item there is a pair of variables $(u,v)$ and $(v,u)$ for each
edge $\{u,v\}$ in $E(H)$,
\item for every edge $\{u,v\}$ of $E(H)$ there is an equation
$(u,v) + (v,u) = 0$,
\item for every vertex $u$ of $V(H)$ there is an equation
$$
\sum_{v \in V(H): \atop \{u,v\} \in E(H) } (u,v) = \sigma(u).
$$
\end{itemize}
Let us see that $\hat{\mathbb{A}}(H,\sigma)$ is isomorphic to the
system $\mathbb{A}(\hat{H},\hat\sigma)$ for an appropriately defined
directed graph $\hat{H}$ and an appropriately defined labelling
$\hat\sigma$. The set of vertices $V(\hat{H})$ of $\hat{H}$ is
$V(H) \cup E(H)$. The set of edges
$E(\hat{H})$ of $\hat{H}$ has two directed edges $(u,e)$ and
$(v,e)$ for each undirected edge $e = \{u,v\}$ of $H$. The
labelling $\hat\sigma$ is the extension of $\sigma$ to $V(\hat{H})
\supseteq V(H)$ defined by $\hat\sigma(e) = 0$ for each $e \in
E(H)$. It is not hard to see that the mapping 
$(u,e) \mapsto (u,v)$, for $e = \{u,v\}$, is an isomorphism from
$\mathbb{A}(\hat{H},\hat\sigma)$ to $\hat{\mathbb{A}}(H,\sigma)$. This
justifies the claim that the definition of the Tseitin system from the
previous section is a generalization of the definition in
\cite{BussGIP01}. Another sense in which the definition of the Tseitin
system from the previous section is more general is that the original
definition requires $m$ in $\mathbb{Z}_m$ to be a prime number;
however, going through their proof it is readily seen that this is not
essential. Finally, the original definition also requires the
condition that the sum of the labels $\sigma(u)$ is $1$ (mod $m$), but
again this is not essential in their proof as long as the sum is
non-zero.

Let $\mathbb{B}'(\mathbb{Z}_m,3)$ be the template $\mathbb{B}(\mathbb{Z}_m,3)$ extended with the binary relation $\{ (g,g') \in \mathbb{Z}_m^2 : g + g' = 0 \}$, and let $\EQ'$ denote the modification of the
encoding scheme $\EQ$ in which each twin variable $\bar{X}(a,b)$ is
replaced by $1-X(a,b)$. It turns out that the system of polynomial
equations $\EQ'(\hat{\mathbb{A}}(H_n,\sigma_n),\mathbb{B}'(\mathbb{Z}_m,3))$
for a fixed family of $3$-regular expander graphs $(H_n)_{n \geq 1}$ and a
labelling $\sigma_n : V(H_n) \rightarrow \mathbb{Z}_m$ of total sum $1$
mod $m$ is literally the same
as the system of polynomial equations that \cite{BussGIP01} calls BTS$_{n,m}$. Note that BTS$_{n,m}$ has $\Theta(n)$ variables. 
We have the following:

\begin{theorem}[see Corollary 21 in
  \cite{BussGIP01}] \label{thm:BussGIP01} For every integer $m \geq 2$
  and every field $F$ of a characteristic that does not divide $m$
  there exists a positive $\delta$ such that for every
  sufficiently large $n$ every PC refutation over $F$ of
  BTS$_{n,m}$ has degree at least $\delta n$.
\end{theorem}

\noindent This gives us a family of instances of $\mathbb{B}'(\mathbb{Z}_m,3)$ that are hard for Polynomial Calculus over the real-field.
Since the template $\mathbb{B}'(\mathbb{Z}_m,3)$ is pp-definable in $\mathbb{B}(\mathbb{Z}_m,3)$, Theorem~\ref{thm:closurealgebraic} implies an existence of such a family for $3\mathrm{LIN}(\mathbb{Z}_m)$.

In order to complete the proof of Theorem~\ref{thm:lbpc}
from~Theorem~\ref{thm:BussGIP01} it suffices to invoke a version of
Lemma~\ref{lem:abeliangroups} for Polynomial Calculus, whose statement
and proof are virtually identical to those of
Lemma~\ref{lem:abeliangroups}, and are thus omitted.

\subsection{Lower bound for Sums-of-Squares}

In the case of Sums-of-Squares, similarly as for Polynomial Calculus, we do not need to adapt
an existing lower bound proof from the literature for $\mathbb{Z}_2$
to all finite Abelian groups because this was already done. The lower
bound that we need to complete the proof of Theorem~\ref{thm:boundeddepth} is
the following:

\begin{theorem}[\cite{Chan2016}] \label{thm:chan}
  For every non-trivial finite Abelian group $G$ there exists a
  positive $\delta$ and a family of unsatisfiable instances 
  $(\Astruct_n)_{n \geq 1}$ of $3\mathrm{LIN}(G)$, where $\Astruct_n$ has $\Theta(n)$
   variables and $\Theta(n)$ equations, such that for
  every sufficiently large $n$ every SOS refutation of
  $\EQ(\Astruct_n,\Bstruct(G,3))$ has degree at least~$\delta n$.
\end{theorem}

The exact statement that we are referring to is Theorem~G.8 from
Appendix~G in~\cite{Chan2016}. In order to be able to state the
theorem and compare it to the statement of Theorem~\ref{thm:chan} we
need to introduce some definitions.

Let $G$ be a finite Abelian group and let $C$ be a subgroup of $G^k$,
where $k \geq 3$. The problem Additive-CSP$(C)$, as defined
in~\cite{Chan2016}, is the constraint satisfaction problem that has
constraint relations of the form $\{ (c_1,\ldots,c_k) :
(c_1-b_1,\ldots,c_k-b_k) \in C \}$, for all $(b_1,\ldots,b_k) \in
G^k$. Note that if the set of variables is $V$, then the set of all
possible constraints can be identified with the set $V^k \times
G^k$. The instances are presented as distributions $\pi$ over $V^k
\times G^k$. This amounts to assigning weights to the constraints.
The \emph{value} of an instance is the maximum over all
assignments of values to variables of the probability that a random
constraint chosen from~$\pi$ is satisfied by the assignment. We say
that $C \subseteq G^k$ is \emph{balanced pairwise independent} if for
every pair $i,j \in [k]$ with $i \not= j$, and every two elements $a,b
\in G$, the number of $k$-tuples $(c_1,\ldots,c_k)$ from $C$ such that
$c_i = a$ and $c_j = b$ is $|C|/|G|^2$. For example, any $C$ of the
form $\{ (c_1,\ldots,c_k) : c_1 + \cdots + c_k = 0 \}$ is balanced
pairwise independent, and it is a subgroup of $G^k$. Chan's
Theorem~G.8 in~\cite{Chan2016} states that if $C$ is any balanced
pairwise independent subgroup of $G^k$ and $\epsilon$ is an arbitrary
positive constant, then for every sufficiently large $n$, there is an
instance $M$ of Additive-CSP$(C)$ with $n$ variables, whose value is
bounded by $|C|/|G|^k + \epsilon$, and that has a \emph{Lasserre
  solution of value $1$ for $c n$ rounds}, where $c =
c_{G,k,\epsilon}$ is a constant that depends only on the group $G$,
the arity $k$, and the tolerance parameter $\epsilon$. Moreover, it follows from the proof in~\cite{Chan2016} (see Theorem~G.7) that the instance $M$ can be chosen to have $e_{G,k,\epsilon}n$ constraints, where $e_{G,k,\epsilon}$ is a constant that depends only on the group $G$,
the arity $k$, and the parameter $\epsilon$. We discuss what
a Lasserre solution is and how it relates to SOS proofs.

Before we do that we fix some of the parameters. We want to build an
unsatisfiable instance $\mathbb{A}$, and we do so by choosing the
parameters to make the value of $M$ in Chan's Theorem strictly
smaller than $1$. Fix $k = 3$ and $C = \{ (c_1,c_2,c_3) : c_1 + c_2 +
c_3 = 0 \}$, and take $\epsilon = 1/4$.  Then the value of the
instance $M$ is bounded by $|C|/|G|^3 + \epsilon = 1/|G| + 1/4 \leq
1/2 + 1/4 < 1$. This means that the collection of constraints of $M$
that have non-zero probability in $\pi$ is unsatisfiable; i.e., not
all constraints can be satisfied at the same time by a single
assignment. Thus, our unsatisfiable instance~$\mathbb{A}$ will just be
the set of all constraints with non-zero probability in $\pi$. Now we are
ready to define what a Lasserre solution of value~$1$ is.

According to Definition~G.3 from Appendix~G in~\cite{Chan2016}, a
Lasserre solution of value~$1$ for $t$ rounds is a collection $u = \{
u_f : f \in G^S,\; S \subseteq V,\; |S| \leq t \}$ of vectors in
Euclidean space~$\mathbb{R}^d$, of some finite dimension $d$, such
that for every $S \subseteq V$ with $|S| \leq 2t$ there exists a
probability distribution $\mu_S$ over $G^S$ with the following
properties: for every $R,S,T \subseteq V$ with $|S|,|T| \leq t$ and $R
= S \cup T$, and every $f \in G^S$ and $g \in G^T$, it holds that
\begin{align}
& \Pr_{h \in \mu_{R}}[\; h|_S = f \text{ and } h|_T = g \;] = \langle u_f,u_g \rangle, \label{eqn:psd}
\end{align}
and for every constraint with variables $S$ in the support of
$\pi$ and every $f \in G^S$ that does not satisfy this
constraint we have
\begin{equation}
\Pr_{h \in \mu_S}[\; h = f \;] = 0. \label{eqn:sup}
\end{equation}

At this point we have all the necessary material to argue that
$\mathbb{A}$, or more precisely, $\EQ(\mathbb{A},\mathbb{B}(G,3))$,
does not have SOS refutations of degree $\delta n$, where $\delta =
2c_{G,k,\epsilon}$. Let $\EQ'$ be the result of replacing each twin variable
$\bar{X}(a,b)$ in $\EQ(\mathbb{A},\mathbb{B}(G,3))$ by $1-X(a,b)$. By the remarks at the end of Section~\ref{sec:algproofs}, it suffices to show that $\EQ'$ does not have SOS refutations of degree $\delta n$ for
the definition of Sums-of-Squares without twin variables. Assume, for the sake of
contradiction, that $\EQ'$ \emph{does} have such an SOS refutation
of degree at most~$2t$, where~$t := c_{G,k,\epsilon} n$ is the number
of rounds for which there exists a Lasserre solution of value~$1$ for the instance
$M$. 
The refutation has the form
\begin{equation}
\sum_{i=1}^r P_i \cdot S_i = -1
\end{equation}
where $P_1,\ldots,P_r$ are polynomials that either come from $\EQ'$, or they are axiom polynomials from the
lists~\eqref{eqn:axioms} and~\eqref{eqn:axiomssemi} without twin variables, or they are
squares, $S_1,\ldots,S_r$ are arbitrary or square
polynomials without twin variables as appropriate (i.e., arbitrary if the $P_i$ they multiply
come from an equation, and squares if the $P_i$ they multiply come
from an inequality), and the total degree of each product $P_i \cdot S_i$ is at most $2t$.
Multiplications by $X$ and $1-X$ can be simulated by multiplications by their squares, thanks to the axioms $X^2 - X = 0$ from~\eqref{eqn:axioms}, so we can assume that the refutation has the form
\begin{equation}
\sum_{i=1}^m P_i \cdot S_i + \sum_{i=1}^{\ell} Q_i^2 = -1, \label{eqn:proof}
\end{equation}
where $P_1,\ldots,P_m$ are polynomials that either come from $\EQ'$, or they are one of the axiom polynomials of the form $X^2 - X$ from~\eqref{eqn:axioms}, and $S_1,\ldots,S_m, Q_1, \ldots, Q_{\ell}$ are arbitrary polynomials.

Recall
that the variables of $\EQ'$ have the form $X(a,b)$ where
$(a,b) \in V \times G$. We say that the element $a$ is mentioned in
$X(a,b)$, and that it is mentioned in any monomial that contains this
variable. Now we define a linear functional $E : \mathcal{P}_{2t}
\rightarrow \mathbb{R}$, where $\mathcal{P}_{2t}$ denotes the vector
space of polynomials of degree at most $2t$ on the $X(a,b)$-variables, as
follows. 

For each monomial $M$ of degree at most $2t$ on the $X(a,b)$-variables,
with all mentioned elements in $S \subseteq V$, define
\begin{equation}
E(M) = \mathbb{E}_{h \in \mu_S} [\; h(M) \;], \label{eqn:defE}
\end{equation}
where the notation $h(M)$ stands for the evaluation of the monomial
$M$ by the partial assignment given by $h$; i.e., all variables
$X(a,h(a))$ with $a \in S$ are set to $1$, all variables $X(a,b)$ with
$a \in S$ and $b \not= h(a)$ are set to $0$, and all other variables
are left unset. Note that~\eqref{eqn:psd} ensures
that~\eqref{eqn:defE} is a well-defined quantity that does not depend
on the choice of $S$, as long as $S$ contains all the elements that
are mentioned in $M$. 
Once $E$ is defined for all monomials of
degree at most $2t$, we extend it to $\mathcal{P}_{2t}$ by linearity.

The final step in the argument is to show that $E$ evaluates the
left-hand side in~\eqref{eqn:proof} to some non-negative quantity;
this will imply that the identity in~\eqref{eqn:proof} does not hold,
and finish the proof.  In order to prove this, the following matrix
$(A_{M,N})_{M,N}$ will be instrumental. The indices are monomials $M$
of degree at most $t$ on the $X(a,b)$-variables. The entry $A_{M,N}$ of $A$ is defined to be
$E(MN)$.  For later use, observe that if $S$ denotes the set of
elements that are mentioned in $M$ and there exists $f \in G^S$ such
that $f(M) = 1$, then this partial assignment $f$ with domain $S$ is
uniquely determined by $M$. We let $f_M \in G^S \cup \{ \bot \}$
denote this unique partial assignment $f$ that makes $f(M) = 1$, when
it exists, or the default value $\bot$ when it does not exist. We
argue that Equation~\eqref{eqn:psd} ensures that $A$ is a positive
semi-definite matrix. First, extend the collection of vectors $u$ to a
new collection of vectors $u^* = \{ u^*_f : f \in G^S \cup \{ \bot
\},\; S \subseteq V,\; |S| \leq t \}$ by defining $u^*_f = u_f$ for $f
\in G^S$, and $u^*_f = 0$ for $f = \bot$. Fix indices $M$ and $N$, let
$S$ and $T$ be sets of elements mentioned in $M$ and $N$,
respectively. Let $R = S \cup T$. Then $E(MN)$, according to its
definition~\eqref{eqn:defE}, is the probability of the event that $h
\in \mu_R$ makes $h(MN) = 1$, or equivalently, that $h \in \mu_R$
makes $h|_S(M) = 1$ and $h|_T(N) = 1$, or equivalently, that $h \in
\mu_R$ makes $h|_S = f_M$ and $h|_T = g_N$. Thus,
equation~\eqref{eqn:psd} and the definition of the extended collection of
vectors $u^*$ ensures that $A_{M,N} = \langle u^*_{f_M},u^*_{g_N}
\rangle$ and hence $A$ is a Gram matrix. Thus $A$ is positive
semi-definite.

Now we use the positive semi-definiteness of $A$ to show that, for
squares $Q_i^2$ in~\eqref{eqn:proof}, we have $E(Q_i^2) \geq 0$.
Indeed, if $Q_i = \sum_M a_M M$ where the sum extends over all
monomials of degree at most $t$ and $a = (a_M)_M$ is the corresponding
vector of coefficients, then
\begin{equation}
E(Q_i^2) = E\Big(\sum_M \sum_N a_M a_N M N\Big) = 
\sum_M \sum_N a_M a_N A_{M,N}
= a^{\mathrm{T}} A a, 
\label{eqn:possquarematrix}
\end{equation}
which is non-negative because $A$ is a positive semi-definite matrix.

For terms
in~\eqref{eqn:proof} that are liftings of equations from $\EQ'$,
the evaluation through $E$ is~$0$. This is clear for equations of type 2, since every monomial which contains a pair of variables $X(a,b)$ and $X(a,b')$, for $b \neq b'$ evaluates to~$0$ by~\eqref{eqn:defE}. For the same reason if we take any equation of type 1 in $\EQ'$, i.e, $\prod_{b \in G} (1 - X(a,b)) = 0$, for some $a \in V$, and an arbitrary monomial $M$ on the $X(a,b)$-variables such that $P \cdot M$ has a total degree at most $2t$, it holds that 
\begin{equation}
E\Big(\prod_{b \in G} (1 - X(a,b)) \cdot M\Big) = 
E(M) - \sum_{b \in G} E(X(a,b) \cdot M).
\end{equation}
By~\eqref{eqn:defE} again, we have $E(M) = \sum_{b \in G} E(X(a,b) \cdot M)$
and the right-hand side vanishes too.
Finally, liftings of equations of type $3$ from $\EQ'$ evaluate to~$0$ thanks to
equation~\eqref{eqn:sup}. 

For terms in~\eqref{eqn:proof} that are liftings of axioms
in~\eqref{eqn:axioms}, the evaluation
through $E$ is also~$0$ since the partial assignments on the
$X(a,b)$-variables take Boolean values. All in all, the evaluation of the left-hand
side of~\eqref{eqn:proof} through $E$ is non-negative, and the proof
is complete.

\section{Upper bounds in Lov\'asz-Schrijver} \label{sec:upperboundinLS}

In this section we show that all unsatisfiable instances of
3LIN($\mathbb{Z}_2$) have LS refutations of degree $6$ and size
polynomial in the number of variables.  Indeed, the argument to get
polynomial-size upper bound in constant degree works equally well for
3LIN($\mathbb{Z}_p$), when $p$ is prime, with some inessential
complications~\cite{DBLP:journals/corr/Atserias15}. We focus on $\mathbb{Z}_2$ for simplicity.

\paragraph{Initial remarks on the encoding.} We identify the elements of the two-element field $\mathbb{Z}_2$ with
$\{0,1\}$.  Let $\mathbb{E}$ be an instance of $k$LIN($\mathbb{Z}_2$)
with $n$ variables. In the encoding
$\INEQ(\mathbb{E},\mathbb{B}(\mathbb{Z}_2,k))$ of $\mathbb{E}$ as a
system of linear inequalities, there are four variables
$X(a,0)$, $X(a,1)$ ,$\bar{X}(a,0)$, $\bar{X}(a,1)$ for each variable $a$ in
$\mathbb{E}$. Note, however, that they are restricted to satisfy
$X(a,0) = \bar{X}(a,1)$ and $\bar{X}(a,0) = X(a,1)$ by the inequality
$X(a,0) + X(a,1) - 1 \geq 0$ from $\INEQ$ and the axiom equations
in~\eqref{eqn:axioms}, which in this case read $X(a,0)^2 - X(a,0) =
X(a,1)^2 - X(a,1) = 0$ and $X(a,0) + \bar{X}(a,0) - 1 = X(a,1) +
\bar{X}(a,1) - 1 = 0$. Consequently, in the following we will ignore
the variables of the type $X(a,0)$ and their twins and keep only the
variables $X(a,1)$ and $\bar{X}(a,1)$. In order to simplify the
notation even further, we will assume that the variables of
$\mathbb{E}$ are called $X_1,\ldots,X_n$, and that those of $\INEQ$
are called $X_1,\ldots,X_n$ and $\bar{X}_1,\ldots,\bar{X}_n$.

We interpret the variables $X_1,\ldots,X_n$ as ranging over
$\mathbb{Z}_2$ or $\mathbb{Q}$ depending on the context. Let $E$ be an
equation of $\mathbb{E}$, say $E : a_1 X_1 + \cdots + a_n X_n = b$,
where $a_1,\ldots,a_n \in \mathbb{Z}_2$ and $b \in \mathbb{Z}_2$. Without loss of generality we can assume
that there are exactly $k$ many $a_i$'s that are $1$. In $\INEQ$, the
encoding of the constraint represented by this equation is given by
the following inequalities:
\begin{align*}
  & \sum_{i \in T} \bar{X}_i + \sum_{i\in I\setminus T} X_i - 1 \geq 0 &
  \text{ for all } T \subseteq I \text{ such that } |T| \equiv
  1-b\!\!\mod 2,
\end{align*}
where $I = \{ i \in [n] : a_i \not= 0 \}$. Note that $|I|=k$. We write
$\mathcal{S}(E)$ to denote this set of inequalities; it has exactly
$2^{k-1}$ many inequalities, and all of them are satisfied in $\Qbb$
by a $\{0,1\}$-assignment if and only if the equation $E$ is satisfied
in $\mathbb{Z}_2$ by the same assignment.  Let
$\mathcal{S}(\mathbb{E})$ be the union of all $\mathcal{S}(E)$ as $E$
ranges over the equations in $\mathbb{E}$. Observe that, except for
the small detail that only half of the variables are used, $\INEQ$ is
basically the same as~$\mathcal{S}(\mathbb{E})$.

\paragraph{Some technical lemmas.} For every linear form $L(X_1,\ldots,X_n) = \sum_{i=1}^n a_i X_i$ with
rational coefficients $a_1,\ldots,a_n$ and every integer $c$, let
$D_c(L)$ be the quadratic polynomial $(L - c) (L - c + 1)$. In
words, the inequality $D_c(L) \geq 0$ states that $L$ does not fall in
the open interval $(c-1,c)$. Such statements have short proofs of low
degree:
 
\begin{lemma}[\cite{GrigorievHirschPasechnik2002}] \label{lem:GHP} 
  For every integer $c$ and for every linear form $L(X_1,\ldots,X_n) =
  \sum_{i=1}^n a_i X_i$ with integer coefficients $a_1,\ldots,a_n$,
  there is an LS proof of the inequality $D_c(L) \geq 0$ (from nothing) of
  degree at most $3$ and size polynomial in $\max\{|a_i| :
  i=1,\ldots,n\}$, $|c|$ and $n$.
\end{lemma}

In the following, for $I \subseteq [n]$ and $T \subseteq I$, let
$M^I_T(X_1,\ldots,X_n) := \prod_{i\in T} X_i \prod_{i \in I\setminus
  T} \bar{X}$. As usual, $M^I_\emptyset(X_1,\ldots,X_n) = 1$. Such
polynomials are called \emph{extended monomials}.

\begin{lemma} \label{lem:fullsum}
For every $I \subseteq [n]$, there is an LS proof of $\sum_{T
  \subseteq I} M^I_T - 1 = 0$ (from nothing) of degree $|I|$ and size
polynomial in $2^{|I|}$.
\end{lemma}

\begin{proof}
For simplicity, let $q = |I|$ and assume $I = \{1,\ldots,q\}$. We
build the proof inductively on $q$. For $q = 0$, what we need is
trivial since the left-hand side is $0$. Assume now $q \in
\{1,\ldots,n\}$ and that we have $\sum_{T \subseteq [q-1]} M^I_T - 1 =
0$.  Multiply this once by $X_q$ and once by~$\bar{X}_q$. Adding the
results we get $\sum_{T \subseteq [q]} M^I_T - X_q - \bar{X}_q = 0$,
from which $\sum_{T \subseteq [q]} M^I_T - 1 = 0$ follows from adding
the axiom $X_q + \bar{X}_q - 1 = 0$ to it. The size is exponential in
$|I|$ because the inductive step is used twice.
\end{proof}

The next lemma is as technical as useful.

\begin{lemma} \label{lem:thing}
  Let $T \subseteq I \subseteq [n]$. Then there is an LS proof of
  $\big({\sum_{i \in I} X_i - |T|}\big) M^I_T = 0$ (from nothing)
  of degree at most $|I|+1$ and size linear in $|I|$.
\end{lemma}

\begin{proof} Write $M$ for $M^I_T$. For every $i \in I\setminus T$,
  using $X_i \bar{X}_i = 0$ we get $X_i M = 0$.  For every $i
  \in T$, using $X_i^2 - X_i = 0$ we get $X_i M = M$.  Adding up we
  get $\sum_{i \in I} X_i M = |T| M$.
\end{proof}

\paragraph{Simulating Gaussian elimination.} We use these lemmas to prove the main result of this section.

\begin{theorem} \label{thm:maintwoelementfield} Let $\mathbb{E}$ be an
  instance of $3\mathrm{LIN}(\mathbb{Z}_2)$ with $n$ variables and $m$
  equations.  If $\mathbb{E}$ is unsatisfiable, then
  $\mathcal{S}(\mathbb{E})$ has an LS refutation of degree $6$ and
  size polynomial in $n$ and $m$.
\end{theorem}

\begin{proof}
  Write $\mathbb{E}$ in matrix form $A x = b$, where $x$ is a column
  vector of $n$ variables, $A$ is a matrix in $\mathbb{Z}_2^{m \times
    n}$, and $b$ is a vector in $\mathbb{Z}_2^m$.  Let
  $a_{j,1},\ldots,a_{j,n}$ be the $j$-th row of $A$, so the $j$-th
  equation of $\mathbb{E}$ is $E_j : a_{j,1} X_1 + \cdots + a_{j,n}
  X_n = b_j$. Assume $\mathbb{E}$ is unsatisfiable over
  $\mathbb{Z}_2$.  Then $b$ cannot be expressed as a
  $\mathbb{Z}_2$-linear combination of the columns of $A$, so the
  $\mathbb{Z}_2$-rank of the matrix $[\;A \;|\; b\;]$ exceeds the
  $\mathbb{Z}_2$-rank of $A$. Since the rank of $A$ is at most $n$,
  this means that there exists a subset of at most $n$ rows $J$ such
  that, with arithmetic in $\mathbb{Z}_2$, we have $\sum_{j \in J}
  a_{j,i} = 0$ for every $i \in [n]$, and at the same time $\sum_{j
    \in J} b_j = 1$. In order to simplify the notation, we assume
  without loss of generality that $J = \{1,\ldots,|J|\}$.

  For every $k \in \{0,\ldots,|J|\}$, define the linear form
  $$
  L_k(X_1,\ldots,X_n) := \frac{1}{2} \left({ \sum_{j=1}^k \sum_{i=1}^n
      a_{j,i} X_i + \sum_{j=k+1}^{|J|} b_j}\right).
  $$
  In this definition of $L_k$, the coefficients $a_{j,i}$ and $b_j$
  are interpreted as rationals.  We provide proofs of $D_c(L_k) \geq
  0$ for every $c \in R_k := \{0,\ldots,(k+1) n \}$ by reverse
  induction on $k \in \{0,\ldots,|J|\}$.

  The base case $k = |J|$ is a special case of Lemma~\ref{lem:GHP}. To
  see why note that the condition $\sum_{j \in J} a_{j,i} = 0$ over
  $\mathbb{Z}_2$ means that, if arithmetic were done in $\Qbb$, then
  $\sum_{j \in J} a_{j,i}$ is an even natural number. But then all the
  coefficients of
  $$
  L_{|J|}(X_1,\ldots,X_n) = \frac{1}{2} \sum_{j=1}^{|J|} \sum_{i=1}^n
  a_{j,i} X_i = \sum_{i=1}^n \left({\frac{1}{2} \sum_{j=1}^{|J|}
      a_{j,i}}\right) X_i
  $$ 
  are integers. Hence Lemma~\ref{lem:GHP} applies.

  Suppose now that $0 \leq k \leq |J|-1$ and that we have a proof of
  $D_d(L_{k+1}) \geq 0$ available for every $d \in R_{k+1}$.  Fix $c
  \in R_k$; our immediate goal is to give a proof of $D_c(L_k) \geq
  0$.  As $k$ is fixed, write $L$ in place of $L_{k+1}$, and let the
  $(k+1)$-st equation $E_{k+1}$ be written as $\sum_{i \in I} X_i =
  b$, where $I = \{ i \in [n] : a_{k+1,i} = 1 \}$. Note that $L = L_k
  + \ell/2$ where $\ell := -b + \sum_{i\in I} X_i$. Fix
  $T \subseteq I$ such that $|T| \equiv b\!\!\mod 2$, and let $d = c
  + (t-b)/2$ where $t = |T|$. Note that $d \in R_{k+1}$ as $c \in
  R_k$ and $0 \leq t \leq n$ and $0 \leq b \leq 1$ are such that $t -
  b$ is even. Multiplying $D_{d}(L) \geq 0$ by the extended monomial
  $M^I_T$ we get
  $
    (L - d) (L - d + 1) M^I_T \geq 0.
  $
  Replacing $L = L_k + \ell/2$ in the factor $(L - d)$
  and recalling $d = c + (t-b)/2$, this inequality can be written
  as
  \begin{equation}
    (L_k - c) (L - d + 1) M^I_T
    + (L - d + 1) \textstyle{\frac{1}{2}} A \geq 0
  \label{eq:inter1}
  \end{equation}
  where $A := (\ell + b - t) M^I_T$. By Lemma~\ref{lem:thing} we
  have a proof of $A = 0$, and hence of $(L-d+1) A/2 = 0$. Composing with~\eqref{eq:inter1} we get a
  proof of
  $
    (L_k - c) (L - d + 1) M^I_T \geq 0.
  $
  The same argument applied to the factor $(L - d + 1)$ of
  this inequality gives
  $
    (L_k - c) (L_k - c + 1) M^I_T \geq 0.
  $
  This is precisely $D_c(L_k) M^I_T \geq 0$. Adding up over all
  $T \subseteq I$ with $|T| \equiv b\!\!\mod 2$ we get
  \begin{equation}
    D_c(L_k) \sum_{T \subseteq I \atop |T| \equiv b}
    M^I_T \geq 0. \label{eq:last1}
  \end{equation}
  
  Now note that for each $T \subseteq I$ such that $|T| \equiv 1 - b
  \mod 2$, the inequality $-M^I_T \geq 0$ is the multiplicative
  encoding of one of the clauses in $\mathcal{S}(E)$. Thus, by
  Lemma~\ref{lem:encodings}, we get constant-size proofs of $-M^I_{T}
  \geq 0$, and hence of $M^I_T = 0$, for every $T \subseteq I$ such
  that $|T| \equiv 1-b\!\!\mod 2$.  But then also of $D_c(L_k) M^I_T =
  0$ for every such $T$. Adding up and composing with~\eqref{eq:last1}
  we get
  \begin{equation}
  D_c(L_k) \sum_{T \subseteq I} M^I_T \geq 0. \label{eqn:thisone}
  \end{equation}
  From Lemma~\ref{lem:fullsum} we get $1 - \sum_{T \subseteq I} M^I_T
  = 0$, and hence $D_c(L_k) - D_c(L_k) \sum_{T \subseteq I} M^I_T
  \geq 0$, from which $D_c(L_k) \geq 0$ follows from addition
  with~\eqref{eqn:thisone}.

  We proved $D_c(L_0) \geq 0$ for every $c \in R_0 =
  \{0,\ldots, n \}$. Recall now that $\sum_{j=1}^{|J|} b_j$ is odd,
  say $2q+1$, and at most $n$. In particular $q+1$ belongs to $R_0$
  and $L_0 = q+1/2$. Thus we have a proof of $D_{q+1}(L_0) \geq 0$
  where $D_{q+1}(L_0) = - (1/2)(1/2) = - 1/4$.  Multiplying by $4$ we
  get the contradiction $-1 \geq 0$.
\end{proof}

\section{Applications to k-coloring} \label{sec:threecoloring}

We illustrate the power of the general method of reductions for CSPs
by applying it to the graph $k$-coloring problem for $k \geq 3$. This
will allow us to rederive one of the results in
\cite{DBLP:conf/coco/LauriaN17}, as well as answer one of their open
problems. 

\subsection{Blackbox application to k-coloring}
An undirected graph $G$ is $k$-colorable if and only if it has a
homomorphism into the $k$-clique graph $\mathbb{K}_k = ([k],[k]^2
\setminus \{(b,b) : b \in [k]\})$. Thus the $k$-coloring problem is a
special case of the CSP of the template $\mathbb{K}_k$, which we
abbreviate by $k\text{-COLOR}$. We say that it is a special case and
not exactly the same problem because the inputs to $k\text{-COLOR}$
need not be undirected graphs; in full generality they are directed
graphs that allow loops. Note, however, that a directed graph that has
loops would never have a homomorphism into the template, and that a
loopless directed graph has a homomorphism into the template if and
only if the underlying undirected graph that ignores the directions of
the edges is $k$-colorable. Thus, for all practical purposes, the two
problems are the same, and proof complexity lower bounds for one
version of the problem will give proof complexity lower bounds for the
other. We discuss this in due time; for now we focus on the proof
complexity of the CSP with template~$\mathbb{K}_k$.

It is well-known that $\mathbb{K}_k$ is a template
of unbounded width for each $k \geq 3$ (see e.g.~\cite{FV98}). As a
consequence of our main result we get the following:

\begin{corollary} \label{cor:coloring} For every integer $k \geq 3$,
  there exist families $(G_n)_{n \geq 1}$ of unsatisfiable instances of
  \mbox{$k$\textrm{-COLOR}}, where $G_n$ has $\Theta(n)$ vertices and
  $\Theta(n)$ edges,
  such that for
 every positive integer $d$ there exists a positive~$\epsilon$ such
  that, for any local encoding scheme $E$ of the appropriate type, and
  any sufficiently large~$n$, 
  the following hold:
  \begin{enumerate} \itemsep=0pt
  \item every resolution refutation of $E(G_n)$ has width at least
    $\epsilon n$ and size at least $2^{\epsilon n}$,
  \item every Frege refutation of $E(G_n)$ of depth $d$ has size at
    least $2^{n^{\epsilon}}$,
  \item every PC refutation over the reals of $E(G_n)$ has degree at
    least $\epsilon n$ and size at least $2^{\epsilon n}$,
  \item every SOS refutation of $E(G_n)$ has degree at least $\epsilon
    n$.
\end{enumerate}
\end{corollary}

\noindent Let us note that the 
size lower bound claim for Polynomial Calculus follows from its degree lower bound
claim in conjunction with the size-degree tradeoff
\cite{ImpagliazzoPudlakSgall1999}. Also, the claim for resolution follows
from the claim for Polynomial Calculus and the fact that the latter efficiently simulates
resolution (c.f., Lemma~\ref{lem:simulation} in
Section~\ref{sec:generalproofcomplexityfacts}); it does not follow
directly from our main dichotomy theorem.

Next we consider the specific encoding scheme used in
\cite{DBLP:conf/coco/LauriaN17} and the issue of directed vs. undirected
graphs.  Let $G$ be a directed graph, with vertex-set $V$ and edge-set
$E \subseteq V^2$, let~$k$ be a positive integer, and consider the
following system of polynomial equations:
\begin{enumerate} \itemsep=0pt
\item $\sum_{b \in [k]} X(u,b) - 1 = 0$ for each $u \in V$,
\item $X(u,b) X(u,c) = 0$ for each $u \in V$ and $b,c \in [k]$ with $b \not= c$,
\item $X(u,b) X(v,b) = 0$ for each $u, v \in V$ with $(u,v) \in E$ and
$b \in [k]$.
\end{enumerate}
It is easy to see that this is a local encoding scheme for $k$-COLOR in the sense of Section~\ref{sec:encodings}. Thus,
Corollary~\ref{cor:coloring} applies to it and we get a family of
instances $(G_n)_{n \geq 1}$ that are hard for Polynomial Calculus in the indicated
encoding scheme. Note that, since the instances are hard, they must be
loopless graphs. Indeed, if $(u,u)$ is a loop in $G_n$, then $X(u,b)X(u,b)
= 0$ is an equation in the encoding of the instance $G_n$ for all $b
\in [k]$. These equations, together with the axioms $X(u,b)^2 - X(u,b)
= 0$ and the equation $\sum_{b \in [k]} X(u,b) - 1 = 0$, would give a
PC derivation of $1 = 0$ in degree $2$ and constant size. Thus the
instances in the family are loopless graphs.  We may also assume that
they are undirected graphs for the simple reason that the equations
$X(u,b) X(v,b) = 0$ and $X(v,b) X(u,b) = 0$ are identical (recall that
all our variables commute by assumption).  It follows that
Corollary~\ref{cor:coloring} has the real-field case of Theorem~1.1
from~\cite{DBLP:conf/coco/LauriaN17} as a special case, except for the
fact that, unlike Theorem~1.1 from~\cite{DBLP:conf/coco/LauriaN17},
Corollary~\ref{cor:coloring} does not state that the family of graphs
is explicit. In the next section we show that we can also get an
explicit family of graphs with the same properties.

\subsection{Opening the box}
In the rest of this section we open the box of the method that
underlies Corollary~\ref{cor:coloring}. This will allow us to
re-derive Theorem~1.1 from~\cite{DBLP:conf/coco/LauriaN17} for all
fields, and not just for the real-field as is stated in
Corollary~\ref{cor:coloring}. Moreover, it will suggest a way to apply
the method to any other problem that is NP-complete via \emph{gadget
  reductions}.

Since $\mathbb{K}_k$ for $k \geq 3$ is a template of unbounded width,
Theorem~\ref{thm:bartokozik} applies to it. It is not difficult to see that
$\mathbb{K}_k$ is a core, hence by Theorem~\ref{thm:bartokozik}
there exists a
non-trivial finite Abelian group $G$ such that $\mathbb{B}(G,3)$ is
pp-interpretable in $\mathbb{K}_k^+$, where $\mathbb{K}_k^+$ is the
expansion of $\mathbb{K}_k$ with all constants; i.e., the expansion
with the relations $R_1 = \{1\},\ldots,R_k = \{k\}$. Indeed, this is
the case even for the group $G = \mathbb{Z}_2$. Concrete such
pp-interpretations are well-known and also easy to construct. For the
sake of completeness and by way of example, we propose one such
pp-interpretation in two steps. First we pp-interpret
$\mathbb{B}(\mathbb{Z}_2,3)$ in the template of $3$-SAT, and then we
pp-interpret the template of $3$-SAT in $\mathbb{K}_k^+$. Since
pp-interpretations compose, we get what we want.

Recall that $\mathbb{B}(\mathbb{Z}_2,3)$ has domain $\{0,1\}$ and two
relations $E_0 = \{(b_1,b_2,b_3) \in \{0,1\}^3 : b_1 + b_2 + b_3 = 0
\mod 2 \}$ and $E_1 = \{(b_1,b_2,b_3) \in \{0,1\}^3 : b_1 + b_2 + b_3
= 1 \mod 2 \}$. The template of $3$-SAT also has domain $\{0,1\}$ and
the eight arity-3 relations $R_{000}$, $R_{001}$, $R_{010}$,
$R_{100}$, $R_{011}$, $R_{101}$, $R_{110}$, and $R_{111}$ defined by
the eight possible signings of a 3-clause. A pp-interpretation
of $\mathbb{B}(\mathbb{Z}_2,3)$ is given by the following formulas:
\begin{enumerate} \itemsep=0pt
\item $\delta(x_1) := x_1 = x_1$, { i.e., \textit{true} so the domain
    in still $\{0,1\}$, }
\item $\epsilon(x_1,x_2) := x_1 = x_2$,
\item $E_0(x_1,x_2,x_3) := R_{001}(x_1,x_2,x_3) \wedge R_{010}(x_1,x_2,x_3)
\wedge R_{100}(x_1,x_2,x_3) \wedge R_{111}(x_1,x_2,x_3)$,
\item $E_1(x_1,x_2,x_3) := R_{000}(x_1,x_2,x_3) \wedge R_{011}(x_1,x_2,x_3)
\wedge R_{101}(x_1,x_2,x_3) \wedge R_{110}(x_1,x_2,x_3)$.
\end{enumerate}
Next we give the standard pp-interpretation of the template of
$3$-SAT in $\mathbb{K}_k^+$ when $k \geq 3$. We use the first
two colors $1$ and $2$ to represent $0$ and $1$, respectively:
\begin{enumerate} \itemsep=0pt
\item $\delta(x_1) := \exists y_3 \cdots \exists y_k
  \big(\bigwedge_{b=3}^k R_{b}(y_b) \wedge E(x_1,y_b)\big)$,
{ i.e., the domain is $\{1,2\}$, }
\item $\epsilon(x_1,x_2) := x_1 = x_2$,
\item $R_{abc}(x_1,x_2,x_3) := \exists x'_1\exists x'_2\exists
  x'_3 (
  \bigwedge_{i=1}^3 (\delta(x_i) \wedge \delta(x'_i)
  \wedge E(x_i,x'_i)) \wedge 
    \exists y_1 \cdots \exists y_6 (\bigwedge_{i=1}^6 \exists z_4 \cdots \exists z_k\\\big(\bigwedge_{b=4}^k R_b(z_b) \wedge E(y_i,z_b)\big) \wedge
  E(y_1,y_3) \wedge
  E(y_2,y_3) \wedge E(y_1,y_2) \wedge E(y_3,y_4) \wedge E(y_4,y_6)
  \wedge E(y_5,y_6) \wedge E(y_4,y_5) \wedge R_2(y_6) \wedge
  E(y_1,x^{(a)}_1) \wedge E(y_2,x^{(b)}_2) \wedge E(y_5,x^{(c)}_3)))$,
\end{enumerate}
where $x^{(d)}_i$ is shorthand notation for $x_i$ if $d = 0$, and
$x'_i$ if $d = 1$.  Checking that these interpretations are correct is
a straightforward exercise. Note that, as written, the formulas for
$R_{abc}$ are not quite pp-formulas, but they are easily converted
into pp-formulas by standard rewriting into prenex normal form.  At
this point we can apply Corollary~\ref{col:bdFrege} of Theorem~\ref{thm:closureboundeddepthFrege} in
conjunction with Theorem~\ref{thm:lowerbound} to obtain the statement \emph{2} in Corollary~\ref{cor:coloring}; Theorem~\ref{thm:closurealgebraic} in conjunction with Theorem~\ref{thm:lbpc} to obtain the statement \emph{3}; and Theorem~\ref{thm:closuresemialgebraic} in conjunction with
Theorem~\ref{thm:chan} to obtain the statement \emph{4} in Corollary~\ref{cor:coloring}.

Our next goal is to extend the PC lower bound for $k$-COLOR to all
fields. Before we do so, let us note that exactly the same strategy as in
the previous paragraph is not enough. The reason
is that $3$LIN($\mathbb{Z}_2$) is easy for Polynomial Calculus over fields of
characteristic two. Surely we could start with an instance of
$3$LIN($\mathbb{Z}_3$), which is going to be hard for fields of
characteristic two, but the result is again not going to be hard for
all fields simultaneously as it will fail to be hard for fields of
characteristic three. The solution is to start with a problem that has
instances that are hard for Polynomial Calculus for all fields simultaneously. Luckily,
3-SAT is such a case:

\begin{theorem}[see Theorem 3.13 in \cite{AlekhnovichRazborov2003}] 
  \label{thm:allfields}
  There exists a positive real $\delta$ and an explicit family
  $(G_n)_{n \geq 1}$ of unsatisfiable instances of 3-SAT, where $G_n$ has
  $\Theta(n)$ variables and $\Theta(n)$ clauses, such that, for every
  field $F$ and every sufficiently large $n$, every PC refutation over
  $F$ of $G_n$ with respect to the $\EQ$ encoding scheme has degree at
  least~$\delta n$.
\end{theorem}

\noindent Let us note that in order to get Theorem~\ref{thm:allfields}
from the exact statement of Theorem~3.13 in~\cite{AlekhnovichRazborov2003} one needs explicit families of
3-regular \emph{unique-neighbor} expanders. Such families were
described in~\cite{AlonCapalbo2002}.

With the lower bound of~Theorem~\ref{thm:allfields} in place we can
get the version of the PC lower bound of Corollary~\ref{cor:coloring}
for all fields: the corresponding explicit instances of $k$-COLOR are
obtained by applying the conjunction of Theorem~\ref{thm:allfields}
and Theorem~\ref{thm:closurealgebraic} on the already noted fact that
the template of 3-SAT pp-interprets in $\mathbb{K}^+_k$.  This gives a
new proof of Theorem~1.1 from~\cite{DBLP:conf/coco/LauriaN17}.

Let us point out the main differences and
similarities between the original proof
from~\cite{DBLP:conf/coco/LauriaN17} and our new proof. At a high
level, those proofs are very similar: both are \emph{gadget reductions}
that convert hard CNF formulas into hard instances of $k$-COLOR. 
In our proof, the gadgets are based on the way the template of 3-SAT is constructed from the template of $k$-COLOR by the addition of constants followed by the pp-interpretation (as presented in Section~\ref{sec:ppdefinitions}). Hence, the starting hard formulas can be any family of 3-CNF
formulas that are hard for Polynomial Calculus. 
The proof
from~\cite{DBLP:conf/coco/LauriaN17} is also a gadget reduction, but
in their case the reduction is specifically tailored to a concrete
family of CNF formulas that encode a sparse version of the functional
pigeonhole principle. Besides the construction of the special-purpose
gadgets, the bulk of their proof is to check that the conversion
preserves the hardness for~PC. In our proof both these parts are
handled automatically by our general results.

We close by noting that, for CSPs, this method is completely general.
Take any template~$\mathbb{B}$ that is not known to be solvable in
polynomial time, i.e., any template that is known to be
NP-complete. By the Algebraic Dichotomy Theorem for CSP, any finite
structure~$\mathbb{C}$ pp-interprets in the core of $\Bstruct$ with
added constants (see \cite{10.2307/43556439,Bulatov2017,Zhuk2017}).  In
particular, by taking~$\mathbb{C}$ to be the template of 3-SAT and
applying Theorems~\ref{thm:allfields} and~\ref{thm:closurealgebraic}
we get explicit families of instances of CSP($\mathbb{B}$) that are
hard for Polynomial Calculus over all fields.  The same applies to all proof systems
for which we can prove closure under reductions: explicit lower bounds
for any CSP imply explicit lower bounds for all NP-complete
CSPs. Since explicit lower bounds for $k$LIN($\mathbb{Z}_2$) are known
for Sums-of-Squares~\cite{DBLP:journals/tcs/Grigoriev01,DBLP:conf/focs/Schoenebeck08}, this answers the question in Open Problem~5.3 in
\cite{DBLP:conf/coco/LauriaN17} that asks for explicit hard 3-coloring
instances for~SOS.

\section{Concluding remarks} \label{sec:conclusions}

Theorems~\ref{thm:closureboundeddepthFrege},
\ref{thm:closurealgebraic} and~\ref{thm:closuresemialgebraic} imply
that for the proof systems under consideration the class of constraint
languages admitting efficient refutations can be characterised
algebraically. For most of those proof systems such a characterisation
follows from the fact that efficient proofs of unsatisfiability exist
exactly for languages of bounded width. However, by
Theorem~\ref{thm:maintwoelementfield} the class of constraint
languages admitting efficient refutations in Lov\'asz-Schrijver, and
consequently also the class of constraint languages admitting
efficient Frege refutations, exceeds bounded width. At the same time
both of those classes are shown to admit algebraic
characterisations. Providing such characterisations is a natural open
problem that arises from our work. In particular, with the Algebraic
CSP Dichotomy Conjecture recently
confirmed~\cite{Bulatov2017,Zhuk2017}, it would be interesting to
verify or refute the tempting conjecture that the class of languages
admitting polynomial size Frege (or Extended Frege) refutations
coincides with the class of all polynomial time solvable constraint
languages.

Other proof systems which are shown to be closed under reducibilities
and can surpass bounded width are Polynomial Calculus proof systems
over fields of prime characteristics.  Finding algebraic
characterisations for the classes of constraint languages admitting
efficient unsatisfiability proofs in each of those proof systems is
another question suggested by our work. Importantly, since Polynomial
Calculus over a field of non-zero characteristic $p$ has efficient
refutations for systems of linear equations over $\mathbb{Z}_p$ and
does not have efficient refutations for systems of linear equations
over $\mathbb{Z}_m$ if $p$ does not devide $m$
(c.f.~Theorem~\ref{thm:BussGIP01}), for two fields of distinct prime
characteristics, such characterisations will necessarily be
different. 

Both questions raised so far could lead to the discovery of some
interesting new families of algebras as has happened before in the
development of the algebraic approach to CSPs (c.f., the class of
algebras with few subpowers \cite{Idziaketal}).

A related direction suggested by our work is whether the
proof complexity of approximating MAX CSPs is also preserved by
reductions. On the one hand, it is known that most of the classical CSP constructions
preserve \emph{almost satisfiability}; e.g., if $\Bstruct'$ is
pp-definable without equality in $\Bstruct$, then if $\Astruct$ is an instance of MAX
CSP($\Bstruct')$ that is almost satisfiable, then its standard
transformation into an instance $\Astruct'$ of MAX CSP($\Bstruct$) is
also almost satisfiable. The question we raise is the following: For
which proof systems is it also the case that if there are efficient
proofs of the fact that~$\Astruct'$ is far from satisfiable then there also are
efficient proofs of the fact that~$\Astruct$ is far from satisfiable?  Depending
on how the terms ``almost satisfiable'' and ``far from satisfiable''
are quantified, a positive answer for such questions could lead to an
algebraic approach to the proof complexity of approximating MAX CSPs and
the UGC.

\bigskip \noindent\textbf{Acknowledgments.} We are grateful to Massimo
Lauria for his help in reconstructing the proofs in
Section~\ref{sec:completenessSA}, and for several other insightful
comments during the development of this work. We are also grateful to
Massimo Lauria and Jakob Nordstr\"om for useful discussions on the
applicability of our results to the 3-coloring problem and the
connection to their results in \cite{DBLP:conf/coco/LauriaN17}.  Both
authors were partially funded by European Research Council (ERC) under
the European Union's Horizon 2020 research and innovation programme,
grant agreement ERC-2014-CoG 648276 (AUTAR). First author partially
funded by MINECO through TIN2013-48031-C4-1-P (TASSAT2). Part of this
work was done while the authors were in residence at the Simons
Institute for the Theory of Computing.

\bibliographystyle{plain}
\bibliography{bibfileforthis}

\end{document}